\documentclass[11pt]{elsarticle}

\makeatletter
\def\ps@pprintTitle{%
	\let\@oddhead\@empty
	\let\@evenhead\@empty
	\def\@oddfoot{}%
	\let\@evenfoot\@oddfoot}
\makeatother

\usepackage[left=2.25cm,right=2.25cm,top=2cm,bottom=3cm]{geometry}
\usepackage{empheq} 
\usepackage{algorithm}
\usepackage{algorithmic}
\usepackage{subfigure} 
\usepackage{amsthm}

\usepackage{comment} 
\usepackage{amsfonts,amssymb,amsmath,amsthm,bm}
\numberwithin{equation}{section}
\usepackage{enumerate}
\usepackage{multirow}
\usepackage{xcolor}
\usepackage{graphicx,float,epstopdf}
\usepackage{array}
\usepackage{calc} 
\usepackage{booktabs}
\usepackage{hyperref}
\usepackage{subeqnarray}
\usepackage{cases}
\usepackage{mathtools}

\newtheorem{theorem}{Theorem}[section]
\newtheorem{definition}[theorem]{Definition}
\newtheorem{lemma}[theorem]{Lemma}

\newtheorem{remark}[theorem]{Remark}
\theoremstyle{definition}
\newtheorem{example}{Example}[section]

\begin{document}
	\begin{frontmatter}

		\title{{\bfseries Computing excited states of dipolar Bose--Einstein condensates by constrained gentlest ascent dynamics}}
		
		\author{Fu Zheng}\ead{1469190552@qq.com}
		\author{Wei Liu\corref{cor}}\ead{wl@nudt.edu.cn}
		\author{Xu Qian}\ead{qianxu@nudt.edu.cn}
		
		\address{College of Science, National University of Defense Technology, Changsha 410073, China}
		\cortext[cor]{Corresponding author.}

{\small
\begin{abstract}
Dipolar Bose–Einstein condensates (BECs), i.e., BECs with long-range dipole–dipole interactions (DDI), exhibit extraordinarily rich physics and serve as a versatile experimental platform for exploring exotic quantum states of matter. Within the mean-field framework, a dipolar BEC can be described by the Gross-Pitaevskii (GP) model with a DDI term. Its stationary states consist of the ground state, which minimizes the GP energy functional under the normalization constraint, and excited states, which are higher-energy critical points that typically correspond to constrained saddle points. Excited states encode collective excitation modes and play a key role in the dynamical stability and quantum phase transitions of cold-atom systems. However, in contrast to the extensively studied ground state, they remain far less explored numerically, especially in dipolar systems. In this work, we develop a numerical method based on the constrained gentlest ascent dynamics (CGAD) that systematically computes excited states of dipolar BECs by stably targeting constrained saddle points with any prescribed Morse index. Specifically, a semi-implicit sine-pseudospectral discretization with an orthonormalization correction step is adopted, and the nonlocal DDI terms are evaluated efficiently and accurately via the nonuniform fast Fourier transform (NUFFT). The resulting scheme is applied to one-, two-, and three-dimensional dipolar BECs, enabling a detailed numerical investigation of their excited states. Remarkably, our simulations reveal that excited states can persist even in regimes where the ground state of the condensate becomes mean-field unstable. This work provides a computational framework for studying excited states in dipolar quantum gases and offers insights for future theoretical and experimental investigations.
\end{abstract}


} 

	\end{frontmatter}

\section{Introduction}\label{sec:introduction}
Since the first realization of atomic Bose--Einstein condensates (BECs) in 1995~\cite{anderson1995observation,bradley1995evidence,davis1995bose}, dipolar condensates have occupied a distinctive position among the ultracold quantum gases realized to date. Their structure is dictated not solely by the external trapping potential but also by the interparticle interactions themselves~\cite{yi2001trapped,lahaye2009physics}. Following the condensation of $^{52}$Cr~\cite{griesmaier2005bose,griesmaier2007generation} and, subsequently, of the strongly magnetic lanthanides $^{164}$Dy~\cite{Lu_2011} and $^{168}$Er~\cite{aikawa2012bose}, the preparation of condensates exhibiting long-range dipole--dipole interactions (DDI) has become feasible~\cite{ferrofluid,koch2008stabilization,experiment}. In a BEC with DDI, the interaction energy is determined not only by the local density but also by the spatial arrangement of the density relative to the polarization axis, with the DDI, the contact interaction, and the external trapping potential collectively governing the shape and stability of the condensate. This interplay has been experimentally elucidated through the dipolar deformation of expanding clouds~\cite{Anisotropicexpansion}, the stabilization of a purely dipolar gas against collapse by tuning the trap geometry~\cite{koch2008stabilization}, and the roton-like excitation spectrum~\cite{roton}.

Among the various investigations of dipolar BECs, a central research theme concerns their stationary states, which embody rich physical information~\cite{koberle2009bifurcations,mottonen2005collapse}. Within the mean-field Gross--Pitaevskii (GP) theory, the stationary states of a dipolar BEC are characterized as critical points of the GP energy functional under the mass normalization constraint~\cite{bao2012gross,bao2013mathematical,koberle2009bifurcations}. These states are classified as ground or excited states in the physics literature. The ground state is usually defined as a constrained minimizer, while stationary states with energies larger than the ground-state energy are called excited states. Typically, any constrained saddle point is an excited state; it possesses unstable directions along which the energy decreases locally, and therefore its energy must be higher than the ground state energy. For dipolar BECs, the DDI causes the energy functional to depend on both the local density and a nonlocal integral, thereby rendering the energy landscape significantly more complex than that of systems with isotropic short-range interactions. Consequently, the computation of stationary states of dipolar BECs is particularly challenging, especially for excited states.

Compared with the ground state, excited states have received much less attention, even though numerous methods, including gradient-flow methods (imaginary-time methods)~\cite{CNGF,bao2006efficient,3,sobolev,sobolev2,sobolev3,Riemannian,GFLM} and optimization methods~\cite{min,2017Tang,Regularized,preconditioned}, have been developed for the former. Computational methods for ground states can also yield the ``first excited state" from odd initial data when odd parity is imposed to remove unstable directions~\cite{bao2006efficient,schneider1999numerical,weak}. However, such a symmetry-based strategy is difficult to generalize to the dipolar case, because a general polarization direction and the DDI break this odd symmetry. On the other hand, continuation methods and Newton-type methods provide routes for computing excited states by solving the Euler–Lagrange equation~\cite{Chien2007Adaptive,WeichungWang2011Exploring,eigenfunctions,Boullé2020Deflation}. These methods are capable of finding multiple stationary state solutions, including excited states. However, these methods generally cannot compute excited states of an {\em a priori} type, such as a prescribed Morse index.

This paper focuses primarily on excited states of constrained saddle-point type. A constrained saddle point can be quantified by its Morse index, the number of negative eigenvalues of the projected Hessian on the tangent space of the constraint manifold~\cite{Chang1993,chisd,liu2023constrained}. Actually, a constrained saddle point of index \(k\) possesses \(k\) linearly independent unstable directions, along which the energy decreases locally. The presence of unstable directions renders the computation of constrained saddle points intrinsically difficult, and gradient flow methods often fail to converge to them. In fact, local minimizers could be obtained from a constrained saddle point by descending along its unstable directions. Specifically, a steady state of the gradient flow method is linearly stable if and only if the projected Hessian is positive definite on the tangent space (see Section~\ref{2}). Therefore, computing excited states of constrained saddle-point type relies on specialized algorithms for saddle points with prescribed Morse index.

Among saddle-point search methods~\cite{eignvector,dimer,li2001minimax,weinan2011gentlest}, the gentlest ascent dynamics (GAD) of E and Zhou~\cite{weinan2011gentlest} couples the evolution of the position variable and the direction variable so that the energy ascends along the eigen-direction corresponding to the smallest eigenvalue of the Hessian and descends along all other directions, thereby making an index-1 saddle point a stable equilibrium of the dynamical system. Based on the same principle of GAD, some variants were proposed, e.g., the shrinking dimer dynamics~\cite{SDD}, the iterative minimization formulation~\cite{IMF2015}, and the improved high-index saddle dynamics~\cite{hisd}. However, the computation of excited states of BECs needs to account for the mass normalization constraint. Methods for computing index-1 constrained saddle points include the constrained string method~\cite{string}, the geodesic nudged elastic band method~\cite{skym}, and the constrained shrinking dimer dynamics~\cite{CSDD}. Recently, the idea of GAD has been extended to finding constrained high-index saddle points. Liu, Xie, and Yuan proposed a constrained gentlest ascent dynamics (CGAD) and applied it to compute excited states of single-component BECs with prescribed Morse indices~\cite{liu2023constrained}. Yin, Huang, and Zhang derived the constrained high-index saddle dynamics (CHiSD) to construct the solution landscape with equality constraints~\cite{chisd}, which was also used to reveal excited states of rotating BECs \cite{Revealing}. While these dynamical approaches have shown great potential in computing constrained saddle points and thus hold promise for excited states of dipolar BECs, their application to this setting remains unexplored. To the best of our knowledge, no systematic work has addressed the numerical computation of excited states in dipolar BECs. Apart from the mass normalization constraint, the essential difficulties are twofold: the intrinsic instability of excited states as constrained saddle points of the GP energy functional, and the nonlocal nature of the DDI, which not only significantly increases the computational cost but also poses challenges for the stability and robustness of numerical algorithms.

In this paper, we propose a CGAD-based spectral method for computing excited states of dipolar BECs. We first present a continuous dynamical system based on the CGAD to find dipolar excited states by stably locating constrained saddle points of the GP energy functional with any prescribed Morse index. The steady state of the presented continuous dynamical system corresponds exactly to constrained critical points of the GP energy functional together with eigenfunctions of the projected Hessian. Moreover, a steady state becomes linearly stable if and only if the wave function is an index-\(k\) constrained saddle point and the search directions are the eigenfunctions associated with \(k\) negative eigenvalues. Then, for the numerical realization, the nonlocality of the DDI complicates the numerical discretization. Specifically, in two and three dimensions, since the dipole kernel of the DDI is nonsmooth/singular at the origin in Fourier space, the standard Fourier/sine spectral discretization of the DDI terms loses spectral accuracy and suffers from the so-called ``numerical locking''~\cite{bao2016accurate,jiang2014fast}. To overcome this difficulty, following~\cite{bao2016accurate,jiang2014fast}, we adopt a coordinate transformation from Cartesian to polar/spherical coordinates in Fourier space, whose Jacobian determinant exactly cancels the nonsmoothness/singularity of the dipole kernel, rendering the integrand smooth near the origin. Together with high-order quadrature and the nonuniform fast Fourier transform (NUFFT), this enables the high-accuracy and high-efficiency evaluation of the DDI terms. We further integrate this technique with the sine-pseudospectral method in space and a semi-implicit time discretization with an orthonormalization correction, yielding efficient and accurate numerical algorithms for computing excited states of dipolar BECs. Based on the proposed numerical method, we conduct one-, two-, and three-dimensional numerical experiments to investigate the effects of the contact interaction, the strength and sign of the DDI, and the polarization direction on the excited states. Notably, in certain regimes, excited states exist over a broader parameter range than the ground state.

The remainder of this paper is organized as follows. Section~\ref{2} introduces the GP model for dipolar BECs and their constrained variational structure. Section~\ref{3} presents the continuous dynamical formulation and its properties. Section~\ref{sec:scheme} describes a semi-implicit sine-pseudospectral scheme and the numerical evaluation of the nonlocal interactions. Numerical results are reported in Section~\ref{results}, and conclusions are drawn in Section~\ref{conclusion}.

\section{The excited states of dipolar BECs}\label{2}
\subsection{GP energy functional for dipolar BECs}
Within the framework of mean-field theory, the GP energy functional for a single-component dipolar BEC in 3D space, after proper nondimensionalization, takes the following form \cite{bao2013mathematical,bao2010efficient,bao2016accurate,trapped}:
	\begin{equation}
		E_{3D}(\phi) =
        \int_{\mathbb{R}^3} \left[ \frac{1}{2}|\nabla \phi|^2 + V_3(\mathbf{x})|\phi|^2 + \frac{\beta_3}{2}|\phi|^4\right] d\mathbf{x}+ \frac{\lambda_3}{2}\int_{\mathbb{R}^3}\int_{\mathbb{R}^3}U_{\mathrm{dip}}^{(3)}(\mathbf{x}-\mathbf{x}') |\phi(\mathbf{x})|^2|\phi(\mathbf{x}')|^2d\mathbf{x}'d\mathbf{x},
		\label{eq:energy_3D}
	\end{equation}
where the dimensionless parameters $\beta_3$ and $\lambda_3$ denote the strengths of the contact and dipolar interactions in 3D, respectively, and the real-valued function $V_3(\mathbf{x})$ represents the external trapping potential. The 3D dipole kernel $U_{\mathrm{dip}}^{{(3)}}$ is defined as 
	\begin{equation}
		U_{\mathrm{dip}}^{{(3)}}(\mathbf{x}) 
        =\frac{3}{4\pi|\mathbf{x}|^3}\left(1-\frac{3(\mathbf{x}\cdot\mathbf{n})^2}{|\mathbf{x}|^2}\right)
        = -\delta(\mathbf{x}) - 3\partial_{\mathbf{nn}}\left(\frac{1}{4\pi|\mathbf{x}|}\right),\label{ker3d}
	\end{equation}
where $\mathbf{n}=(n_1,n_2,n_3)^T$ is a given unit vector representing the dipole axis, and $\partial_{\mathbf{nn}} = \mathbf{n} \cdot \nabla (\mathbf{n} \cdot \nabla)$.

Lower-dimensional dipolar BECs are typically realized in physical experiments by imposing strong harmonic confinement along certain directions. The 3D model can then be effectively reduced to 2D or 1D effective models when the condensate is strongly trapped in one or two directions, respectively. Actually, by a standard dimension reduction, i.e., by freezing the wave function to the harmonic-oscillator ground state (or Gaussian) in those strongly trapped directions and integrating out the corresponding coordinates, the effective interaction parameters scale with a confinement length factor $\varepsilon>0$ as
	\begin{equation}
		\beta_2 = \frac{\beta_3}{\sqrt{2\pi}\varepsilon},\; 
        \lambda_2 = \frac{\lambda_3}{\sqrt{2\pi}\varepsilon}\;\; \mbox{in 2D, }\; \mbox{and}\quad 
        \beta_1 = \frac{\beta_3}{2\pi\varepsilon^2},\; 
        \lambda_1 = \frac{\lambda_3}{2\pi\varepsilon^2}\;\; \mbox{in 1D}.
		\label{eq:beta_scaling}
	\end{equation}
The corresponding effective dipole kernels $U_{\mathrm{dip}}^{(2)}$ and $U_{\mathrm{dip}}^{(1)}$ are derived by integrating the 3D dipole kernel over the corresponding strongly trapped directions, resulting in
	\begin{equation}
		U_{\mathrm{dip}}^{(d)}(\mathbf{x}) = 
		\begin{cases} 
			\displaystyle (3n_3^2-1)\delta(\mathbf{x}) - \frac{3\sqrt{2\pi}\varepsilon}{2}(\partial_{\mathbf{n}_\perp\mathbf{n}_\perp}-n_3^2\Delta)\left(\frac{1}{2\pi|\mathbf{x}|}\right), & d=2, \\[10pt]
			\displaystyle \frac{1-3n_3^2}{2}\delta(x) - \frac{3(3n_3^2-1)}{4}\partial_{xx}\left(e^{x^2/(2\varepsilon^2)}\int_{|x|}^{\infty} e^{-s^2/(2\varepsilon^2)}ds\right), & d=1,
		\end{cases}
		\label{eq:U_dip_lower_d}
	\end{equation}
where $\mathbf{x}=(x,y)$ and $\mathbf{n}_\perp=(n_1,n_2)^T$ denotes the projection of the dipole axis onto the $x$-$y$ plane, and $\Delta$ is the 2D Laplacian. For $U_{\mathrm{dip}}^{(1)}$ in \eqref{eq:U_dip_lower_d}, the original 3D system is strongly confined in $x$ and $y$ directions, and the remaining variable $z$ is relabeled as $x$. Then the $d$-dimensional ($d=1, 2,3$) GP energy functionals share a unified form: 
   \begin{equation}
		E(\phi) = \int_{\mathbb{R}^d} \left[ \frac{1}{2}|\nabla \phi|^2 + V_d(\mathbf{x})|\phi|^2 + \frac{\beta_d}{2}|\phi|^4 + \frac{\lambda_d}{2}\Phi_{{\mathrm{dip}}}(|\phi|^2)|\phi|^2 \right] d\mathbf{x}.
		\label{eq:energy}
	\end{equation}
	Hereafter, the nonlocal operator $\Phi_{\mathrm{dip}}$ that describes the long-range dipole interaction is defined by a convolution 
    \begin{equation} \label{duichen}
        \Phi_{\mathrm{dip}}(f)(\mathbf{x}) := \big(U_{\mathrm{dip}}^{(d)} * f\big)(\mathbf{x})=\int_{\mathbb{R}^d}U_{\mathrm{dip}}^{(d)}(\mathbf{x}-\mathbf{x}') f(\mathbf{x}')d\mathbf{x}'
    \end{equation}
    for a given function $f(\mathbf{x})$, 
    where \(U_{\mathrm{dip}}^{(d)}\) is defined in (\ref{ker3d}) and (\ref{eq:U_dip_lower_d}). 

\begin{remark}
    The formulation of the 2D/1D kernel $U_{{\mathrm{dip}}}^{(d)}$ in \eqref{eq:U_dip_lower_d} is slightly different from the ones in \cite{bao2012gross,bao2013mathematical}. In fact, to achieve a clean decoupling of $\beta_d$ and $\lambda_d$ ($d=1,2$) as in \eqref{eq:beta_scaling} while maintaining consistent orders of magnitude, and to facilitate the investigation of dipole-dominated regimes on ground and excited states, the local contact terms (i.e., terms involving the Dirac $\delta$-function) arising from the dimension reduction or singularity separation are absorbed into the kernel $U_{{\mathrm{dip}}}^{(d)}$ in \eqref{eq:U_dip_lower_d}. 
\end{remark}

Throughout this paper, we assume that the wave function $\phi$ is real-valued. Assume that the potential function satisfies \(V_d(\mathbf{x}) \geq 0\) for all \(\mathbf{x} \in \mathbb{R}^d\) and $\lim_{|\mathbf{x}|\to\infty}V_d(\mathbf{x})=\infty$. We adopt the standard notation for Sobolev spaces, and denote the \(L^p(\mathbb{R}^d)\) norm by \(\|f\|_{L^p}\), i.e., \(\|f\|_{L^p}^p := \int_{\mathbb{R}^d} |f(\mathbf{x})|^p \, d\mathbf{x}\) for \(p \in [1, \infty)\), with the spatial domain \(\mathbb{R}^d\) omitted when the context is clear. 
Define the function space \(X\) as
	\[
	X = \left\{ u \in H^1(\mathbb{R}^d) \,\Big|\, \|u\|_{X}^2 := \|u\|_{2}^2 + \|\nabla u\|_{2}^2 + \int_{\mathbb{R}^d} V_d(\mathbf{x})|u(\mathbf{x})|^2 \, d\mathbf{x} < \infty \right\}.
	\]
It is noted that, by the symmetry of the convolution, the identity  
    \begin{equation}\label{eq:Phi-dip-symm}
        \int_{\mathbb{R}^d}\Phi_{\mathrm{dip}}(f)(\mathbf{x})g(\mathbf{x})d\mathbf{x}=\int_{\mathbb{R}^d}\Phi_{\mathrm{dip}}(g)(\mathbf{x})f(\mathbf{x})d\mathbf{x}
    \end{equation}
    holds for all functions \( f \) and \( g \) for which the integrals on both sides make sense. Moreover, $\Phi_{\mathrm{dip}}$ is a bounded linear operator from \(W^{1,p}(\mathbb{R}^d)\) to \(L^p(\mathbb{R}^d)\), \(p\in(1,\infty)\), as stated in the following lemma. 
    \begin{lemma}\label{lem:Phi-dip-boundedness}
		For $d\in\{1,2,3\}$ and \(p\in(1,\infty)\), there exists a constant $C>0$ depending only on \(d\), \(p\), and the {system parameters} such that, for every \(f \in W^{1,p}(\mathbb{R}^d)\),
		\begin{equation}
			\big\| \Phi_{\mathrm{dip}}(f)\big\|_{p} = \big\| U_{\mathrm{dip}}^{(d)} * f\,\big\|_{p} \le C
            \begin{cases}
                \|f\|_{{p}}, & d=1 \mbox{ or } d=3,\\
                \|f\|_{{p}}+\|\nabla f\|_{{p}}, & d=2.
            \end{cases}
		\end{equation}
	\end{lemma}
    \begin{proof}
        The proof is based on the results in \cite{bao2012gross,CMS2008Nonlinearity}, and is presented in \ref{appendix:proof-lem:Phi-dip-boundedness} for brevity.
     \end{proof}
Based on Lemma \ref{lem:Phi-dip-boundedness}, direct calculations show that the GP energy functional $E$ given in \eqref{eq:energy} is second-order continuously Fr\'{e}chet-differentiable, and in the weak sense,
\begin{align}
		E'(\phi) &= -\Delta\phi +2 V_d(\mathbf{x})\phi + 2\beta_d |\phi|^2\phi + 2\lambda_d \Phi_{\mathrm{dip}}(|\phi|^2)\phi, \label{eq:E'} \\
		E''(\phi)v &= -\Delta v + 2\Big( V_d(\mathbf{x}) + 3\beta_d|\phi|^2 + \lambda_d\Phi_{\mathrm{dip}}(|\phi|^2) \Big)v + 4\lambda_d  \Phi_{\mathrm{dip}}(\phi v)\phi,\quad \forall v\in X.
		\label{eq:E''}
\end{align}

Due to the confining potential, the macroscopic wave function decays rapidly in the far field. Hence, for numerical computation, the whole-space problem can be truncated to a sufficiently large bounded domain \(\mathcal{D}\subset\mathbb{R}^d\), and homogeneous Dirichlet boundary conditions are imposed on the wave function. We therefore take
\(X=H_0^1(\mathcal{D})\)
in the remainder of the paper.

This truncation does not replace the free-space dipolar interaction with a bounded-domain interaction. Assume that \(f\) is defined on \(\mathcal{D}\) and satisfies the homogeneous Dirichlet boundary condition. Let \(\mathcal{E}_0 f\) denote its zero extension to \(\mathbb{R}^d\). By \eqref{duichen}, we can define the long-range dipolar interaction on \(\mathcal{D}\) as
\begin{equation}   \label{zero_padd}
\Phi_{\rm dip}^{\mathcal{D}}(f):=\left.\left(U_{\rm dip}^{(d)}*\mathcal{E}_0 f\right)\right|_{\mathcal{D}}=\int_ \mathcal{D}U_{\rm dip}^{(d)}(\mathbf{x}-\mathbf{y}) f(\mathbf{y})\,d\mathbf{y},
\end{equation}
where the singular convolution is understood in the distributional sense (or equivalently, in the Fourier-multiplier sense). Therefore, the homogeneous Dirichlet condition is imposed only on the wave function; no artificial boundary condition is imposed on the dipolar potential. For notational simplicity, we write \(\Phi_{\rm dip}\) instead of \(\Phi_{\rm dip}^{\mathcal{D}}\) below. By \eqref{eq:Phi-dip-symm}, we have
\(
\int_\mathcal{D} \Phi_{\rm dip}(f) g\,d\mathbf{x}
=
\int_\mathcal{D} f \Phi_{\rm dip}(g)\,d\mathbf{x}
\).

By zero extension, the truncated energy functional remains twice continuously Fr\'echet differentiable on the corresponding admissible space. This free-space treatment of the dipolar interaction also motivates the numerical evaluation of the DDI term, see Section~\ref{section:numicel sheme} for details.

For convenience, we introduce the notations
    \begin{align}
        G(\phi)&:=\frac12E'(\phi)=-\frac{1}{2}\Delta\phi + V_d\,\phi + \beta_d |\phi|^2\phi + \lambda_d \Phi_{\mathrm{dip}}(|\phi|^2)\phi, \label{eq:Gphi-expr}\\
        L_{\phi}v&:=\frac12E''(\phi)v= -\frac12\Delta v + V_d\,v + 3\beta_d|\phi|^2v + \lambda_d\Phi_{\mathrm{dip}}(|\phi|^2)v + 2\lambda_d \Phi_{\mathrm{dip}}(\phi v)\phi, \quad \forall v\in X. \label{eq:Lphi-expr}
    \end{align}
    It is noted that $L_{\phi}$ is self-adjoint and is exactly the linearized operator of $G(\phi)$, i.e., $ L_{\phi}v=G'(\phi)v$.

	\subsection{Constrained saddle points and excited states of dipolar BECs}
	The stationary states of a dipolar BEC are characterized by the critical points of the GP energy functional $E(\phi)$ \eqref{eq:energy} subject to the normalization constraint
	\begin{equation}
		\|\phi\|^2 := \int_{\mathcal{D}} |\phi(\mathbf{x})|^2 \, \mathrm{d}\mathbf{x} = 1.
		\label{yueshu}
	\end{equation}
   This amounts to finding the solution pair $(\phi, \mu)$ to the Euler--Lagrange equation $E'(\phi) = 2\mu \phi$, i.e., 
	\begin{equation}
		\mu \phi =G(\phi)= \left[ -\frac{1}{2}\Delta + V_d + \beta_d |\phi|^2 + \lambda_d \Phi_{\mathrm{dip}}(|\phi|^2) \right] \phi,
		\label{eq:time_independent_GPE}
	\end{equation}
    under the constraint \eqref{yueshu}, where \eqref{eq:time_independent_GPE} is also called the time-independent GPE \cite{bao2013mathematical,liu2023constrained}. 
	Due to the physical normalization constraint \eqref{yueshu}, this system intrinsically forms a constrained nonlinear eigenvalue problem, where $\phi$ acts as the eigenfunction and the Lagrange multiplier $\mu$ is the corresponding eigenvalue, physically representing the chemical potential, given as
	\begin{align}
		\mu(\phi) 
        &= \langle G(\phi),\phi\rangle
        =E(\phi) + \int_{\mathcal{D}} \left[ \frac{\beta_d}{2}|\phi|^4 + \frac{\lambda_d}{2}\Phi_{\mathrm{dip}}(|\phi|^2) |\phi|^2 \right] \mathrm{d}\mathbf{x}.\label{eq:mu}
	\end{align}
	Thus, all eigenfunctions of \eqref{yueshu}-\eqref{eq:time_independent_GPE} are the constrained critical points of the GP energy functional $E(\phi)$ \eqref{eq:energy} on the unit spherical manifold
    \begin{align}
        \mathcal{S}=\{\phi\in X:\|\phi\|^{2}=1\}.
    \end{align}
	The ground state and the excited states inherently both satisfy the Euler--Lagrange equation \eqref{eq:time_independent_GPE}. Specifically, the ground state is the constrained minimizer of $E(\phi)$ on $\mathcal{S}$. In contrast, constrained saddle points possess higher energy than the ground state and are thus identified as excited states. 
    
    For $\phi\in\mathcal{S}$, denote the tangent space of $\mathcal{S}$ at $\phi$ as
    $$T_{\phi}\mathcal{S}=\{v\in X:\langle\phi,v\rangle=0\},$$
	where \(\langle \cdot, \cdot \rangle\) is standard \(L^2(\mathcal{D})\) inner product notation, i.e., \(\langle f, g \rangle = \int_{\mathcal{D}} f(\mathbf{x}) g(\mathbf{x}) \,\mathrm{d}\mathbf{x}\). Let $P_{\phi}$ be the $L^2$-orthogonal projection operator mapping from $X$ onto the tangent space $T_{\phi}\mathcal{S}$, expressed as $P_{\phi}v=v-\langle \phi, v \rangle\phi$, $\forall v\in X$. The $L^2$ projected gradient of the GP energy functional, defined with a convenient factor of $1/2$, is then given by
    \begin{align}\label{eq:Fphi-def}
        F(\phi):=\frac12 P_\phi E'(\phi)
        = G(\phi)-\mu(\phi)\phi
        = \left[-\frac{1}{2}\Delta + V_d + \beta_d |\phi|^2 + \lambda_d \Phi_{\mathrm{dip}}(|\phi|^2)-\mu(\phi)\right]\phi.
    \end{align}
    Further, we denote by $H(\phi):=\frac12E''(\phi)-\mu(\phi)I=L_{\phi}-\mu(\phi)I$ the effective Hessian operator with respect to the normalization constraint \eqref{yueshu}, and then define the projected Hessian operator as $\hat{H}(\phi)=P_{\phi}H(\phi)P_{\phi}:\, T_{\phi}\mathcal{S}\to T_{\phi}\mathcal{S}$. More explicitly, for all $v\in T_{\phi}\mathcal{S}$, 
    \begin{align}\label{eq:ProjHessian}
        \hat{H}(\phi)v
        &=P_{\phi}H(\phi)P_{\phi}v=L_{\phi}v-\mu(\phi)v-\left\langle L_{\phi}v,\phi\right\rangle \phi, \nonumber \\
        &=-\frac12\Delta v + V_d\,v + 3\beta_d|\phi|^2v + \lambda_d\Phi_{\mathrm{dip}}(|\phi|^2)v + 2\lambda_d \Phi_{\mathrm{dip}}(\phi v)\phi-\mu(\phi)v \nonumber\\
        &\quad\; -\phi\int_{\mathcal{D}} \bigg(
	\frac{1}{2}\nabla\phi \cdot \nabla v
	+ \big[V_d 
	+ 3\beta_d |\phi|^2 
	+ 3\lambda_d\Phi_{\mathrm{dip}}(|\phi|^2)\big]\phi v
	\bigg)\,\mathrm{d}\mathbf{x}.
    \end{align} 
    The local stability/instability of a constrained critical point can be depicted by examining the spectrum of $\hat{H}(\phi)$ at that point, which introduces the following concept of Morse index \cite{Chang1993,liu2023constrained}.
    
    \begin{definition}
        Let $\phi^*\in\mathcal{S}$ be a constrained critical point of the GP energy functional. Assume that the linear self-adjoint operator $\hat{H}(\phi^*)$ admits a spectral decomposition on the tangent space of the form  $T_{\phi^*}\mathcal{S} = T^{-} \oplus T^{0} \oplus T^{+}$, $\dim(T^{0}) < \infty$, with $T^{-}$, $T^{0}$, and $T^{+}$ denoting the maximal negative-definite, null, and maximal positive-definite subspaces of $\hat{H}(\phi^*)$, respectively. The Morse index of $\phi^*$ is defined as
    \[ \operatorname{Index}(\phi^*) = \dim(T^{-}). \] 
    The subspace $T^{-}$ is called the unstable tangent subspace of $\phi^*$, and every eigenfunction $v\in T^{-}$ is called an unstable tangent direction at $\phi^*$.
    \end{definition}
    
    When \(\operatorname{Index}(\phi^*) = k\) \((k = 1, 2, \dots)\), \(\phi^*\) is a constrained saddle point, admitting \(k\) linearly independent tangent directions along which the energy decreases locally. Thus, all constrained saddle points have higher energy than the ground state and must be excited states. An index-\(k\) constrained saddle point is specifically referred to as an {\em index-\(k\) excited state}. 

\begin{remark}
To provide an intuitive understanding of the Morse index, it is instructive to consider the linear non-interacting limit ($\beta_d = 0, \lambda_d = 0$). In this regime, the Euler--Lagrange equation reduces to a standard linear Schr\"odinger eigenvalue problem:
\begin{equation}
	-\frac{1}{2}\Delta \phi+ V_d\,\phi = \mu \phi,\quad
    \|\phi\|^2=1. \label{eq:linear}
\end{equation}
As stated in Section~5.1 in \cite{liu2023constrained}, under certain assumptions on $V_d$, for an eigenstate $\phi_k$ with eigenvalue $\mu_k$ (ordered as $\mu_0 < \mu_1 \le \mu_2\le \dots$), the maximal negative-definite subspace $T^{-}$ of the projected Hessian $\hat{H}(\phi_k)$ is spanned by the $k$ lower-energy eigenfunctions $\{\phi_0, \phi_1, \dots, \phi_{k-1}\}$.
Perturbations along these directions reduce the energy, identifying them as the unstable tangent directions, and $\operatorname{Index}(\phi_k) = k$. Thus the index-$k$ excited state corresponds to the $k$-th eigenstate of the linear problem, and the Morse index determines the energy hierarchy.
\end{remark}

We emphasize that, due to the local instability of index-$k$ excited states ($k\geq1$) on the energy landscape, the normalized gradient flow (NGF), also known as the imaginary-time evolution method,
\begin{align}\label{eq:cngf}
    \partial_t\phi =-F(\phi)
    =\left[\frac{1}{2}\Delta - V_d - \beta_d |\phi|^2 - \lambda_d \Phi_{\mathrm{dip}}(|\phi|^2)+\mu(\phi)\right]\phi\quad \mbox{with}\;\; \phi(\cdot,0)=\phi^0\in \mathcal{S},
\end{align}
cannot capture them stably, as stated in the following theorem.

\begin{theorem}\label{thm:cngf-stability}
   Let $\phi^*\in\mathcal{S}$ be a steady state of the normalized gradient flow \eqref{eq:cngf}. Then $\phi^*$ is a constrained critical point of the GP energy functional, and $\phi^*$ is linearly stable if and only if $\hat{H}(\phi^*)$ is positive-definite on the tangent space $T_{\phi^*}\mathcal{S}$.
\end{theorem}
\begin{proof}
    The steady state condition $\partial_t\phi|_{\phi=\phi^*}=0$ gives precisely the Euler--Lagrange equation \eqref{eq:time_independent_GPE} at $\phi^*$, which confirms immediately that $\phi^*$ is a constrained critical point. For the linear stability, we consider the Jacobian operator (i.e., linearized operator) of the right-hand-side of \eqref{eq:cngf} at $\phi^*$. It takes the following expression for all $v\in T_{\phi^*}S$:
    \begin{align}\label{eq:dFphistar}
       J^*v &= -F'(\phi^*)v
    =\frac12\Delta v - V_d\,v - 3\beta_d|\phi^*|^2v - \lambda_d\Phi_{\mathrm{dip}}(|\phi^*|^2)v - 2\lambda_d \Phi_{\mathrm{dip}}(\phi^* v)\phi^*\!+\! \mu(\phi^*)v\!+\!\langle \mu'(\phi^*),v\rangle \phi^* \nonumber\\
    &=-L_{\phi^*}v + \mu(\phi^*)v +\langle L_{\phi^*}v,\phi^*\rangle \phi^*=-\hat{H}(\phi^*)v.
    \end{align}
    Therefore, $\phi^*$ is linearly stable if and only if all eigenvalues of Jacobian $J^*$ are negative, equivalently, $\hat{H}(\phi^*)$ is positive-definite on the tangent space $T_{\phi^*}\mathcal{S}$.
\end{proof}

    In the next section, we describe a continuous dynamical method based on the CGAD to stably compute index-$k$ excited states of dipolar BECs with arbitrarily given $k=1,2,\ldots$.

	\section{Constrained gentlest ascent dynamics for finding dipolar excited states} \label{3}

Index-\(k\) excited states (\(k\geq1\)) correspond to constrained saddle points of the GP energy functional on the spherical manifold \(\mathcal{S}\). At such saddle points, the projected Hessian operator \(\hat{H}(\phi^*)\) is no longer positive definite; instead, it possesses exactly \(k\) strictly negative eigenvalues. The eigenfunctions corresponding to these negative eigenvalues are the unstable directions. If one directly evolves the system using the standard NGF, according to Theorem~\ref{thm:cngf-stability}, the gradient flow cannot converge stably to a constrained saddle point. Therefore, in order to capture such constrained saddle points in a stable way, the dynamical mechanism of the standard NGF must be fundamentally modified. A natural strategy is to selectively stabilize \(k\) unstable directions.

A closely related realization of this idea is the GAD proposed by E and Zhou \cite{weinan2011gentlest} for locating index-1 unconstrained saddle points. The central idea of GAD is to formulate a coupled dynamical system in which the position variable ascends along the steepest ascent direction within a designated unstable mode and descends along the steepest descent directions in the orthogonal complement of that mode. Meanwhile, the unstable direction itself evolves via a gradient flow that minimizes the Rayleigh quotient, so as to approximate the eigenfunction associated with the smallest (negative) eigenvalue of the Hessian matrix. For dipolar BECs, however, the evolution dynamics must satisfy the normalization constraint (\ref{yueshu}). Following \cite{liu2023constrained}, we extend the idea of GAD to the constrained and high-index setting for dipolar excited states within the GP model.

\subsection{Dynamics for index-1 excited states}
We start with the case of index-1 excited states. Due to the constraints, the constructed dynamics must lie in the tangent space of the constraint sphere. Inspired by NGF and GAD, we first project the negative energy gradient onto the tangent space $T_\phi\mathcal{S}$, and then decompose the projected gradient in the tangent space into a component parallel to the unstable direction $v$ and a component perpendicular to $v$ (see Fig.~\ref{fig:CGAD_schematic}). Subsequently, we reverse the sign of the parallel component while keeping the perpendicular component unchanged, so that the energy ascends along the unstable direction and descends in all other directions, yielding the dynamics of $\phi$ as
\begin{equation}\label{eq:CGAD1-phi-usingF}
    \partial_t\phi= -F(\phi) + 2\langle F(\phi),v\rangle v,
\end{equation}
where $F(\phi)=P_\phi G(\phi) = G(\phi)-\mu(\phi)\phi$ and $v$ serves as an auxiliary variable that dynamically tracks the unstable tangent direction at $\phi$. Since $\phi$ approximates an index-1 excited state, its unstable direction $v$ is taken as the eigenfunction corresponding to the unique negative eigenvalue of the projected Hessian $\hat{H}(\phi)$. According to the Rayleigh--Ritz variational principle, the dynamics that evolves $v$ can be derived from the following minimization problem
\begin{equation}\label{eq:RRmin-idx1}
    \min_{v} \big\langle \hat{H}(\phi)v, v \big\rangle  \quad \text{subject to} \quad  v \in T_{\phi} \mathcal{S}, \quad \|v\|^2=1.
\end{equation}
A standard gradient flow for solving the minimization problem \eqref{eq:RRmin-idx1} gives the following dynamical evolutionary equation 
\begin{equation}\label{eq:CGAD1-v-usingHhat}
    \partial_t{v} 
    = -\hat{H}(\phi)v  + \tilde{\sigma} \phi+ \tilde{\kappa} v,
\end{equation}
with Lagrange multipliers \(\tilde{\sigma}\) and \(\tilde{\kappa}\) to enforce the constraints \(\langle \phi, v \rangle=0\) and \(\|v\|^2=1\), respectively. Taking the time derivative of these constraints gives  
 $\langle \partial_t{\phi}, v\rangle + \langle \partial_t{v},\phi \rangle = 0$ and
 $\langle \partial_t{v}, v \rangle = 0$.
Combining \eqref{eq:CGAD1-phi-usingF} and \eqref{eq:CGAD1-v-usingHhat}, we have
\begin{equation}\label{eqs_for_multiplier}
\begin{pmatrix}
\|\phi\|^2 & \langle\phi,v\rangle\\
\langle\phi,v\rangle & \|v\|^2
\end{pmatrix}
\begin{pmatrix}
\tilde{\sigma}\\
\tilde{\kappa}
\end{pmatrix}
=
\begin{pmatrix}
\langle\hat{H}(\phi)v,\phi\rangle
+(1-2\|v\|^2)\langle F(\phi),v\rangle\\
\langle\hat{H}(\phi)v,v\rangle
\end{pmatrix}.
\end{equation}
Assuming the constraints  hold and substituting $\|\phi\|^2=1$, $\langle \phi, v \rangle=0$ and $\|v\|^2=1$ in \eqref{eqs_for_multiplier}, the expressions of Lagrange multipliers are obtained as
\(
\tilde\sigma
=-\langle F(\phi),v\rangle
\) and
\(
\tilde{\kappa }=\langle \hat{H}\left( \phi  \right)v,v \rangle 
\).

\begin{figure}[!t]
		\centering
		\begin{minipage}{1\textwidth}
			\centering
			\includegraphics[width=0.8\linewidth]{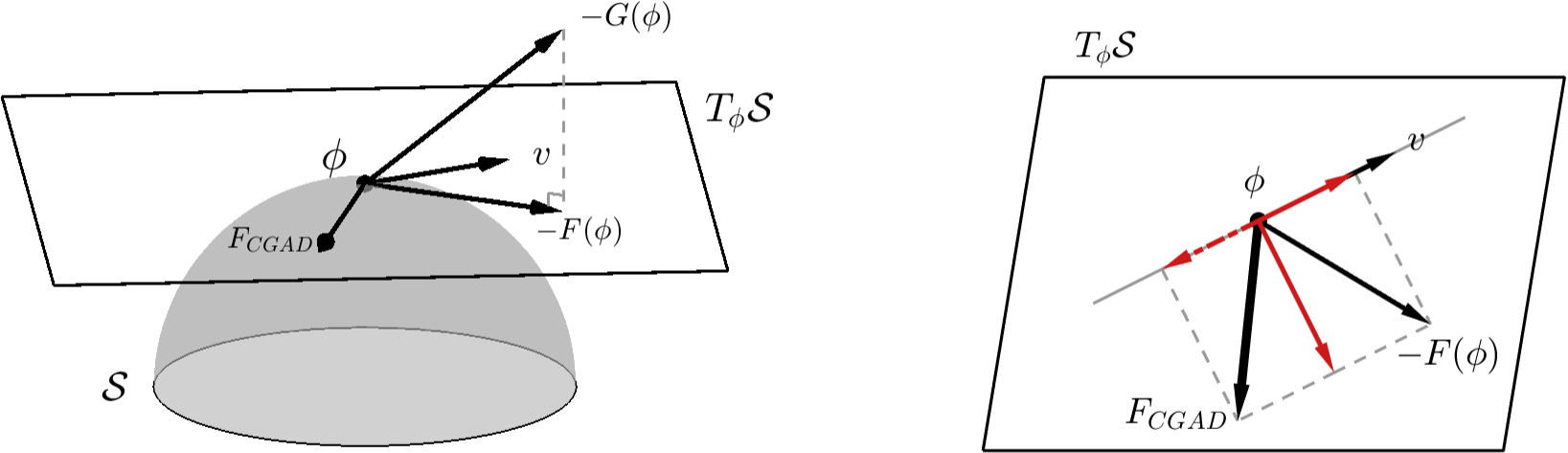}
		\end{minipage}
		\caption{Illustration of the index-1 CGAD. Left: Projection of the negative energy gradient $-G(\phi)$ onto the tangent space $T_{\phi}\mathcal{S}$ to obtain the constrained negative gradient $-F(\phi)$. Right: Decomposition of $-F(\phi)$ into parallel and perpendicular components relative to the unstable direction $v$ to form the driving force $F_{\mathrm{CGAD}}=-F(\phi) + 2\langle F(\phi),v\rangle v$.}
		\label{fig:CGAD_schematic}
	\end{figure}

\subsection{Dynamics for general index-$k$ excited states}
For general index-$k$ excited states ($k\geq 1$), we introduce $k$ orthonormal tangent directions $\{v_i\}_{i=1}^k \subset T_\phi\mathcal{S}$, which approximate eigenfunctions of the projected Hessian operator corresponding to $k$ negative eigenvalues. These directions together with $\phi$ must satisfy the following orthonormality constraints
\begin{equation}
    \|\phi\|^2 = 1, \quad \langle \phi, v_i \rangle = 0, \quad 
    \langle v_i, v_j \rangle = \delta_{ij}, \quad 1 \le j \le i \le k.
    \label{eq:constraints}
\end{equation}
Extending \eqref{eq:CGAD1-phi-usingF} directly to the case of $k$ unstable directions $\{v_i\}_{i=1}^k\subset T_\phi\mathcal{S}$ yields the dynamics
\begin{equation}\label{eq:CGADk-phi-usingG}
    \partial_t \phi = -G(\phi) +\mu(\phi)\phi+ 2\sum_{j=1}^{k} \langle F(\phi), v_j \rangle v_j.
\end{equation}
To drive the dynamics of $v_i$, we consider the ansatz that $v_1, v_2, \dots, v_{i-1}\in T_{\phi}\mathcal{S}$ are exactly the eigenfunctions of $\hat{H}(\phi)$ corresponding to the first $i-1$ smallest eigenvalues. Then the eigenfunction corresponding to the $i$-th smallest eigenvalue can be characterized by the Rayleigh--Ritz minimization problem 
\begin{equation}\label{eq:RRmin-idxk}
    \min_{v_i} \big\langle \hat{H}(\phi)v_i, v_i \big\rangle \quad \text{subject to} \quad  v_i \in T_{\phi} \mathcal{S}, \quad \langle v_i, v_j \rangle = \delta_{ij} \; (j = 1, \dots, i).
\end{equation}
Similar to \eqref{eq:CGAD1-v-usingHhat}--\eqref{eqs_for_multiplier}, the gradient flow for solving \eqref{eq:RRmin-idxk} yields the dynamics of \(v_i\) as follows
\begin{align}\label{eq:CGADk-vi-usingLphi}
    \partial_t v_i
    = -\hat{H}(\phi)v_i
    + \tilde{\sigma}_i\phi
    + \sum_{j=1}^{i}\tilde{\kappa}_{ij}v_j,
    \quad i=1,2,\ldots,k,
\end{align}
where $\tilde{\sigma_i}$ and $\tilde{\kappa}_{ij}$ are Lagrange multipliers to enforce the constraints \(\langle \phi, v_i \rangle=0\) and \(\langle v_i, v_j \rangle = \delta_{ij} \; (j = 1, \dots, i)\), respectively. Taking the time derivative of these constraints gives
\(
\langle \partial_t\phi, v_i\rangle+\langle \partial_t v_i,\phi\rangle=0\) and \(
\langle \partial_t v_i, v_j\rangle+\langle \partial_t v_j, v_i\rangle=0.
\)
Following the derivation of \eqref{eqs_for_multiplier} and considering the constraints in \eqref{eq:RRmin-idxk}, we obtain
\begin{align}
    \tilde{\sigma}_i
    &= -\langle F(\phi),v_i\rangle,
    \label{eq:sigma_tilde}\quad i=1,2,\ldots,k,\\
    \tilde{\kappa}_{ij}
    &= (2-\delta_{ij})\langle\hat{H}(\phi)v_i,v_j\rangle ,
    \quad 1\le j\le i\le k.
    \label{eq:kappa_tilde}
\end{align}
Combining \eqref{eq:CGADk-phi-usingG} and \eqref{eq:CGADk-vi-usingLphi} with the Lagrange multipliers \eqref{eq:sigma_tilde} and \eqref{eq:kappa_tilde} yields the CGAD  for capturing index-\(k\) excited states.

Substituting the Lagrange multipliers \eqref{eq:sigma_tilde} and \eqref{eq:kappa_tilde} into \eqref{eq:CGADk-vi-usingLphi}, we collect the coefficients of $\phi$ and $v_j$, respectively, and obtain the following explicit form of a CGAD system
    \begin{equation}
		\left\{\begin{aligned}
			\partial_t\phi
			&=
			\frac{1}{2}\Delta \phi
			- V_d\phi
			- \beta_d |\phi|^2\phi
			- \lambda_d \Phi_{\mathrm{dip}}(|\phi|^2)\phi
			+ \mu(\phi)\phi + 2\displaystyle\sum_{j=1}^k \xi_j v_j ,
			\\
			\partial_t v_i
			&=
			\frac{1}{2}\Delta v_i
			- V_d v_i
			- 3\beta_d |\phi|^2 v_i
			- \lambda_d \Phi_{\mathrm{dip}}(|\phi|^2)v_i
			- 2\lambda_d  \Phi_{\mathrm{dip}}(\phi v_i) \phi
			+ \sigma_i \phi+ \displaystyle\sum_{j=1}^i \kappa_{ij} v_j, \\
            &\qquad (i = 1,2, \ldots, k)
		\end{aligned} \right.
		\label{eq:cgad_system}
	\end{equation}
where
\begin{align}
    \xi_j&=\langle F(\phi),v_j\rangle =\int_{\mathcal{D}}\! \bigg(
	\frac{1}{2}\nabla\phi \cdot \nabla v_j
	+ \big[V_d 
	+ \beta_d |\phi|^2 
	+ \lambda_d\Phi_{\mathrm{dip}}(|\phi|^2)-\mu(\phi)\big]\phi v_j
	\!\bigg)\mathrm{d}\mathbf{x}, \label{eq:xi} \\
    \sigma_i &=\tilde\sigma_i+\langle L_{\phi}\phi,  v_i \rangle
    = \langle L_{\phi}\phi-F(\phi),  v_i \rangle
	= \int_{\mathcal{D}}\! \bigg(
			2\beta_d|\phi|^2 \phi v_i
			+ 2\lambda_d \Phi_{\mathrm{dip}}(|\phi|^2)\phi v_i + \mu(\phi)\phi v_i\!\bigg) \mathrm{d}\mathbf{x}, \label{eq:lagrange-sigmai} \\
            \kappa_{ij} 
            &=\tilde{\kappa}_{ij}+\delta_{ij}\mu(\phi)=(2-\delta_{ij})\big\langle (L_{\phi}-\mu(\phi))v_i, v_j \big\rangle+\delta_{ij}\mu(\phi) \nonumber \\
            &= (2-\delta_{ij}) \!\int_{\mathcal{D}}\! \bigg(\!
			\frac{1}{2}\nabla v_i \!\cdot\! \nabla v_j
			\!+\! \left[V_d
			\!+\! 3\beta_d |\phi|^2 
			\!+\! \lambda_d \Phi_{\mathrm{dip}}(|\phi|^2)
            \!-\! \mu(\phi)\right] v_i v_j
			\!+\! 2\lambda_d \Phi_{\mathrm{dip}}(\phi v_i)\phi v_j \!\bigg)\mathrm{d}\mathbf{x}
            + \delta_{ij}\mu(\phi),  
            \label{eq:lagrange-kappaij} 
	\end{align}
		The initial data  \((\phi(\cdot,0),v_1(\cdot,0),\cdots,v_k(\cdot,0))=(\phi^0,v_1^{0},\cdots,v_k^{0})\) of \eqref{eq:cgad_system} is assumed to satisfy the constraints \eqref{eq:constraints}. By direct calculations, one can verify that the system \eqref{eq:cgad_system}-\eqref{eq:lagrange-kappaij} automatically preserves the constraints \eqref{eq:constraints}.

Following the same idea as in the linear stability analysis of GAD \cite{weinan2011gentlest} and its constrained and high-index variants, such as CGAD \cite{liu2023constrained} and CHiSD \cite{chisd}, we can conclude that, under certain conditions, a linearly stable steady state of the dynamical system \eqref{eq:cgad_system} gives exactly a nondegenerate index-$k$ constrained saddle point of the GP energy functional and $k$ eigenfunctions of the projected Hessian associated with negative eigenvalues, as stated in the following theorem. For the proof, we refer to Theorem~4.2 in \cite{liu2023constrained} or Theorem~1 in \cite{chisd}; it is omitted here for brevity.

\begin{theorem}
    Let \((\phi^*, v_1^*, \dots, v_k^*)\in X^{k+1}\) satisfy the constraints \eqref{eq:constraints}. Then \((\phi^*, v_1^*, \dots, v_k^*)\) is a steady state of the dynamical system \eqref{eq:cgad_system}, equivalently \eqref{eq:CGADk-phi-usingG} and \eqref{eq:CGADk-vi-usingLphi}, if and only if \(\phi^*\) is a constrained critical point of \(E\) on \(\mathcal{S}\) and \(v_i^*,\) \(i = 1, 2, \dots, k,\) are eigenfunctions of \(\hat{H}(\phi^*)\). Moreover,  \((\phi^*, v_1^*, \dots, v_k^*)\) is a linearly stable steady state if and only if $\phi^*$ is a nondegenerate index-$k$ constrained saddle point of $E$ on \(\mathcal{S}\) and the eigenvalues $\{\omega_j^*\}$ of \(\hat{H}(\phi^*)\) satisfy \(\omega_1^* < \omega_2^* < \cdots < \omega_k^* < 0 < \omega_{k+1}^* \le \omega_{k+2}^* \le \cdots\) with $v_i^*$ the eigenfunction corresponding to $\omega_i^*$, $i=1,2,\ldots,k$. 
\end{theorem}

\section{Numerical scheme}\label{sec:scheme}
In this section, we present the numerical scheme for the continuous dynamical system \eqref{eq:cgad_system}. Taking into account the nonlocal nature of the dipolar interaction and the complex spatial structure of higher-order excited states, we adopt different computational strategies for the dipolar interaction term in different dimensions in order to balance computational efficiency and the accuracy of spatial discretization. For the other terms, we use the sine-pseudospectral method for spatial discretization and employ a stable semi-implicit scheme for time integration. For simplicity, the following description is presented for 2D cases ($d = 2$), and the cases of other dimensions can be straightforwardly generalized.

\subsection{Semi-implicit sine-pseudospectral discretization}\label{time discretization}
The initial data $(\phi^0, v_1^0, \ldots, v_k^0)$ is chosen such that the constraints \eqref{eq:constraints} are satisfied. A suitable numerical discretization of the dynamical system \eqref{eq:cgad_system} that preserves the orthonormality constraints \eqref{eq:constraints} can produce an iterative algorithm for computing index-$k$ excited states. Here we adopt a simple and easy-to-implement predictor-corrector semi-discretization in time: first, a semi-implicit Euler scheme that treats only the Laplacian term implicitly is used to obtain an intermediate state; then, the Gram--Schmidt orthonormalization is applied in the corrector step to yield the next iterate. 
    
Let \(t_n = n\tau\) with \(n\geq0\), where \(\tau>0\) is the chosen time step size. Denote by \((\phi^n, v_1^n, \ldots, v_k^n)\) the numerical solution of the dynamical system \eqref{eq:cgad_system} at \(t=t_n\), and let GSON denote the standard Gram--Schmidt orthonormalization procedure. The predictor-corrector semi-implicit scheme for discretizing \eqref{eq:cgad_system} reads as
\begin{equation}
    \left\{\begin{aligned}
        &\frac{\tilde{\phi}^{n+1} - \phi^n}{\tau}
        =
        \frac{1}{2}\Delta\tilde{\phi}^{n+1}
        - V_d \phi^n
        - \beta_d|\phi^n|^2\phi^n
        - \lambda_d \Phi_{\mathrm{dip}}(|\phi^n|^2)\phi^n
        + \mu^n\phi^n
        + 2\sum_{j=1}^k \xi_j^n v_j^n,
        \\
        &\frac{\tilde{v}_i^{n+1} - v_i^n}{\tau}
        =
        \frac{1}{2}\Delta\tilde{v}_i^{n+1}
        - V_d v_i^n
        - 3\beta_d|\phi^n|^2v_i^n
        - \lambda_d \Phi_{\mathrm{dip}}(|\phi^n|^2)v_i^n
        - 2\lambda_d \Phi_{\mathrm{dip}}(\phi^n v_i^n)\phi^n
        \!+\! \sigma_i^n\phi^n
        \!+\! \sum_{j=1}^i \kappa_{ij}^n v_j^n, 
        \\
        &\qquad\qquad (i=1,2,\ldots,k) \\
        & \big[\phi^{n+1}, v_1^{n+1}, \dots, v_k^{n+1}\big]
        =
        \mathrm{GSON}\left(\big[\tilde{\phi}^{n+1}, \tilde{v}_1^{n+1}, \dots, \tilde{v}_k^{n+1}\big]\right),
    \end{aligned}\right. 
    \label{eq:cgadk-SI}
\end{equation}
Here, \(\mu^n=\mu(\phi^n)\) is defined by \eqref{eq:mu}, while \(\xi_i^n=\langle G(\phi^n),v_i^n\rangle\), \(\sigma_i^n=\langle L_{\phi^n}\phi^n-G(\phi^n),v_i^n\rangle\), and \(\kappa_{ij}^n=(2-\delta_{ij})\langle L_{\phi^n}v_i^n,v_j^n\rangle\) are the values of the quantities defined in \eqref{eq:xi}, \eqref{eq:lagrange-sigmai}, and \eqref{eq:lagrange-kappaij}, respectively, evaluated at the iterates \((\phi^n, v_1^n, \ldots, v_k^n)\) after the GSON correction. It is worth noting that these values are equivalent to their original forms, since GSON prevents the accumulation of constraint violations. We remark that, in ill-conditioned cases, GSON in \eqref{eq:cgadk-SI} can also be replaced by more numerically robust orthonormalization procedures.

In our numerical computation, the iterative scheme for computing index-$k$ excited states of a dipolar BEC is stopped when the following stopping criteria are satisfied:
\begin{equation}
    \left\| F_\phi^n \right\|_\infty < \varepsilon_{tol},\qquad
    \left\| F_{v_i}^n \right\|_\infty < \varepsilon_{tol},\quad i = 1, 2, \dots, k,
    \label{eq:stopping_criteria}
\end{equation}
where \(F_\phi^n:=F(\phi^n)\) is the residual of the Euler--Lagrange equation 
\eqref{eq:time_independent_GPE} at \(\phi^n\), 
 \(F_{v_i}^n:= \hat{H}(\phi^n)v_i^n
-\bigl\langle \hat{H}(\phi^n)v_i^n, v_i^n \bigr\rangle v_i^n
=L_{\phi^n}v_i^n
-(\sigma_i^n+\xi_i^n)\phi^n
-\kappa_{ii}^n v_i^n\) is the residual of the 
eigenvalue problem for the projected Hessian operator 
\(\hat{H}(\phi^n)\) \eqref{eq:ProjHessian} evaluated at \(v_i^n\), and $\varepsilon_{tol} > 0$ is a given tolerance.

For spatial discretization, we consider a computational domain $\mathcal{D}=[-L,L]^2$, equipped with homogeneous Dirichlet boundary conditions for \(\phi\) and \(v_i\), \(i=1,\ldots,k\). Denote $x_l=-L+lh$ and $y_m=-L+mh$, $l,m=0,1,\ldots,N$, where $h=2L/N$ with $N$ an even positive integer. Let \(\phi^n_{lm}\) be the numerical approximation for the value of a function \(\phi^n\) at \((x_l, y_m)\), and similarly, let \((v_i^n)_{lm}\) be the numerical approximation for the value of the function \(v_i^n\) at \((x_l, y_m)\). The discrete sine transform (DST) approximations of \(\phi^n\) and \(v_i^n\) are given as

\[
(\widehat{\phi^n})_{pq} = \frac{4}{N^2} \sum_{l=1}^{N-1} \sum_{m=1}^{N-1} \phi^n_{lm} \sin\left(\frac{l p \pi}{N}\right) \sin\left(\frac{m q \pi}{N}\right),
\;\;\;
(\widehat{v_i^n})_{pq} = \frac{4}{N^2} \sum_{l=1}^{N-1} \sum_{m=1}^{N-1} (v_i^n)_{lm} \sin\left(\frac{l p \pi}{N}\right) \sin\left(\frac{m q \pi}{N}\right).
\]
The DST approximation to the operator \(\Delta\) is given as
\[
(\Delta_h \phi^n)_{lm} = -\sum_{p=1}^{N-1} \sum_{q=1}^{N-1} \left( (\varrho_p^x)^2 + (\varrho_q^y)^2 \right) (\widehat{\phi^n})_{pq} \sin\left(\frac{l p \pi}{N}\right) \sin\left(\frac{m q \pi}{N}\right),
\]
\[
(\Delta_h v_i^n)_{lm} = -\sum_{p=1}^{N-1} \sum_{q=1}^{N-1} \left( (\varrho_p^x)^2 + (\varrho_q^y)^2 \right) (\widehat{v_i^n})_{pq} \sin\left(\frac{l p \pi}{N}\right) \sin\left(\frac{m q \pi}{N}\right),
\]
with \(\varrho_p^x = \frac{p\pi}{2L}, \; \varrho_q^y = \frac{q\pi}{2L}\), and \((\widehat{\phi^n})_{pq}\) being the discrete sine coefficients of the mesh function \(\phi^n_{lm}\).

Since the computation of the dipolar term requires special attention at each time step and directly using the DST leads to low accuracy, we will present a numerical scheme in the next section that exploits the specific structure of the dipolar term. For the purpose of presenting the fully discrete formulation here, we assume that the dipolar term is approximated at the grid point \((x_l, y_m)\) by \((\Phi_{\mathrm{dip}}(|\phi^n|^2))_{lm}\).

Therefore, at $\phi^n$, the discretized $L^2$-norm is given as $\|\phi^n\|_h^2
= h^2 \sum_{l=1}^{N-1} \sum_{m=1}^{N-1}
|\phi^n_{lm}|^2$, and the discretized energy and chemical potential can be computed as

\begin{align*}
E_h(\phi^n)
&= h^2 \sum_{l=1}^{N-1} \sum_{m=1}^{N-1}
\left(
-\frac{1}{2}\,\phi^n_{lm} (\Delta_h \phi^n)_{lm}
+ V_d(x_l, y_m)\,|\phi^n_{lm}|^2
+ \frac{\beta_d}{2} |\phi^n_{lm}|^4
+ \frac{\lambda_d}{2} (\Phi_{\mathrm{dip}}(|\phi^n|^2))_{lm}|\phi^n_{lm}|^2
\right), \\
\mu^n &= E_h(\phi^n)
    + \frac{h^2}{2} \sum_{l=1}^{N-1} \sum_{m=1}^{N-1}
    \left(
    \beta_d |\phi^n_{lm}|^4
    + \lambda_d (\Phi_{\mathrm{dip}} (|\phi^n|^2))_{lm}|\phi^n_{lm}|^2
    \right). 
\end{align*}
The discretization of scalars $\xi_i^n$, $\kappa_{ij}^n$ and $\sigma_i^n$ follows the same procedure, and is omitted here for brevity.

For the implementation of the semi-implicit spatial discretizations presented above, the main computational cost lies in solving the linear elliptic equations with constant coefficients. These equations can be rewritten in a compact form as:
\begin{equation}
    \left( 1 - \frac{\tau}{2}\Delta_h \right) (\tilde{\phi}^{n+1})_{lm} = (\mathcal{R}_{\phi}^n)_{lm}, \quad \quad
    \left( 1 - \frac{\tau}{2}\Delta_h \right) (\tilde{v}_i^{n+1})_{lm} = (\mathcal{R}_{v_i}^n)_{lm},
\end{equation}
where the right-hand side arrays are given by
\begin{align}
    (\mathcal{R}_{\phi}^n)_{lm} 
    &= \phi^n_{lm} + \tau \bigg[ - V_d(x_l, y_m)\phi^n_{lm} 
    - \beta_d |\phi^n_{lm}|^2 \phi^n_{lm} 
    - \lambda_d (\Phi_{\mathrm{dip}}(|\phi^n|^2))_{lm} \phi^n_{lm} 
    + \mu^n \phi^n_{lm} 
    + 2\sum_{j=1}^k \xi_j^n (v_j^n)_{lm} \bigg], 
    \label{H_phi_def} \\[4pt]
    (\mathcal{R}_{v_i}^n)_{lm} 
    &= (v_i^n)_{lm} + \tau \bigg[ - V_d(x_l, y_m)(v_i^n)_{lm} 
    - 3\beta_d |\phi^n_{lm}|^2 (v_i^n)_{lm} 
    - \lambda_d (\Phi_{\mathrm{dip}}(|\phi^n|^2))_{lm} (v_i^n)_{lm} \nonumber \\
    &\qquad\qquad\qquad\qquad 
    - 2\lambda_d \phi^n_{lm} (\Phi_{\mathrm{dip}}(\phi^n v_i^n))_{lm} 
    + \sigma_i^n \phi^n_{lm} 
    + \sum_{j=1}^i \kappa_{ij}^n (v_j^n)_{lm} \bigg]. 
    \label{H_vi_def}
\end{align}

Taking the DST on both sides of the elliptic equations, the system is decoupled in the spectral space. Consequently, the intermediate states $\tilde{\phi}^{n+1}$ and $\tilde{v}_i^{n+1}$ can be solved explicitly and efficiently by taking the inverse DST
\begin{equation}
    \left\{\begin{aligned}
        (\tilde{\phi}^{n+1})_{lm} = \sum_{p=1}^{N-1} \sum_{q=1}^{N-1} \frac{(\widehat{\mathcal{R}_{\phi}^n})_{pq}}{1 + \frac{\tau}{2} \left( (\varrho_p^x)^2 + (\varrho_q^y)^2 \right)} \sin\left(\frac{l p \pi}{N}\right) \sin\left(\frac{m q \pi}{N}\right), \\
        (\tilde{v}_i^{n+1})_{lm} = \sum_{p=1}^{N-1} \sum_{q=1}^{N-1} \frac{(\widehat{\mathcal{R}_{v_i}^n})_{pq}}{1 + \frac{\tau}{2} \left( (\varrho_p^x)^2 + (\varrho_q^y)^2 \right)} \sin\left(\frac{l p \pi}{N}\right) \sin\left(\frac{m q \pi}{N}\right),
    \end{aligned}\right.\label{eq:semi-implicit}
\end{equation}
where $(\widehat{\mathcal{R}_{\phi}^n})_{pq}$ and $(\widehat{\mathcal{R}_{v_i}^n})_{pq}$ are the DST coefficients of the mesh functions $(\mathcal{R}_{\phi}^n)_{lm}$ and $(\mathcal{R}_{v_i}^n)_{lm}$.

\subsection{Dipole-dipole interactions based on the NUFFT} \label{section:numicel sheme}
At each time step, the primary computational challenge in solving the aforementioned semi-implicit scheme lies in the efficient and highly accurate evaluation of the dipole convolution potentials \(\Phi_{\mathrm{dip}}(|\phi^n|^2)\) and \(\Phi_{\mathrm{dip}}(\phi^n v_i^n)\), \(i=1,\ldots,k\). 
We consider \(f^n\) defined on the computational domain \(\mathcal{D}\) satisfying the homogeneous Dirichlet boundary conditions, and present the numerical scheme for computing \(\Phi_{\mathrm{dip}}(f^n)\). 

The dipole potential \(\Phi_{\mathrm{dip}}(f^n)\) can be expressed in Fourier space as
\begin{equation}\label{dipolar_potential in Fourier space}
\Phi_{\mathrm{dip}}(f^n)(\mathbf x)
=
\frac{1}{(2\pi)^d}
\int_{\mathbb R^d}
e^{i\mathbf k\cdot\mathbf x}
\widehat{U_{\mathrm{dip}}^{(d)}}(\mathbf k)
\widehat{\mathcal E_0f^n}(\mathbf k)
\,\mathrm d\mathbf k,
\end{equation}
where \(\widehat{\mathcal E_0f^n}(\mathbf k)
=\int_{\mathbb R^d}
(\mathcal E_0f^n)(\mathbf x)
e^{-i\mathbf k\cdot\mathbf x}\,\mathrm d\mathbf x=
\int_{\mathcal D}
f^n(\mathbf x)
e^{-i\mathbf k\cdot\mathbf x}\,\mathrm d\mathbf x\)
and the Fourier transform of the dipole kernel defined in 
\eqref{ker3d} and \eqref{eq:U_dip_lower_d} reads as
\begin{align}\label{eq:U_dip_Fourier}
\widehat{U_{\mathrm{dip}}^{(d)}}(\mathbf k)=
\left\{\begin{aligned}
& \frac{1-3n_3^2}{2}
+\frac{3(3n_3^2-1)}{4}
\varepsilon^2k^2e^{\varepsilon^2k^2/2}E_1(\varepsilon^2k^2/2), & d=1, \\
&(3n_3^2-1)
+\frac{3\sqrt{2\pi}\varepsilon}{2}
\frac{(n_1k_x+n_2k_y)^2-n_3^2|\mathbf k|^2}{|\mathbf k|},&d=2, \\
&-1+3\frac{(\mathbf n\cdot\mathbf k)^2}{|\mathbf k|^2},&d=3,
\end{aligned}\right.
\end{align}
where \(E_1(s)=\int_s^\infty \frac{e^{-t}}{t}\,\mathrm dt,  s>0\). 
Clearly, \(\widehat{U_{\mathrm{dip}}^{(1)}}(k)\)  is smooth, and a standard FFT method works well for \eqref{dipolar_potential in Fourier space}. In contrast, $\widehat{U_{\mathrm{dip}}^{(2)}}(\mathbf k)$ is continuous but non-differentiable at \(\mathbf k=\mathbf 0\), whereas $\widehat{U_{\mathrm{dip}}^{(3)}}(\mathbf k)$ has a singularity at \(\mathbf k=\mathbf 0\). This nonsmoothness/singularity at the origin leads to the ``numerical locking'' phenomenon when the standard FFT is applied on a uniform Cartesian grid, thereby limiting the attainable accuracy on a fixed computational domain~\cite{bao2016accurate,jiang2014fast}.

To overcome the ``numerical locking'', we adopt a polar coordinate transformation in 2D and a spherical coordinate transformation in 3D. The nonsmoothness/singularity of $\widehat{U_{\mathrm{dip}}^{(d)}}(\mathbf k)$ at the origin is canceled out by the Jacobian determinant introduced by the coordinate transformations, and hence the resulting integrands are smooth. In 2D, the radial integral is evaluated by high-order Gauss--Legendre quadrature, while the angular integral is discretized by the trapezoidal rule. In 3D, Gauss--Legendre quadrature is applied to both the radial and polar-angle integrals, whereas the azimuthal-angle integral is discretized by the trapezoidal rule. Since the quadrature points generated in the Fourier space are nonequispaced, the resulting discrete summation is evaluated by the NUFFT algorithm rather than the standard FFT. This allows the high-order quadrature to be combined with fast Fourier evaluation, leading to an efficient computation of the dipolar interaction.

To evaluate the nonlocal term in \eqref{dipolar_potential in Fourier space}, we first zero-pad the source array \(f^n\) according to \eqref{zero_padd} and extend the physical grid to the enlarged computational domain \(\widetilde{\mathcal D}=[-3L,3L]^d\). However, applying the NUFFT directly to the entire transformed region of \(\widetilde{\mathcal D}\) in polar or spherical coordinates is more expensive than the standard FFT. Therefore, to reduce the computational cost, we introduce a dimension-dependent radial cutoff function \(W_d(|\mathbf k|)\) and split the truncated Fourier integral into
\[
\begin{aligned}
\Phi_{\mathrm{dip}}(f^n)(\mathbf x)
&\approx
\frac{1}{(2\pi)^d}\!\!\int_{\mathbb R^d}\!
e^{i\mathbf k\cdot\mathbf x}
\bigl(1\!-\!W_d(|\mathbf k|)\bigr)
\widehat{U_{\mathrm{dip}}^{(d)}}(\mathbf k)
\widehat{f^n}(\mathbf k)
\,\mathrm d\mathbf k +
\frac{1}{(2\pi)^d}\!\!\int_{B_P(0)}\!
e^{i\mathbf k\cdot\mathbf x}
W_d(|\mathbf k|)
\widehat{U_{\mathrm{dip}}^{(d)}}(\mathbf k)
\widehat{f^n}(\mathbf k)
\,\mathrm d\mathbf k\\
&:=\Phi_{I_1}(f^n)+\Phi_{I_2}(f^n),
\end{aligned}
\]
where \(B_P(0)=\{\mathbf k\in\mathbb R^d:|\mathbf k|\le P\}\) with \(P>0\) is chosen sufficiently large so that the Fourier truncation error is negligible.
In our implementation, the dimension-dependent radial cutoff functions are chosen as
\[
W_2(r)=
\begin{cases}
1, & 0\le r\le R_0,\\[4pt]
\dfrac12\operatorname{erfc}\!\left(
12\,\dfrac{r-(R_0+R_1)/2}{R_1-R_0}
\right), & R_0<r<R_1,\\[8pt]
0, & r\ge R_1,
\end{cases}
\qquad
W_3(r)=e^{-r^2/a},
\]
where \(\operatorname{erfc}(x)=\frac{2}{\sqrt\pi}\int_x^\infty e^{-t^2}\,\mathrm dt\) and \(a>0\) is selected according to the Fourier mesh size. For \(d=2\), \(R_0\) determines the neighborhood in which the low-frequency contribution is fully retained, while \(R_1\) specifies the outer boundary of the transition region. For \(d=3\), the Gaussian cutoff \(W_3\) localizes the NUFFT contribution near the origin and decays rapidly as \(r\) increases.
With these choices, \(\bigl(1-W_d(|\mathbf k|)\bigr)\widehat{U_{\mathrm{dip}}^{(d)}}(\mathbf k)\) is regularized near the origin and $\Phi_{I_1}(f^n)$ can be evaluated on a uniform Cartesian Fourier grid by the standard FFT. The  part $\Phi_{I_2}(f^n)$ is evaluated by the NUFFT on a polar/spherical grid.

Furthermore, taking \(d=2\) as an example, with the case \(d=3\) being analogous, we present the specific numerical scheme for computing \(\Phi_{\mathrm{dip}}(f^n)\) by computing the two parts \(\Phi_{I_1}(f^n)\) and \(\Phi_{I_2}(f^n)\) separately.

\begin{itemize}
    \item {\em Computation for \(\Phi_{I_1}(f^n)\) based on FFT.} 
Let \(f_{lm}^n=f^n(x_l,y_m)\) and embed it at the center of a triple-size array with \(3N\) grid points in each direction and the same mesh size \(h\). The enlarged computational domain is
\(\widetilde {\mathcal{D}}=[-3L,3L)^2\), with grid points
\(\widetilde x_{l'}=-3L+l'h\), \(\widetilde y_{m'}=-3L+m'h\), \(l',m'=0,\ldots,3N-1\).
The zero-padded source array \(\widetilde f_{l'm'}^n\) is defined by
\[
\widetilde f_{l'm'}^n=\begin{cases}
f_{l'-N,m'-N}^n,
&N\le l',m'\le2N-1,\\[4pt]
0,&\text{otherwise}.
\end{cases}
\]
    The Cartesian Fourier nodes are given by \(\mathbf k_{p'q'}=(k_{p'}^x,k_{q'}^y)=(\frac{p'\pi}{3L},\frac{q'\pi}{3L})\), \(p',q'=-\frac{3N}{2},\ldots,\frac{3N}{2}-1\).
The discrete Fourier transform (DFT) of the zero-padded source at these nodes is approximated by
\[
\widehat{\widetilde f^n}(\mathbf k_{p'q'})
=
\sum_{l'=0}^{3N-1}
\sum_{m'=0}^{3N-1}
\widetilde f_{l'm'}^n
e^{-i(k_{p'}^x\widetilde x_{l'}+
k_{q'}^y\widetilde y_{m'})}.
\]
Hence, the value of \(\Phi_{I_1}(f^n)\) at \((x_l,y_m)\) can be computed as
\[
\bigl(\Phi_{I_1}(f^n)\bigr)_{lm}
=\frac{1}{9N^2}
\sum_{p'=-3N/2}^{3N/2-1}\sum_{q'=-3N/2}^{3N/2-1}\left(1-W_2(|\mathbf k_{p'q'}|)\right)
\widehat{U_{\mathrm{dip}}^{(2)}}(\mathbf k_{p'q'}) \widehat{\widetilde f^n}(\mathbf k_{p'q'}) e^{i(k_{p'}^x x_l+k_{q'}^y y_m)}.
\]

\item {\em Computation for \(\Phi_{I_2}(f^n)\) based on NUFFT.}
For the  part \(\Phi_{I_2}(f^n)\), we introduce the polar coordinates \(\mathbf k=r(\cos\theta,\sin\theta)\), \(\mathrm d\mathbf k=r\,\mathrm dr\,\mathrm d\theta\). It follows that
\[
\Phi_{I_2}(f^n)(\mathbf x)
=
\frac{1}{(2\pi)^2}
\int_0^P\int_0^{2\pi}
e^{ir(x\cos\theta+y\sin\theta)}
W_2(r)\,\mathcal K(r,\theta)\,
\widehat{f^n}\!\left(r\cos\theta,r\sin\theta\right)
\,\mathrm d\theta\,\mathrm dr,
\]
where
\[
\mathcal K(r,\theta)
:=
r\,\widehat{U_{\mathrm{dip}}^{(2)}}
\!\left(r\cos\theta,r\sin\theta\right)
=
(3n_3^2-1)r
+
\frac{3\sqrt{2\pi}\varepsilon}{2}
r^2
\left[
\left(n_1\cos\theta+n_2\sin\theta\right)^2-n_3^2
\right].
\]
It can be seen that the  factor \(1/|\mathbf k|=1/r\) in \(\widehat{U_{\mathrm{dip}}^{(2)}}\) is canceled by the Jacobian determinant \(r\) arising from the polar coordinate transformation. In particular, \(\mathcal K(r,\theta)\) contains no negative power of \(r\) and is smooth at \(r=0\). Let \(\{r_a,w_a\}_{a=1}^{N_r}\) be the shifted and scaled Gauss--Legendre nodes and weights on \([0,P]\), and \(\theta_b=\frac{2\pi b}{N_\theta}\), \(b=0,\ldots,N_\theta-1\) be the equispaced angular nodes with weight \(\Delta\theta=2\pi/N_\theta\). The corresponding nonequispaced Fourier nodes are denoted by \(\mathbf q_{ab}=(q_{ab}^x,q_{ab}^y):=r_a(\cos\theta_b,\sin\theta_b)\). The  nonuniform DFT of the zero-padded source at $(q_{ab}^x,q_{ab}^y)$ is given by
\[
\widehat{\widetilde f^n}(\mathbf q_{ab})
=
\sum_{l'=0}^{3N-1}
\sum_{m'=0}^{3N-1}
\widetilde f_{l'm'}^n
e^{-i(q_{ab}^x\widetilde x_{l'}+q_{ab}^y\widetilde y_{m'})}.
\]
Therefore, at the physical grid point \((x_l,y_m)\)
\[
\bigl(\Phi_{I_2}(f^n)\bigr)_{lm}
=
\frac{h^2}{2\pi N_{\theta}}
\sum_{a=1}^{N_r}
\sum_{b=0}^{N_\theta-1}
\widehat{\widetilde f^n}(\mathbf q_{ab})
W_2(r_a)\mathcal K(r_a,\theta_b)
w_a\,
e^{i(q_{ab}^x x_l+q_{ab}^y y_m)}.
\]
The use of the NUFFT preserves the high-order accuracy of the polar 
quadrature while enabling efficient evaluation of 
\(\Phi_{I_2}^n(f)(\mathbf x)\).
\end{itemize}

In  \eqref{eq:cgadk-SI}, we take \(f^n=|\phi^n|^2\) or \(f^n=\phi^n v_i^n\) (\(i=1,\ldots,k\)), evaluate
$\bigl(\Phi_{\mathrm{dip}}(f^n)\bigr)_{lm}=\bigl(\Phi_{I_1}(f^n)\bigr)_{lm}+\bigl(\Phi_{I_2}(f^n)\bigr)_{lm}$
and substitute the results into \eqref{H_phi_def} and \eqref{H_vi_def} accordingly. In this paper, we take \(R_0=\frac{\pi}{3L}\), \(R_1=32R_0\), \(P=1.2R_1\), \(N_r=N_{\theta}=80\) for 2D case and \(a=\frac{\pi}{3L}\) for 3D case.

\section{Numerical results}   \label{results}
The numerical algorithm presented in Section~\ref{sec:scheme} is employed to compute the excited states of 1D, 2D, and 3D dipolar BECs, with spatial step size \(h=\frac{1}{16}\), time step \(\tau=0.01\), and \(\varepsilon_{\mathrm{tol}}=10^{-12}\) in \eqref{eq:stopping_criteria}. In particular, we investigate the asymptotic behaviors of the excited-state energy and chemical potential as functions of the dipolar interaction parameter $\lambda_d$. Additionally, the energies and chemical potentials of the ground state and various excited states with different Morse indices are systematically compared. In the numerical experiments, we take the following harmonic oscillator potential with parameters $\gamma_x,\gamma_y,\gamma_z>0$ \cite{bao2013mathematical}:
\begin{equation}
    V_d(\mathbf{x})=V_{\text{ho}}(\mathbf{x}) := 
    \begin{cases}
         \gamma_x^2x^2/2, & d=1, \\[4pt]
        (\gamma_x^2x^2 + \gamma_y^2y^2)/2, & d=2, \\[4pt]
        (\gamma_x^2x^2 + \gamma_y^2y^2 + \gamma_z^2z^2)/2, & d=3.
    \end{cases}\label{V}
\end{equation}

The initial data for both $\phi^0$ and $\{v_i^0\}$ are constructed from the Hermite functions associated with the isotropic non-interacting linear harmonic oscillator. Specifically, denoting $x_1:=x$, $x_2:=y$, $x_3:=z$, for $\mathbf{x}=(x_1,\ldots,x_d)^T\in\mathbb{R}^d$ ($d=1,2,3$), the eigenfunctions of the linear eigenvalue problem \eqref{eq:linear} with $V_d(\mathbf{x}) =\frac{1}{2}\sum_{\alpha=1}^dx_{\alpha}^2$ in $\mathbb{R}^d$ read as
	\begin{equation*}
		\phi_\mathbf{j}^{\text{ho}}(\mathbf{x}) = 
		\prod_{\alpha=1}^d \hat{h}_{j_\alpha}\!\left(\, x_\alpha\right), \quad
		\mathbf{j} = (j_1, \dots, j_d) \in \mathbb{N}^d,
	\end{equation*}
with $\hat{h}_{j_{\alpha}}$ denoting the normalized 1D Hermite functions: 
\[
\hat{h}_{j_{\alpha}}(\xi) = \frac{1}{\pi^{1/4}\sqrt{2^{j_{\alpha}} {j_{\alpha}}!}} e^{-\xi^2/2} h_{j_{\alpha}}(\xi),\quad 
h_{j_{\alpha}}(\xi) = (-1)^{j_{\alpha}} e^{\xi^2} \frac{d^{j_{\alpha}}}{d\xi^{j_{\alpha}}}(e^{-\xi^2}),\quad
j_{\alpha} = 0, 1, \dots.
\]
In practice, for some given $\mathbf{j} = (j_1, \dots, j_d) \in \mathbb{N}^d$, we take
\[
    \phi^0=\phi_\mathbf{j}^{\mathrm{ho}},\quad
    \{v_i^0\}_{i=1}^{k} = \left\{\phi_{\mathbf{j}'}^{\mathrm{ho}} \,:\, |\mathbf{j}'|<j,\, \mathbf{j}'\in\mathbb{N}^d\right\},
\]
where $j:=|\mathbf{j}|=\sum_{\alpha=1}^{d}j_\alpha$ and \(k\), the number of initial unstable directions, is explicitly given as
\[
k=\#\{\mathbf{j}'\in\mathbb{N}^d:|\mathbf{j}'|<j\}=
\begin{cases}
j, & d=1,\\
\frac12j(j+1), & d=2,\\
\frac16j(j+1)(j+2), & d=3.
\end{cases}
\]
This choice of initial data is expected to target an excited state with Morse index \(k\).

\subsection{Numerical results in 1D}
	\begin{example}
		\label{1d_lambda}
		We take the computational domain to be \([-12,12]\), $\beta_1=10$, $\varepsilon=\sqrt{\frac{1}{80}}$, $n_3=1$, $V_1=V_{\mathrm{ho}}$ in \eqref{V} with $\gamma_x=1$. The ground state \(\phi_0^{\lambda_1}\) is initialized by setting \(\phi^0(x)=\phi_0^{\mathrm{ho}}(x)=\hat h_0(x)\), whereas for computing the index-\(k\) excited state \(\phi_k^{\lambda_1}\), \(k=1,2,\ldots,9\), we take $\phi^0(x)=\phi_k^{\mathrm{ho}}(x)=\hat h_k(x)$ and $v_i^0(x)=\phi_{i-1}^{\mathrm{ho}}(x)=\hat h_{i-1}(x)$, $i=1,2,\ldots,k$. Table~\ref{tab:combined_results_lambda} reports the energies and chemical potentials under various dipolar interaction strengths $\lambda_1$. Fig.~\ref{phi_lambda} plots the profiles of $\phi_k^{\lambda_1}(x)$, $k = 1, 2, \ldots, 9$, for different $\lambda_1 = 100, 500, 1500$. Moreover, the asymptotic behaviors of the energies of excited states in the weakly dipolar interaction regime, $0 < \lambda_1 \ll 1$, and the strongly dipolar interaction regime, i.e., $\lambda_1 \gg 1$, are shown in Fig.~\ref{fig:linear_relation}.

From the numerical results that are partially shown in Figs.~\ref{phi_lambda} and~\ref{fig:linear_relation} and Table \ref{tab:combined_results_lambda}, we have the following numerical observations
\begin{enumerate}
	\item For each fixed state index \(k\), as \(\lambda_1\) increases, the wave functions of  excited states become more spatially extended, while the number and ordering of their nodes remain unchanged. The energy \(E(\phi_k^{\lambda_1})\) and chemical potential \(\mu(\phi_k^{\lambda_1})\) both increase monotonically with \(\lambda_1\); for each fixed \(\lambda_1\), the energy and chemical potential orderings are
\[
E(\phi_0^{\lambda_1})<E(\phi_1^{\lambda_1})<\cdots<E(\phi_9^{\lambda_1}),\qquad
\mu(\phi_0^{\lambda_1})<\mu(\phi_1^{\lambda_1})<\cdots<\mu(\phi_9^{\lambda_1}).
\]
Therefore, the dipolar interaction changes the spatial shape but does not alter the nodal structure or energy hierarchy associated with the indices of the excited-state wave functions.
	
	\item In the weakly dipolar regime \(0<\lambda_1\ll1\), the numerical results suggest a first-order linear perturbative expansion:
\[
E(\phi_k^{\lambda_1})
=
E(\phi_k^0)+D_k\lambda_1+o(\lambda_1),
\qquad k=0,1,\ldots,9,
\]
where \(D_k=\frac12\int_{\mathcal{D}}\Phi_{\mathrm{dip}}(|\phi_k^0|^2)|\phi_k^0|^2\,\mathrm d\mathbf x\).
In the strongly dipolar regime, the numerical results suggest a sublinear power-law growth of the energies, consistent with
\[
E(\phi_k^{\lambda_1})
\sim C_*\lambda_1^{2/3},
\quad \lambda_1\to+\infty, \quad k=0,1,\ldots,9,
\]
where \(C_*\approx 0.16\) numerically.
\end{enumerate}
		\begin{table}[!t]
			\centering\footnotesize
			\caption{Energy \(E\) and chemical potential \(\mu\) of the ground and excited states in Example \ref{1d_lambda} for varying \(\lambda_1\) with fixed \(\beta_1=10\).}
			\label{tab:combined_results_lambda}
			\resizebox{0.99\textwidth}{!}{
				\begin{tabular}{lcccccccccc}
					\toprule
					\(\lambda_1\) & \(E(\phi_0^{\lambda_1})\) & \(E(\phi_1^{\lambda_1})\) & \(E(\phi_2^{\lambda_1})\) & \(E(\phi_3^{\lambda_1})\) & \(E(\phi_4^{\lambda_1})\) & \(E(\phi_5^{\lambda_1})\) & \(E(\phi_6^{\lambda_1})\) & \(E(\phi_7^{\lambda_1})\) & \(E(\phi_8^{\lambda_1})\) & \(E(\phi_9^{\lambda_1})\) \\
					\midrule
					0     & 1.94713 & 2.76538 & 3.64568 & 4.55841 & 5.49091 & 6.43654 & 7.39147 & 8.35325 & 9.32029 & 10.2914 \\
					0.01  & 1.94727 & 2.76550 & 3.64579 & 4.55852 & 5.49100 & 6.43663 & 7.39155 & 8.35333 & 9.32036 & 10.2915 \\
					1     & 1.96118 & 2.77802 & 3.65715 & 4.56893 & 5.50063 & 6.44561 & 7.39998 & 8.36129 & 9.32790 & 10.2987 \\
					10    & 2.08573 & 2.89068 & 3.75969 & 4.66312 & 5.58791 & 6.52708 & 7.47650 & 8.43356 & 9.39646 & 10.3640 \\
					100   & 3.19313 & 3.92384 & 4.72326 & 5.56435 & 6.43420 & 7.32508 & 8.23194 & 9.15127 & 10.0806 & 11.0180 \\
					500   & 6.83746 & 7.45906 & 8.15365 & 8.88652 & 9.64893 & 10.4356 & 11.2428 & 12.0675 & 12.9073 & 13.7602 \\
					1000  & 10.32755 & 10.89012 & 11.54053 & 12.2254 & 12.9359 & 13.6680 & 14.4192 & 15.1873 & 15.9707 & 16.7678 \\
					2000  & 15.98028 & 16.46818 & 17.07570 & 17.7186 & 18.3834 & 19.0661 & 19.7645 & 20.4773 & 21.2034 & 21.9418 \\
					\midrule
					\({\lambda_1}\) & \(\mu(\phi_0^{\lambda_1})\) & \(\mu(\phi_1^{\lambda_1})\) & \(\mu(\phi_2^{\lambda_1})\) & \(\mu(\phi_3^{\lambda_1})\) & \(\mu(\phi_4^{\lambda_1})\) & \(\mu(\phi_5^{\lambda_1})\) & \(\mu(\phi_6^{\lambda_1})\) & \(\mu(\phi_7^{\lambda_1})\) & \(\mu(\phi_8^{\lambda_1})\) & \(\mu(\phi_9^{\lambda_1})\) \\
					\midrule
					0     & 3.10724 & 3.86320 & 4.68057 & 5.53782 & 6.42244 & 7.32672 & 8.24566 & 9.17583 & 10.1148 & 11.0610 \\
					0.01  & 3.10749 & 3.86343 & 4.68079 & 5.53801 & 6.42262 & 7.32690 & 8.24582 & 9.17599 & 10.1150 & 11.0611 \\
					1     & 3.13156 & 3.88600 & 4.70184 & 5.55769 & 6.44108 & 7.34428 & 8.26226 & 9.19159 & 10.1299 & 11.0753 \\
					10    & 3.34667 & 4.08847 & 4.89147 & 5.73539 & 6.60809 & 7.50181 & 8.41141 & 9.33331 & 10.2649 & 11.2045 \\
					100   & 5.23680 & 5.90487 & 6.63044 & 7.39638 & 8.19411 & 9.01745 & 9.86173 & 10.7233 & 11.5995 & 12.4879 \\
					500   & 11.36426 & 11.91630 & 12.54652 & 13.2063 & 13.8896 & 14.5942 & 15.3182 & 16.0601 & 16.8184 & 17.5916 \\
					1000  & 17.20048 & 17.67868 & 18.27117 & 18.8931 & 19.5333 & 20.1891 & 20.8596 & 21.5441 & 22.2420 & 22.9527 \\
					2000  & 26.63727 & 27.00587 & 27.55057 & 28.1374 & 28.7426 & 29.3600 & 29.9878 & 30.6252 & 31.2719 & 31.9275 \\
					\bottomrule
				\end{tabular}
			}
		\end{table}
		\begin{figure}[!t]
			\centering
				\includegraphics[width=0.3\textwidth]{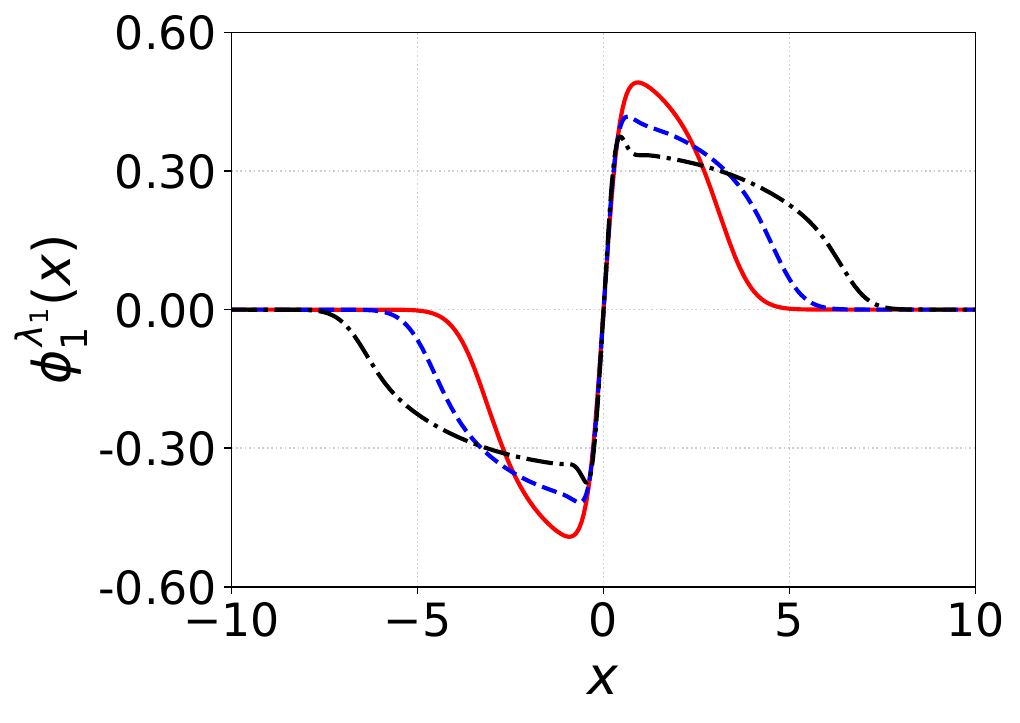}
				\includegraphics[width=0.3\textwidth]{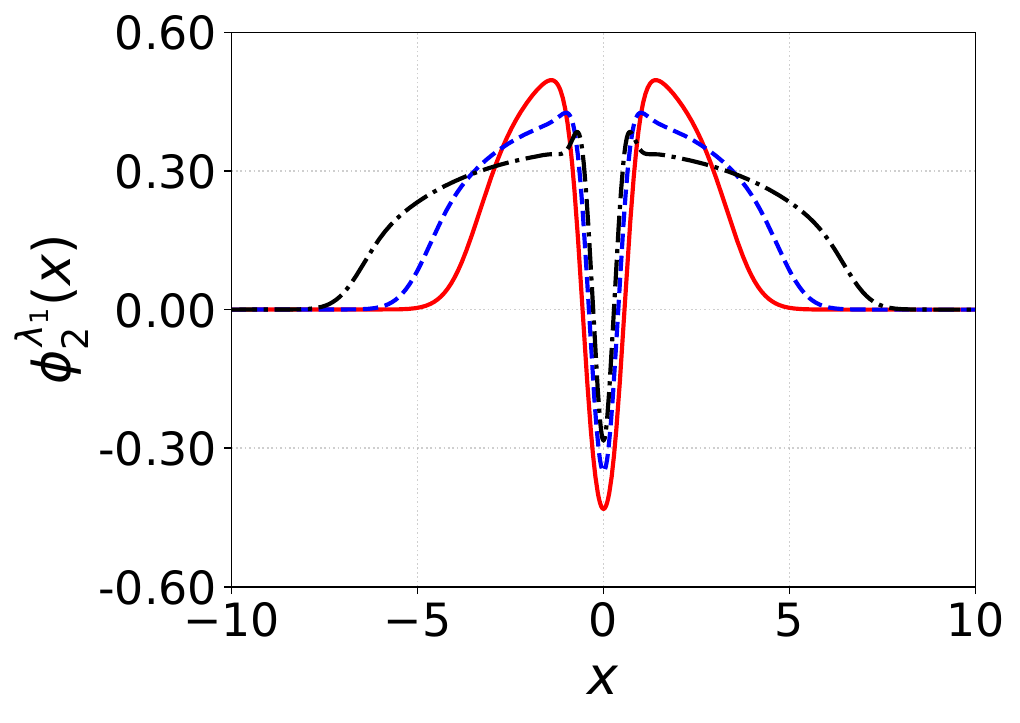}
				\includegraphics[width=0.3\textwidth]{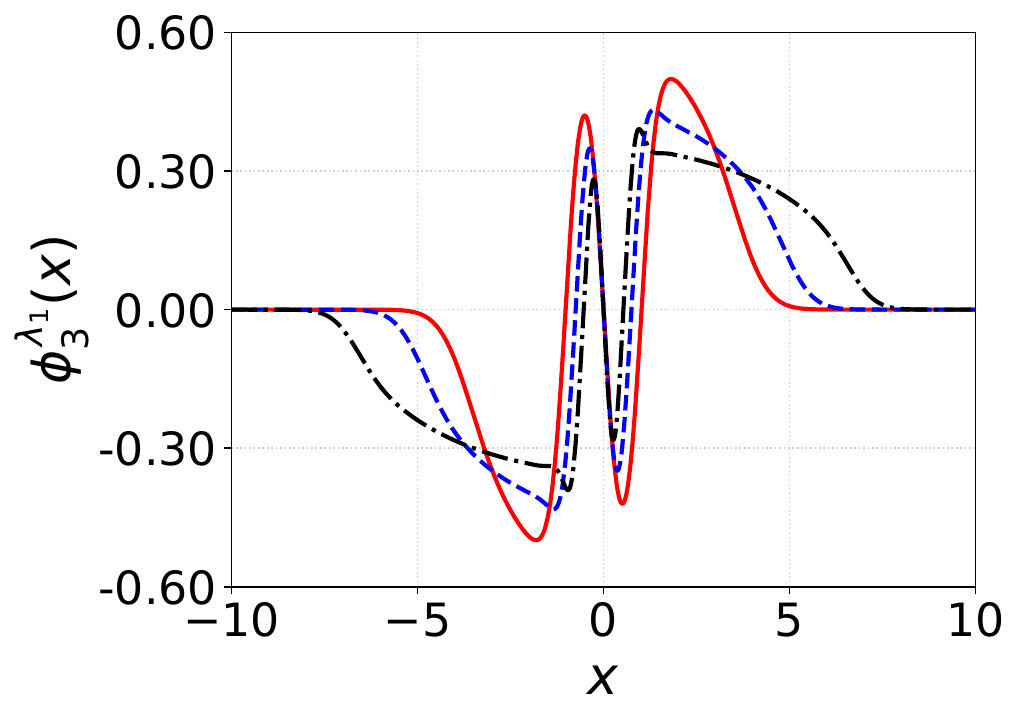} \\
				\includegraphics[width=0.3\textwidth]{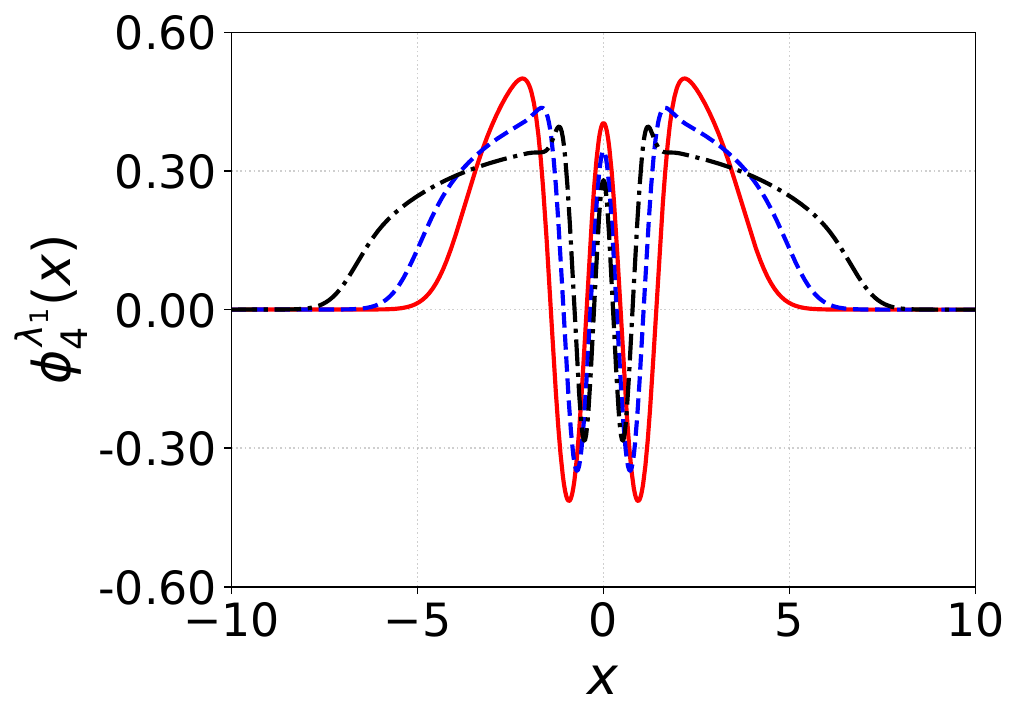}
				\includegraphics[width=0.3\textwidth]{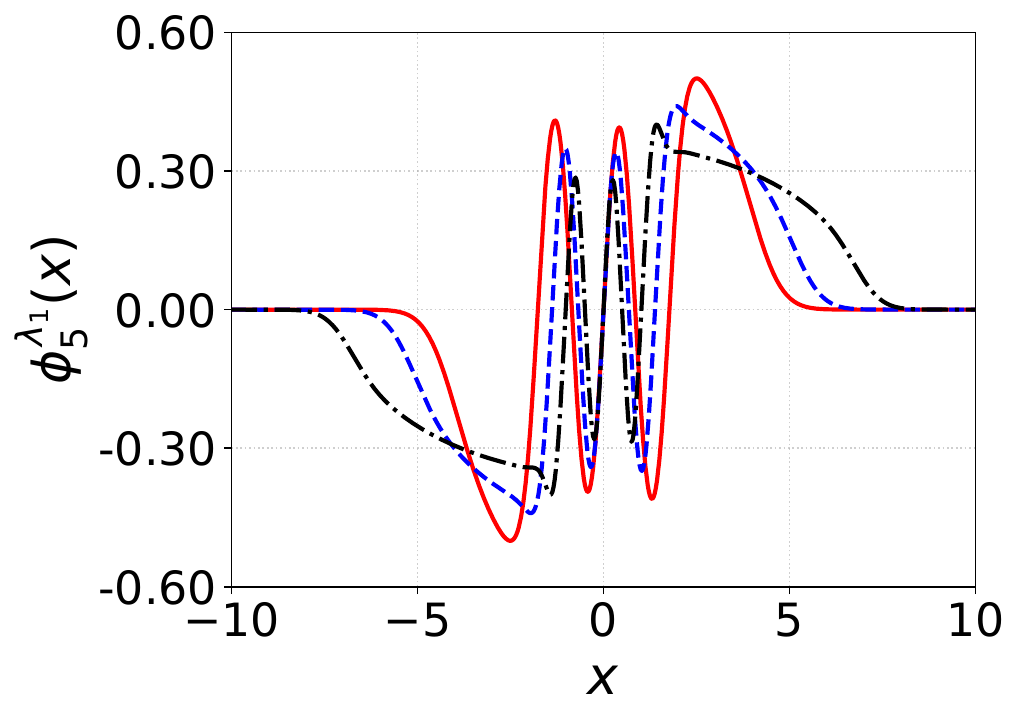}
				\includegraphics[width=0.3\textwidth]{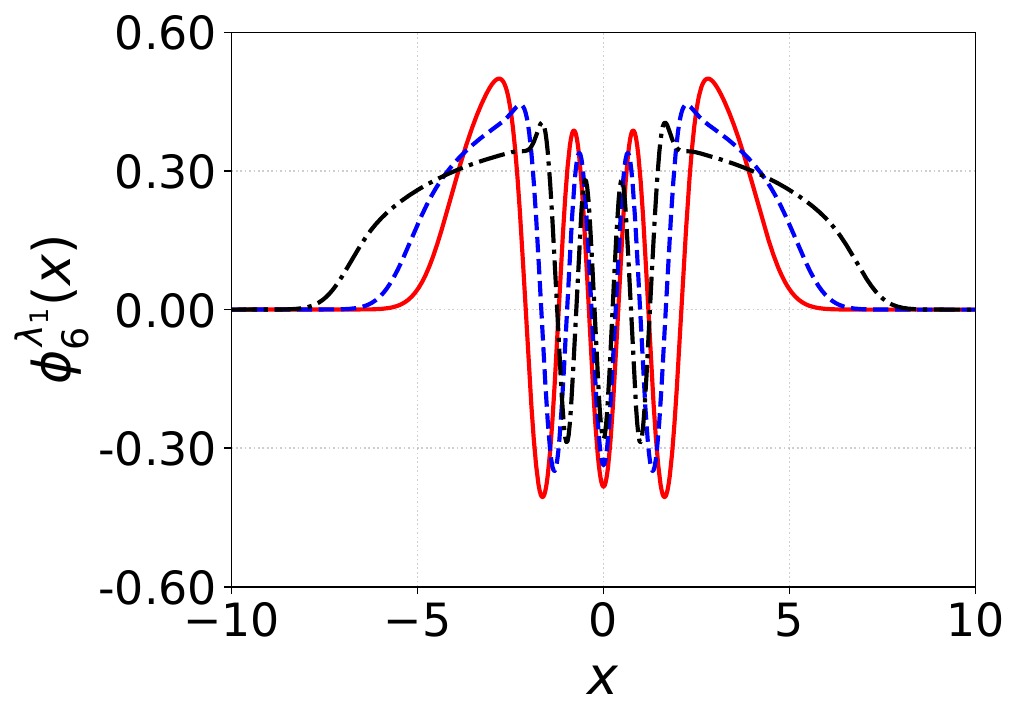} \\
				\includegraphics[width=0.3\textwidth]{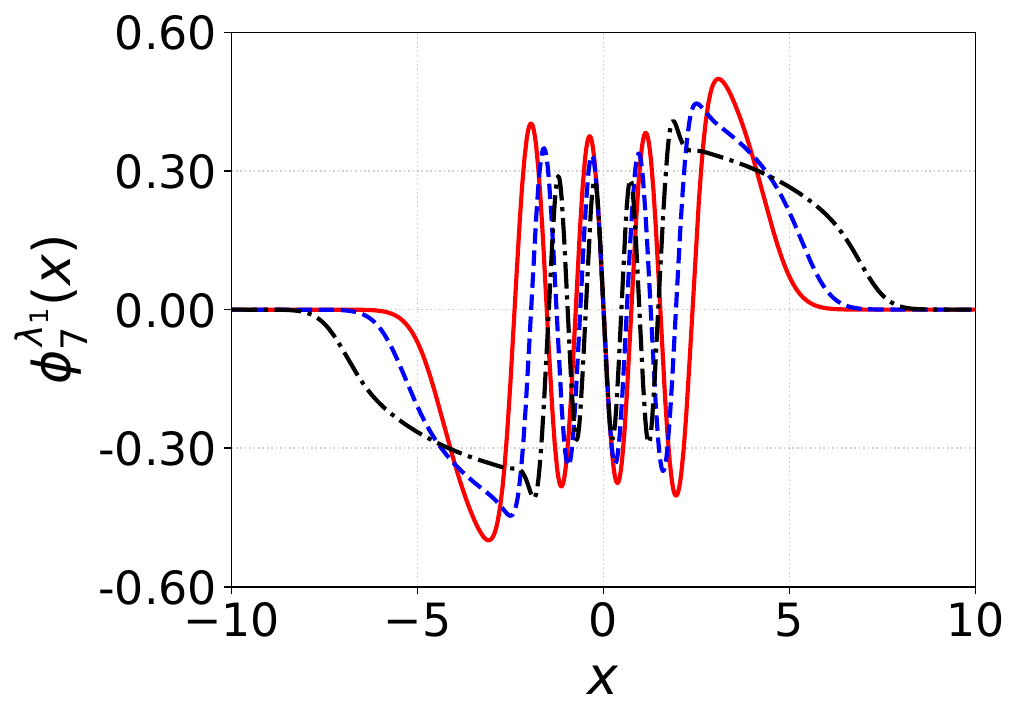}
				\includegraphics[width=0.3\textwidth]{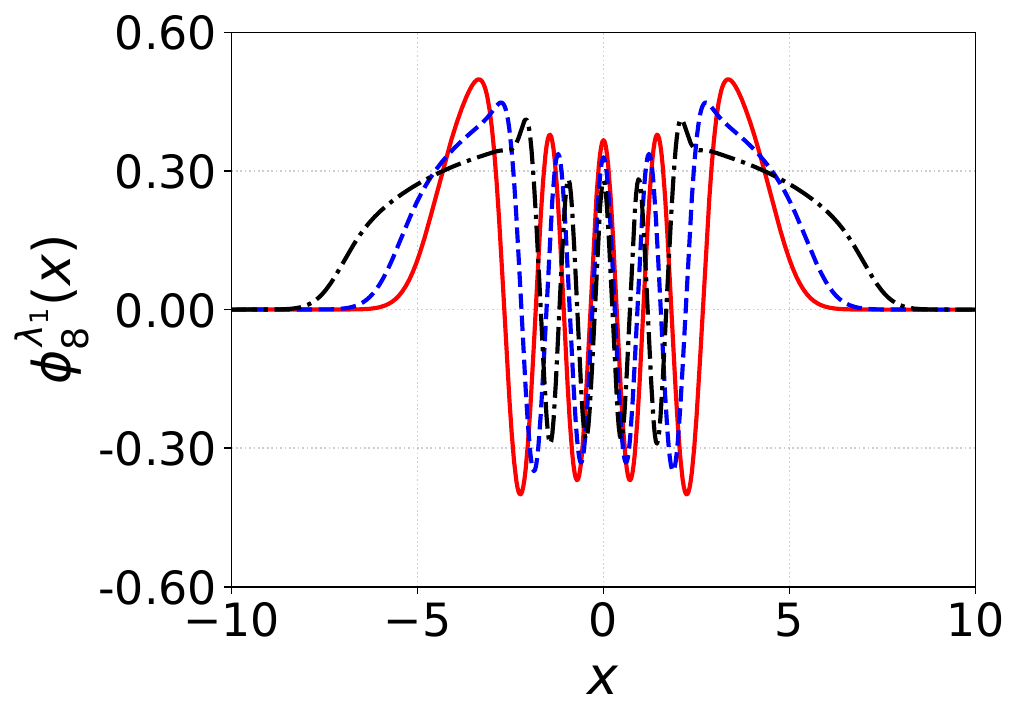}
				\includegraphics[width=0.3\textwidth]{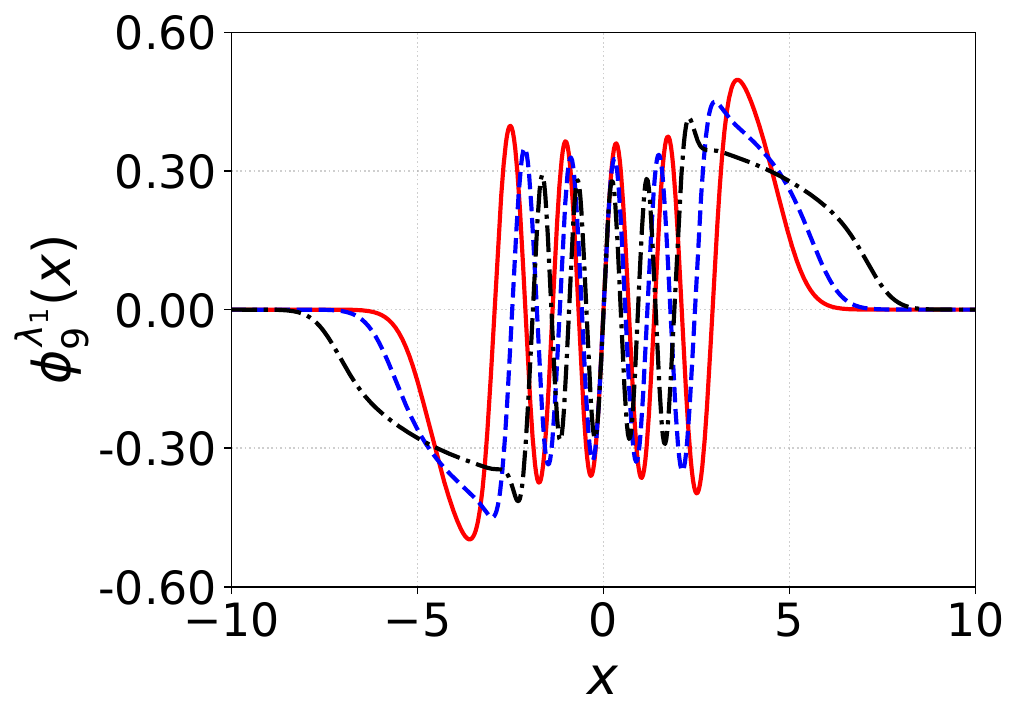} \\
			\caption{Profiles of index-\(k\) (\(k=1,2,\ldots,9\)) excited states \(\phi_k^{\lambda_1}(x)\) in Example \ref{1d_lambda}, with \(\beta_1=10\) fixed and \(\lambda_1=100\) (red solid lines), 500 (blue dashed lines), and 1500 (black dash-dot lines).}
			\label{phi_lambda}
		\end{figure}
		
		\begin{figure}[!t]
			\centering
			\includegraphics[width=0.95\textwidth]{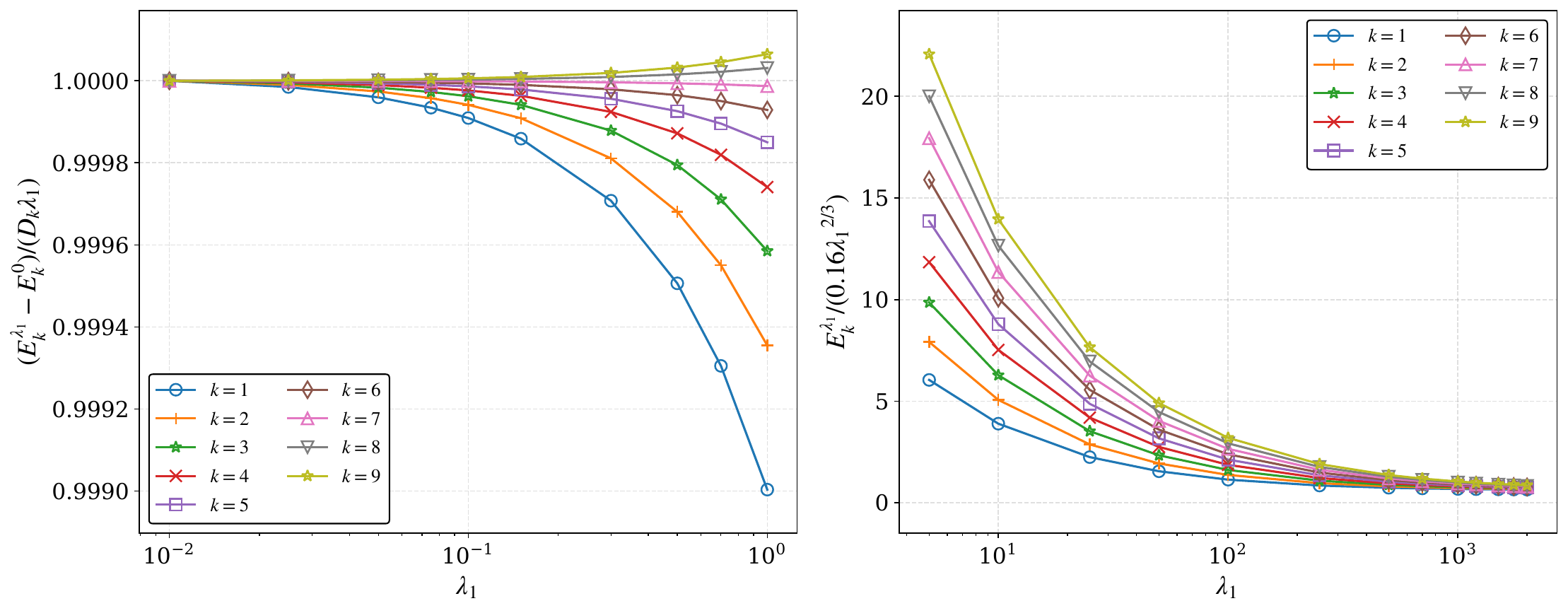}
			\caption{Asymptotic behaviors of the energies for index-$k$ ($k = 1, 2, \dots, 9$) excited states for the weakly (left) and strongly (right) dipolar interaction regime in Example \ref{1d_lambda}.}
            \label{fig:linear_relation}
		\end{figure}
	\end{example}

\subsection{Numerical results in 2D} \label{2D_example}
In this section, we consider the 2D case (\(d=2\)), where the trapping frequencies in \eqref{V} are set to $\gamma_x=\gamma_y=1$, and the computational domain is taken as \([-8,8]^2\). The initial data are given in Table~\ref{tab:initial_2d}. We systematically investigate the effects of both repulsive (\(\lambda_2>0\)) and attractive (\(\lambda_2<0\)) dipolar interactions on the ground and excited states of the BEC for various polarization directions.
    \begin{table}[htbp]
\centering
\caption{Initial data for the numerical experiments in Subsection~\ref{2D_example}.}
\label{tab:initial_2d}
\begin{tabular}{c  c l l}
\toprule
Solution & Morse index $k$  & Initial guess of $\phi^0$ & Initial guess of $\{v_i^0\}$ \\
\midrule
$\phi_{00}$ & 0  & $\phi_{00}^{\mathrm{ho}}$ & -- \\
$\phi_{10}$ & 1  & $\phi_{10}^{\mathrm{ho}}$ & $\{\phi_{00}^{\mathrm{ho}}\}$ \\
$\phi_{01}$ & 1  & $\phi_{01}^{\mathrm{ho}}$ & $\{\phi_{00}^{\mathrm{ho}}\}$ \\
$\phi_{11}$ & 3 & $\phi_{11}^{\mathrm{ho}}$ & $\{\phi_{00}^{\mathrm{ho}},\phi_{10}^{\mathrm{ho}},\phi_{01}^{\mathrm{ho}}\}$ \\
\bottomrule
\end{tabular}
\end{table}
	\begin{example}
	To investigate the anisotropic effects of the DDI, we fix the contact interaction strength at a strongly repulsive value of $\beta_2 = 100$, and set the dipole polarization vector strictly along the $x$-axis, i.e., $\mathbf{n} = (1,0,0)$, $\varepsilon=\sqrt{\frac{1}{80}}$. Table \ref{tab:2d_excited_states} tabulates the energies and chemical potentials of the ground state and the excited states for varying relative dipolar interaction strengths $\lambda_2/\beta_2 \in [0, 1]$. Fig.~\ref{fig:wavefunctions_beta_lambda} plots the pseudo-color images of excited states
		with $\lambda_2/\beta_2 = 0$, $\lambda_2/\beta_2 = 0.5$ and $\lambda_2/\beta_2 = 1$.
		\begin{table}[!t]
			\centering\footnotesize
			\caption{Energies and chemical potentials of the ground and excited states for different $\lambda_2/\beta_2$.	\label{tab:2d_excited_states}}
			\begin{tabular}{c c c c c c c c c}
				\toprule
				\multirow{2}{*}{$\lambda_2/\beta_2$} & 
				\multicolumn{2}{c}{$\phi_{00}$} & \multicolumn{2}{c}{$\phi_{10}$} & \multicolumn{2}{c}{$\phi_{01}$} & \multicolumn{2}{c}{$\phi_{11}$} \\
				\cmidrule(lr){2-3} \cmidrule(lr){4-5} \cmidrule(lr){6-7} \cmidrule(lr){8-9}
				& $E$ & $\mu$ & $E$ & $\mu$ & $E$ & $\mu$ & $E$ & $\mu$ \\
				\midrule
				0.00 & 3.9459 & 5.7598 & 4.7329 & 6.5110 & 4.7329 & 6.5110 & 5.5204 & 7.2613 \\
				0.25 & 3.5279 & 5.1042 & 4.3777 & 5.9354 & 4.3211 & 5.8575 & 5.1700 & 6.6863 \\
				0.50 & 3.0593 & 4.3620 & 3.9896 & 5.2996 & 3.8621 & 5.1186 & 4.7893 & 6.0519 \\
				0.75 & 2.5142 & 3.4837 & 3.5554 & 4.5767 & 3.3347 & 4.2490 & 4.3684 & 5.3341 \\
				1.00 & 1.8242 & 2.3244 & 3.0459 & 3.7015 & 2.6956 & 3.1423 & 3.8895 & 4.4859 \\
				\bottomrule
			\end{tabular}
		\end{table}
		
		\begin{figure}[!t]
			\centering
			\includegraphics[width=0.22\linewidth]{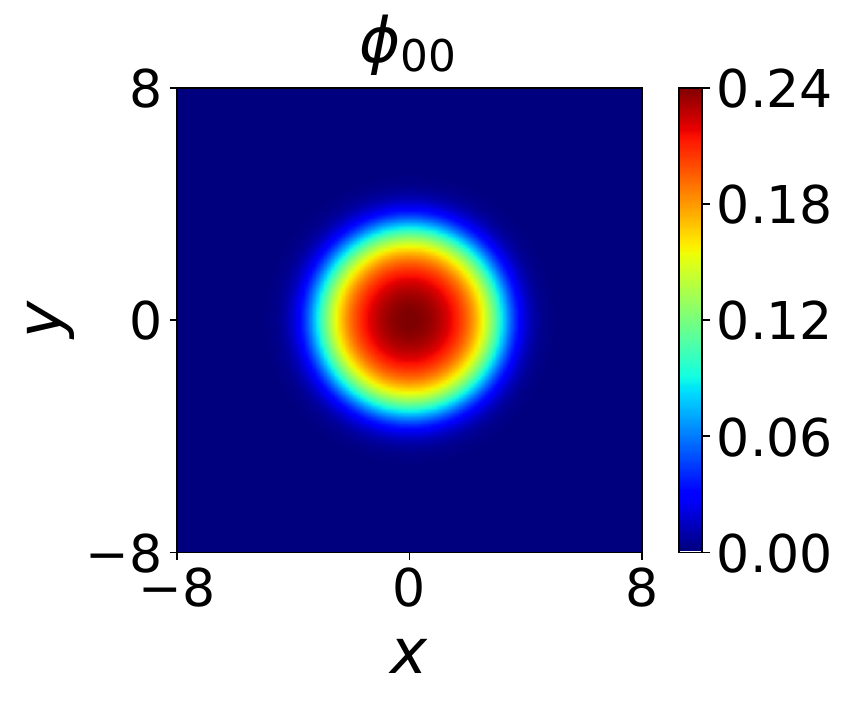}
			\includegraphics[width=0.22\linewidth]{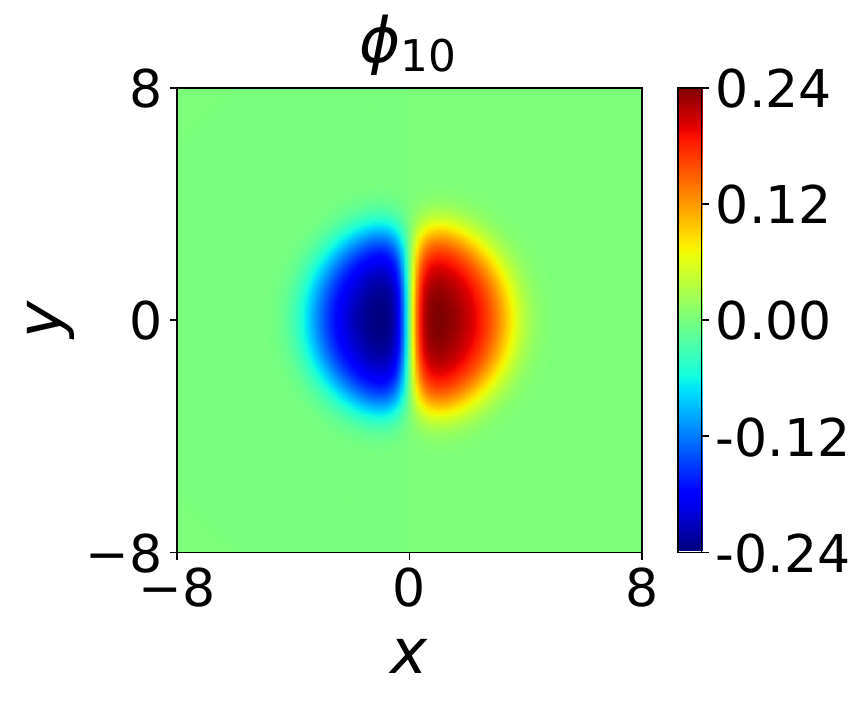}
			\includegraphics[width=0.22\linewidth]{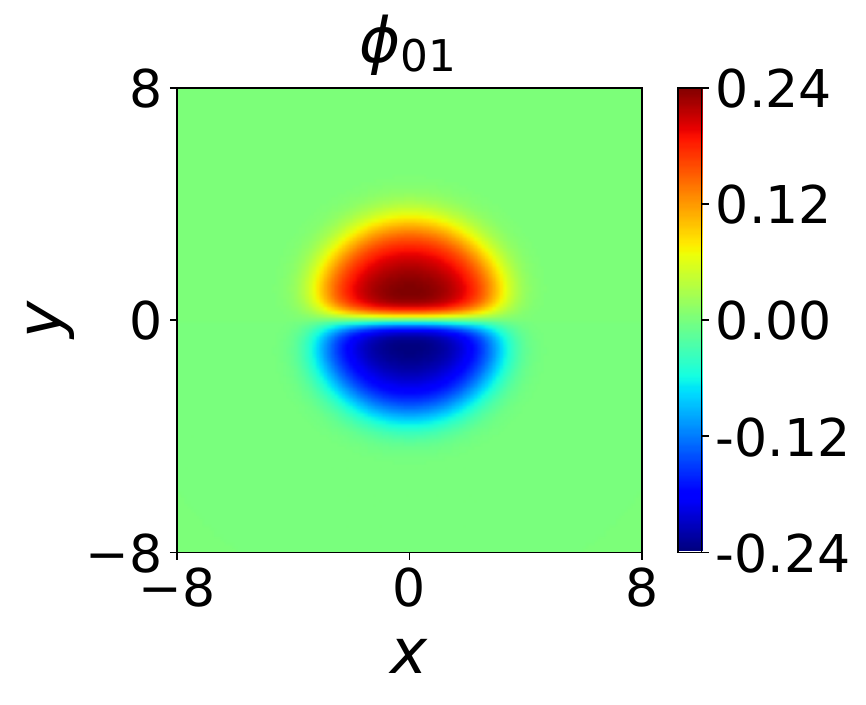}
			\includegraphics[width=0.22\linewidth]{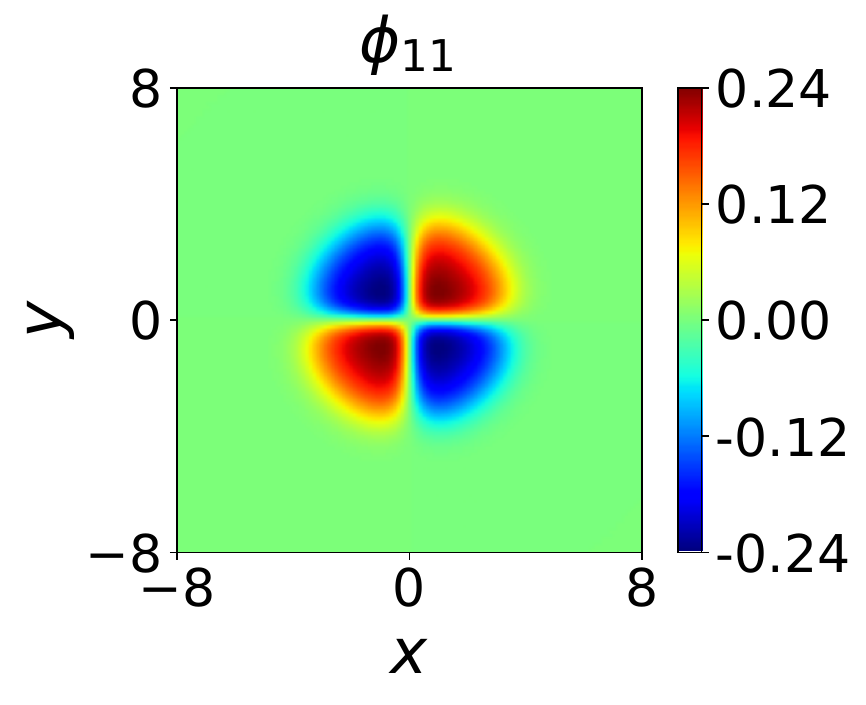}\\
			\includegraphics[width=0.22\linewidth]{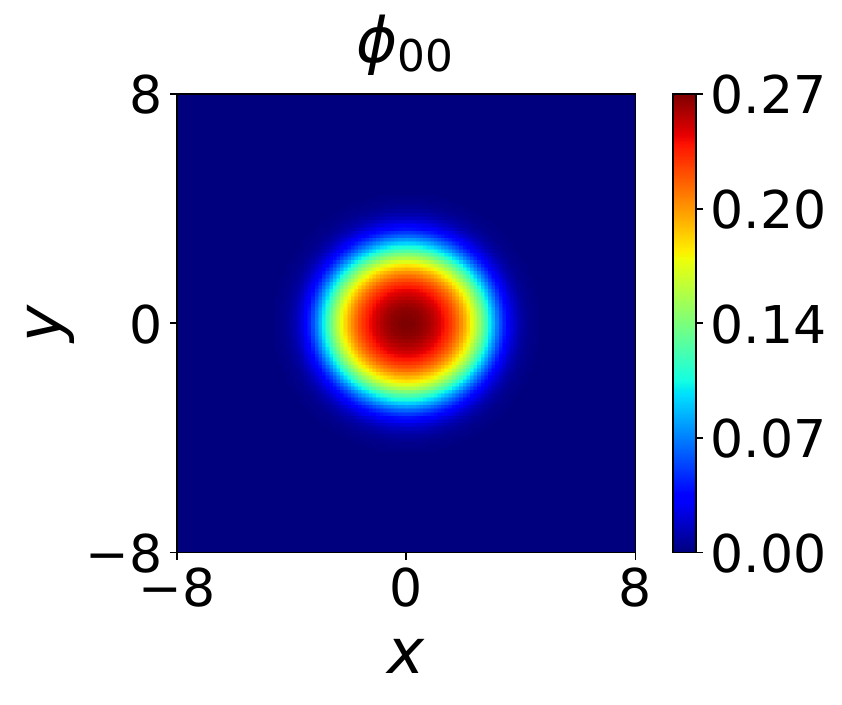}
			\includegraphics[width=0.22\linewidth]{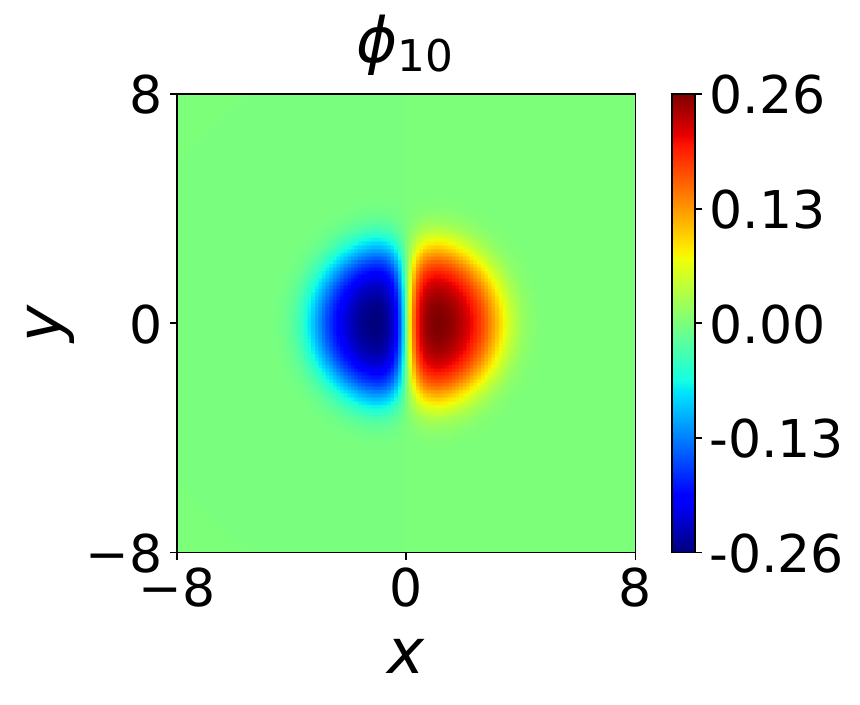}
			\includegraphics[width=0.22\linewidth]{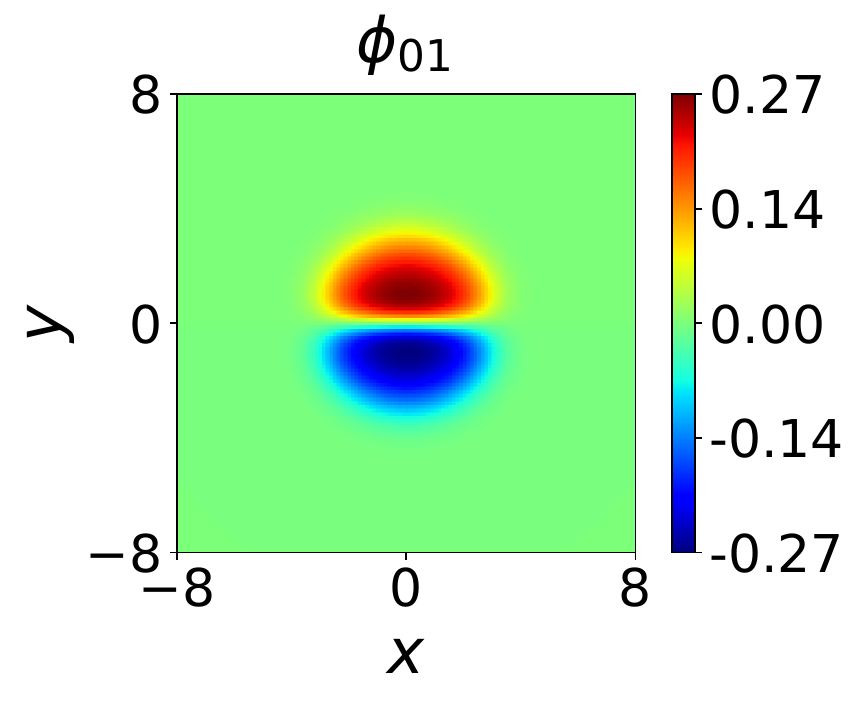}
			\includegraphics[width=0.22\linewidth]{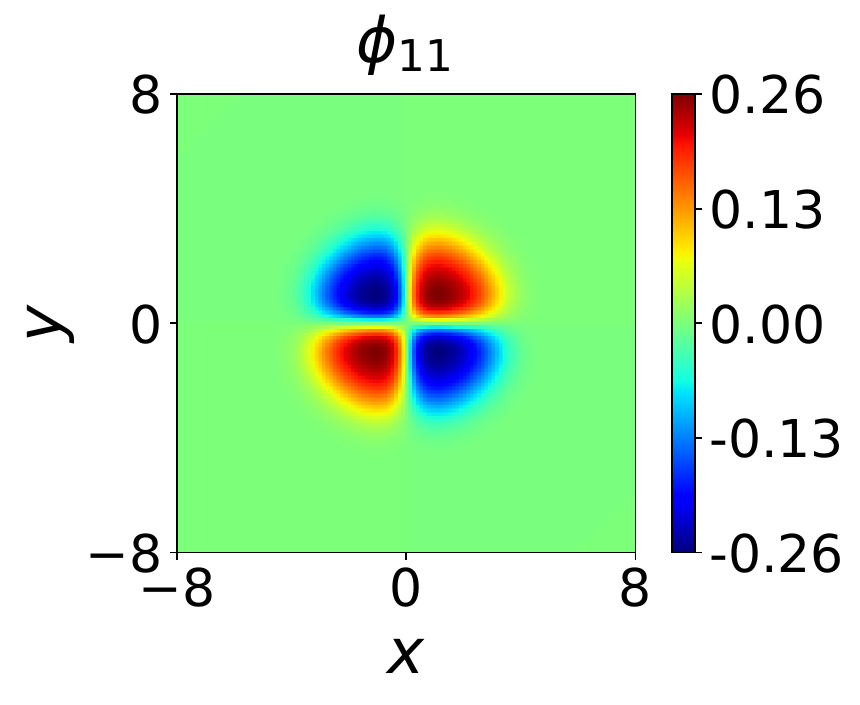}\\
			\includegraphics[width=0.22\linewidth]{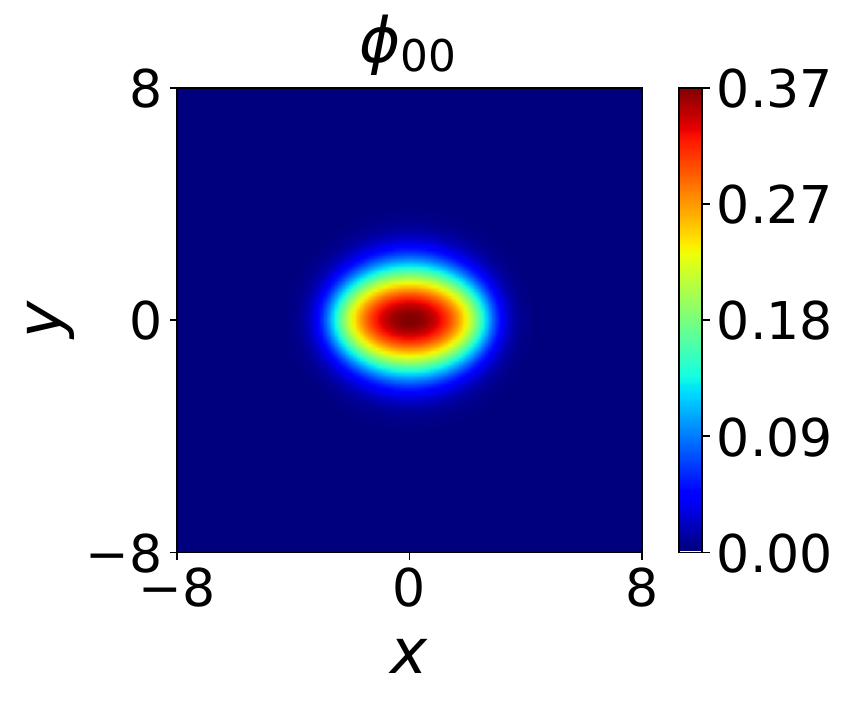}
			\includegraphics[width=0.22\linewidth]{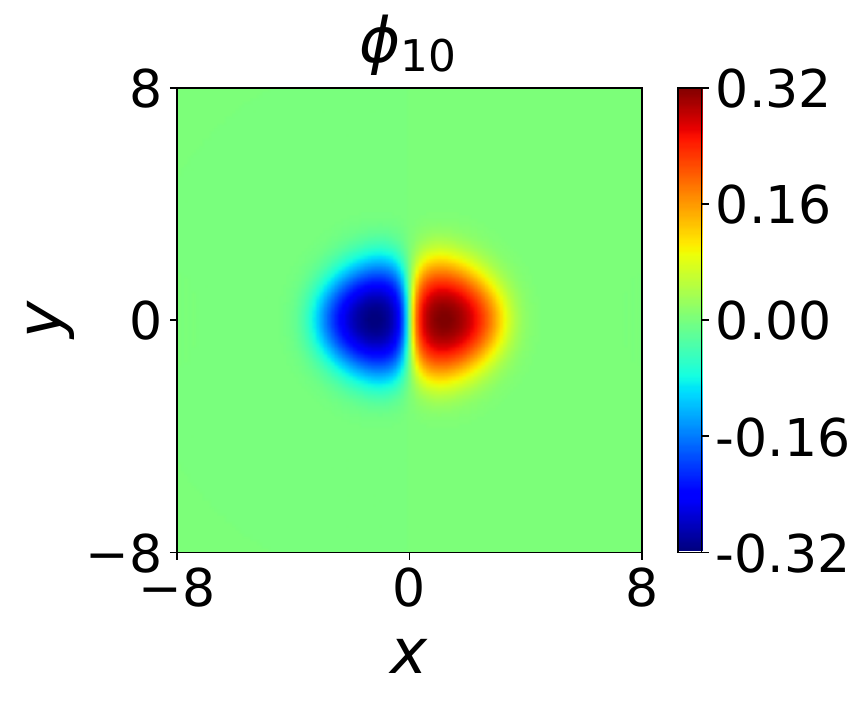}
			\includegraphics[width=0.22\linewidth]{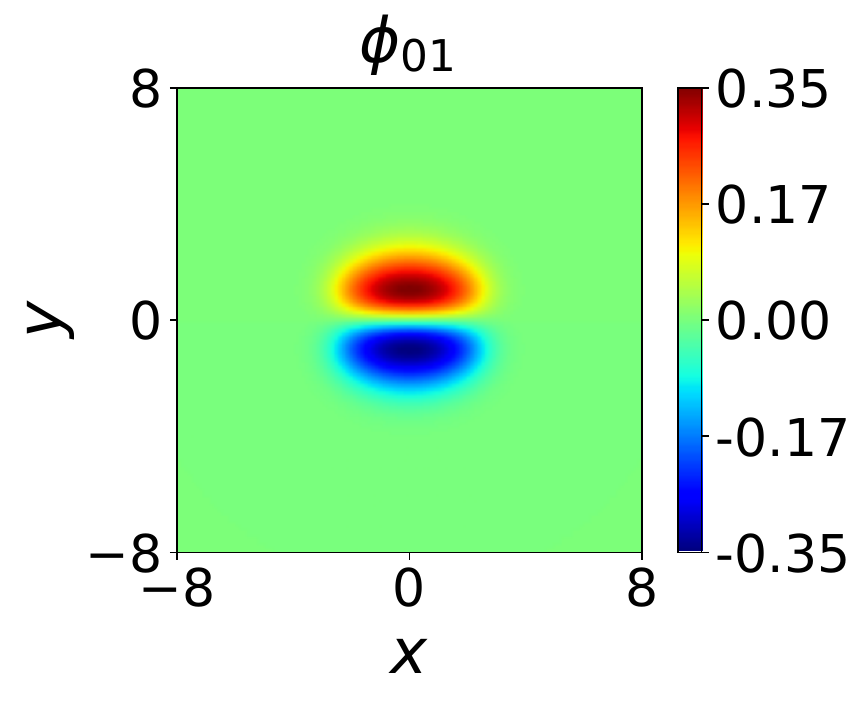}
			\includegraphics[width=0.22\linewidth]{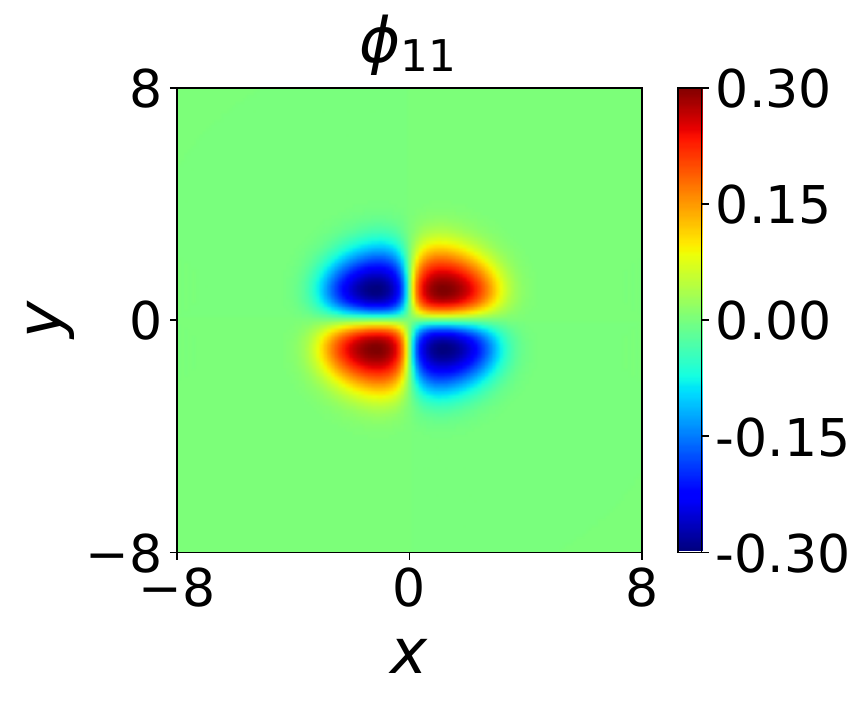}\\
			\caption{$\lambda_2 > 0$ and $\mathbf{n} = (1,0,0)$. The pseudo-color of the states $\phi_{00}$, $\phi_{10}$, $\phi_{01}$, and $\phi_{11}$ (from left to right) for different interaction parameter ratios. The first row  corresponds to the cases of $\lambda_2/\beta_2 = 0$, the second row corresponds to the cases of $\lambda_2/\beta_2 = 0.5$ and the last row  corresponds to the cases of $\lambda_2/\beta_2 = 1$.}
			\label{fig:wavefunctions_beta_lambda}
		\end{figure}
		From the numerical results that are partially shown in Fig.~\ref{fig:wavefunctions_beta_lambda} and Table \ref{tab:2d_excited_states}, we have the following numerical observations:
		\begin{enumerate}
			\item In the case of purely contact interactions ($\lambda_2=0$), the excited states $\phi_{10}$ and $\phi_{01}$ are exactly degenerate, a direct consequence of the rotational symmetry of the isotropic harmonic trapping potential. However, when $\lambda_2>0$, due to the long-range anisotropy of the DDI, the energies of the excited states $\phi_{10}$ and $\phi_{01}$ are no longer the same.
            \item Table~\ref{tab:2d_excited_states} shows the facts $E(\phi_{00}) < E(\phi_{01}) < E(\phi_{10}) < E(\phi_{11})$ and $\mu(\phi_{00}) < \mu(\phi_{01}) < \mu(\phi_{10}) < \mu(\phi_{11})$ for $\lambda_2/\beta_2 > 0$. It demonstrates that, for $\lambda_2 > 0$, the energy satisfies $E(\phi_{10}) > E(\phi_{01})$.
            \item As the relative dipole strength $\lambda_2/\beta_2$ increases from $0$ to $1$, the energy and chemical potential of the same type of excited states decrease monotonically. The profiles of wave functions are stretched along the $x$-direction and compressed along the $y$-direction.
		\end{enumerate}
	\end{example}

\begin{table}[!t]
			\centering\footnotesize
			\caption{Energy $E$ and chemical potential $\mu$ of the 2D ground and excited states for different negative ratios $\lambda_2/\beta_2$ with $\beta_2=150$ and polarization $\mathbf{n} = (\frac{1}{\sqrt{2}},0,\frac{1}{\sqrt{2}})$.}
			\label{tab:2d_negative_lambda}
			\begin{tabular}{c c c c c c c c c}
				\toprule
				\multirow{2}{*}{$\lambda_2/\beta_2$} & 
				\multicolumn{2}{c}{$\phi_{00}$} & \multicolumn{2}{c}{$\phi_{10}$} & \multicolumn{2}{c}{$\phi_{01}$} & \multicolumn{2}{c}{$\phi_{11}$} \\
				\cmidrule(lr){2-3} \cmidrule(lr){4-5} \cmidrule(lr){6-7} \cmidrule(lr){8-9}
				& $E$ & $\mu$ & $E$ & $\mu$ & $E$ & $\mu$ & $E$ & $\mu$ \\
				\midrule
				0.00 & 4.7659 & 7.0106 & 5.5394 & 7.7537 & 5.5394 & 7.7537 & 6.3131 & 8.4962 \\
				-0.50 & 4.2409 & 6.1937 & 5.0191 & 6.9376 & 5.0844 & 7.0251 & 5.8613 & 7.7670 \\
				-1.00 & 3.6489 & 5.2646 & 4.4341 & 6.0095 & 4.5838 & 6.2149 & 5.3656 & 6.9558 \\
				-1.50 & 2.9532 & 4.1563 & 3.7520 & 4.9041 & 4.0178 & 5.2848 & 4.8092 & 6.0258 \\
				-2.00 & 2.0482 & 2.6562 & 2.8970 & 3.4466 & 3.3395 & 4.1353 & 4.1587 & 4.8938 \\
				\bottomrule
			\end{tabular}
		\end{table}

	\begin{figure}[!t]
		\centering
		\includegraphics[width=0.22\linewidth]{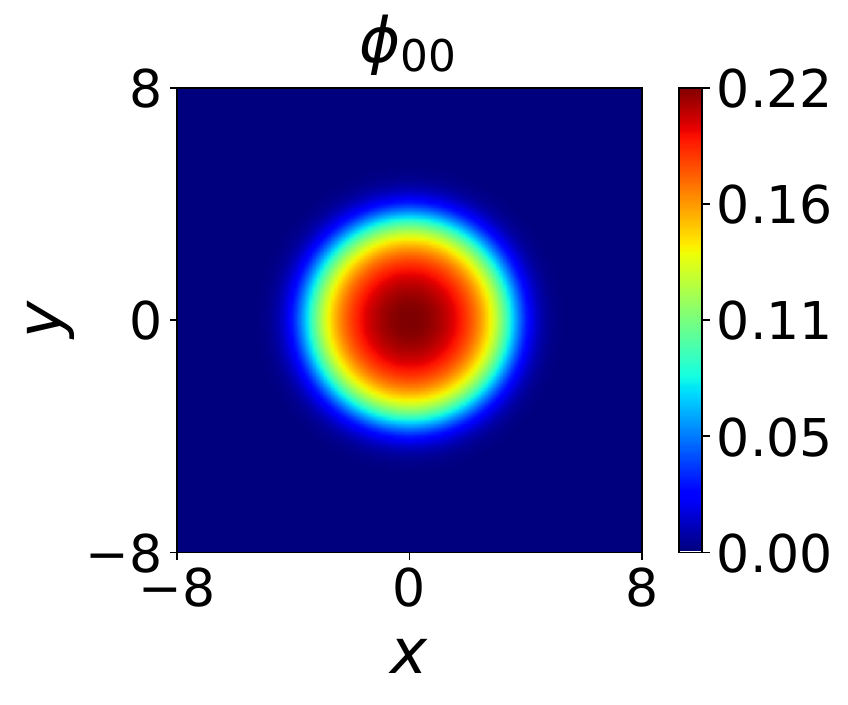}
		\includegraphics[width=0.22\linewidth]{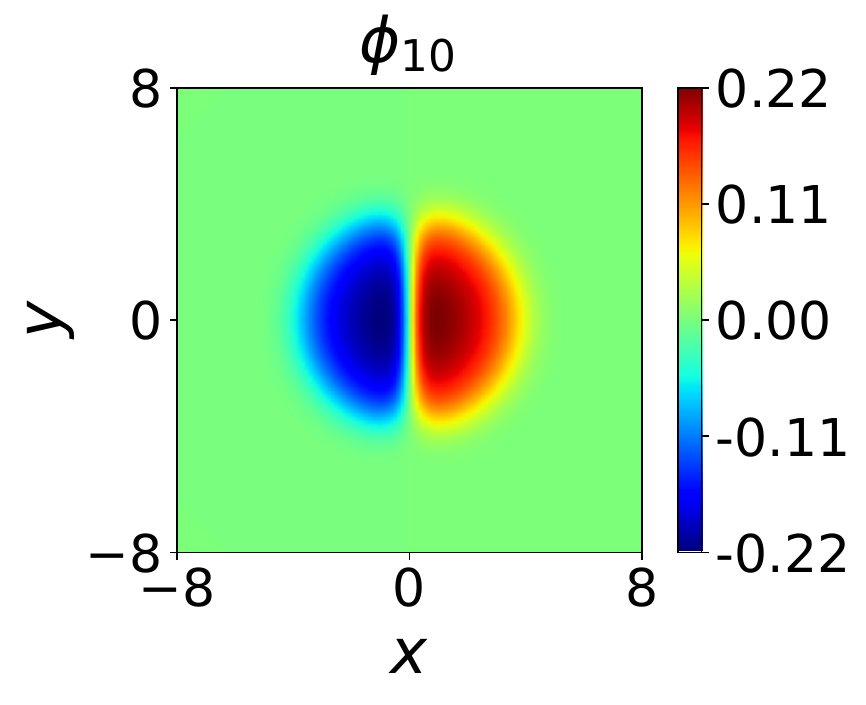}
		\includegraphics[width=0.22\linewidth]{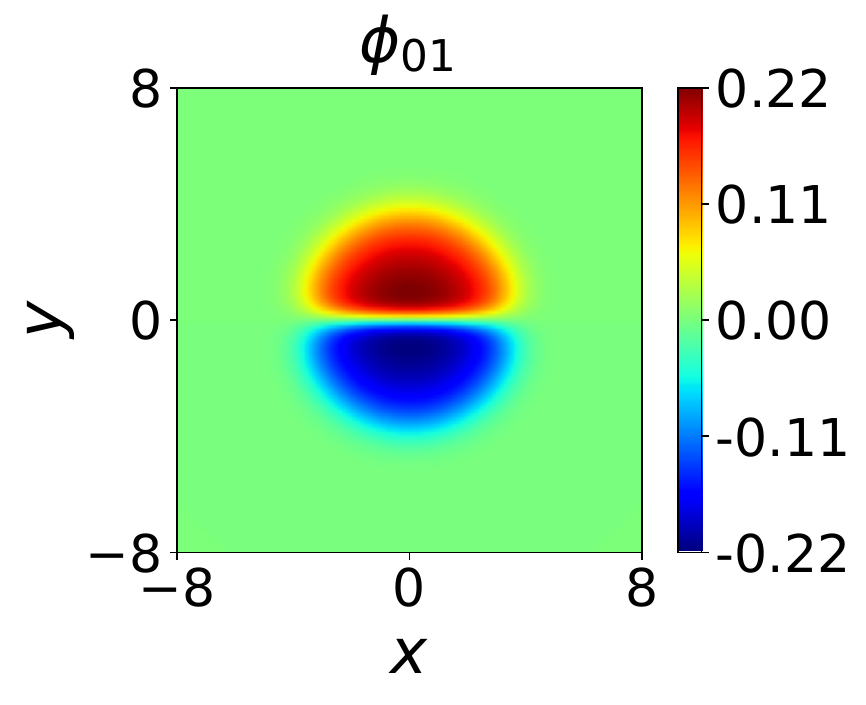}
		\includegraphics[width=0.22\linewidth]{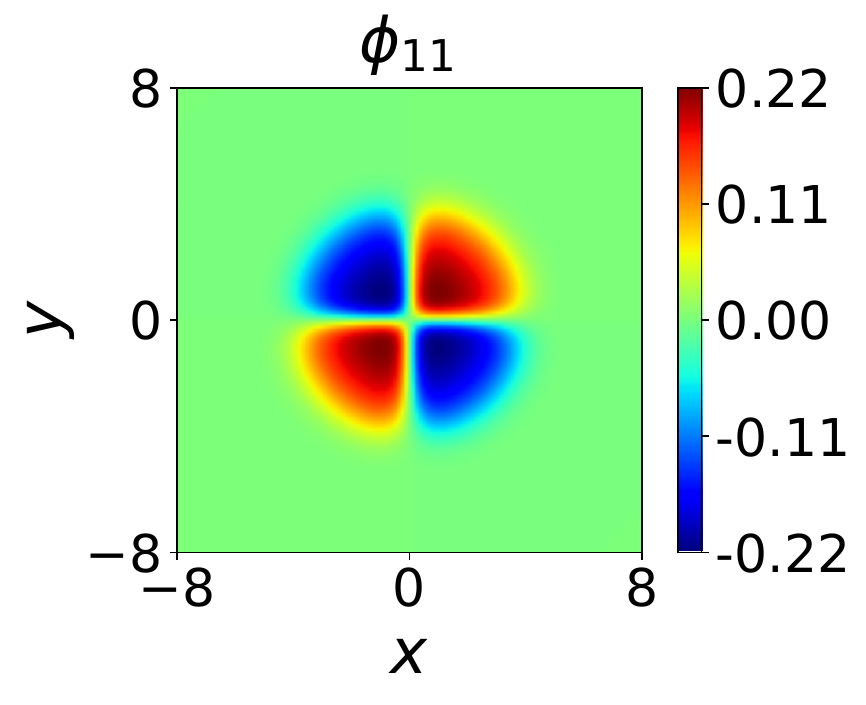}\\
		\includegraphics[width=0.22\linewidth]{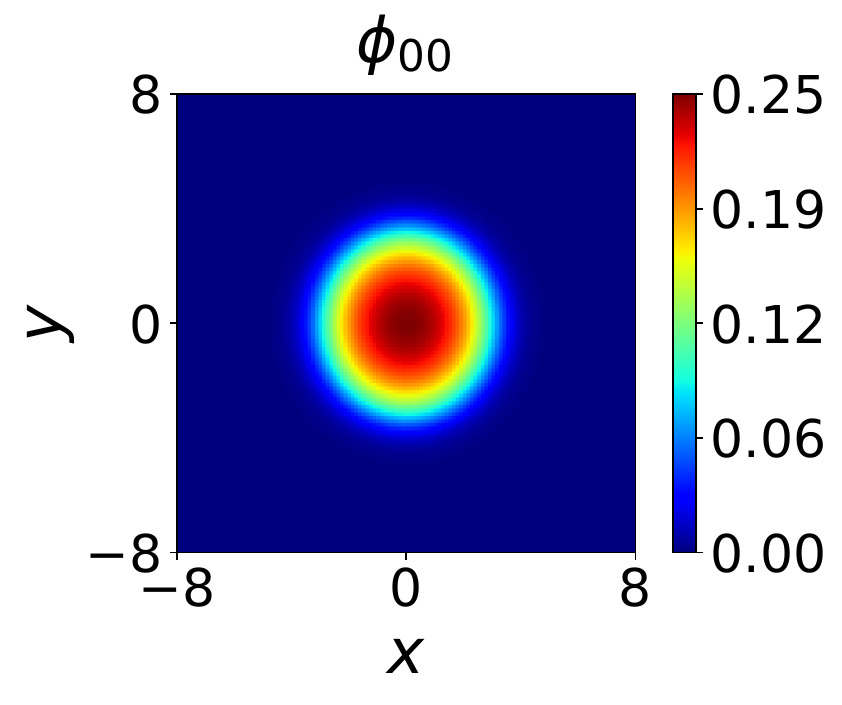}
		\includegraphics[width=0.22\linewidth]{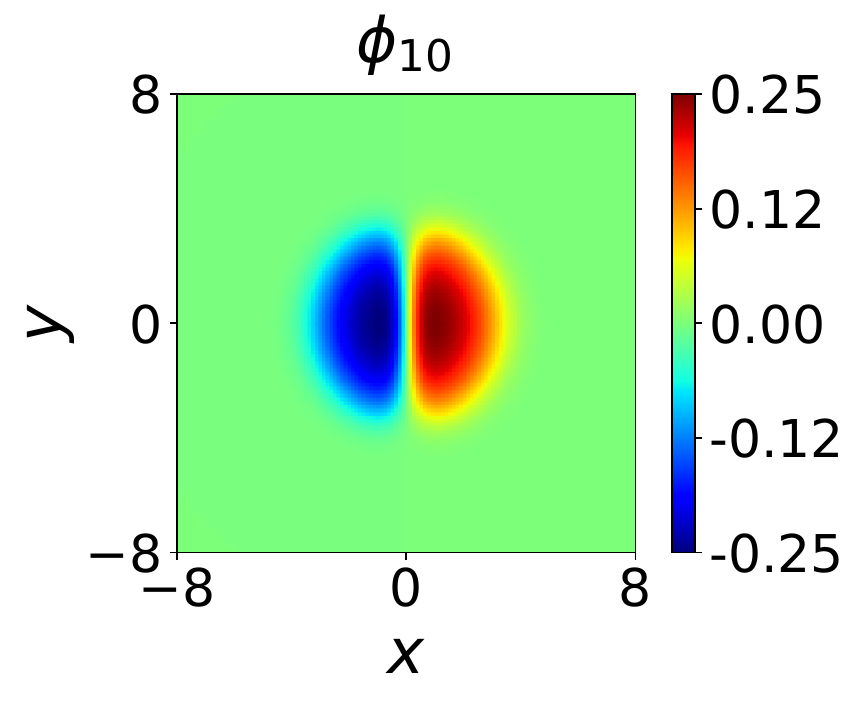}
		\includegraphics[width=0.22\linewidth]{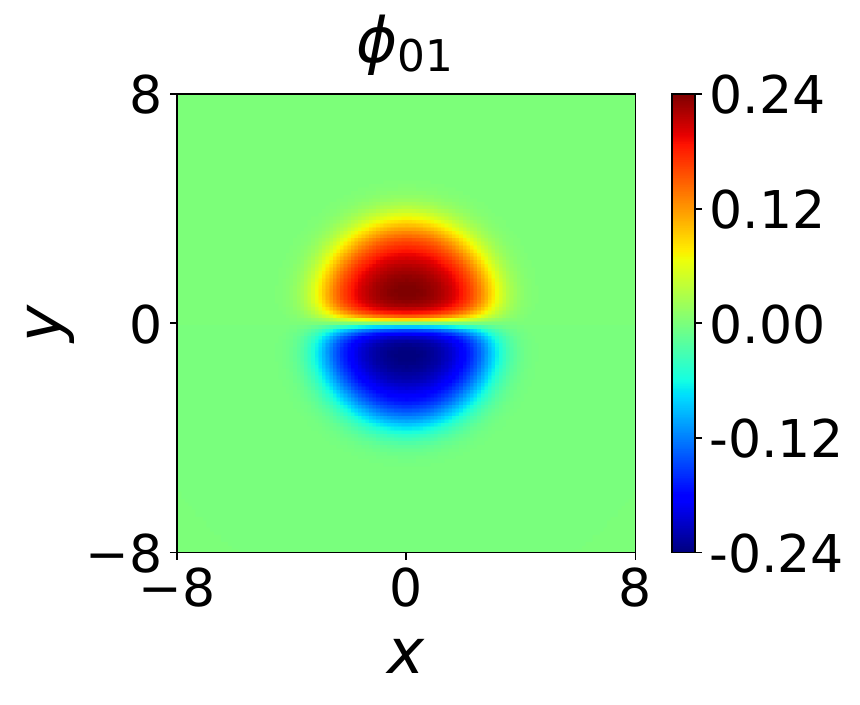}
		\includegraphics[width=0.22\linewidth]{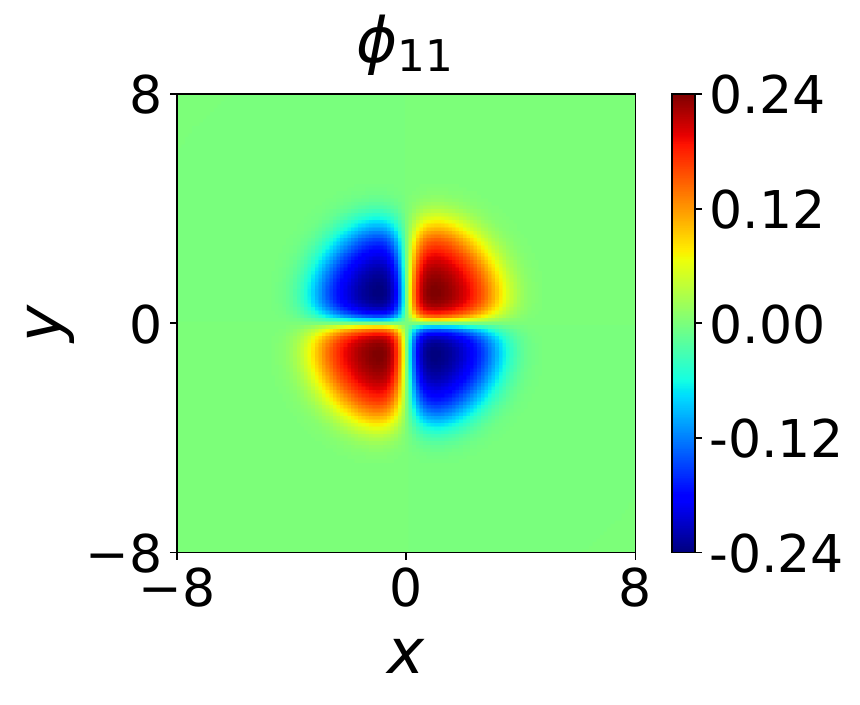}\\
		\includegraphics[width=0.22\linewidth]{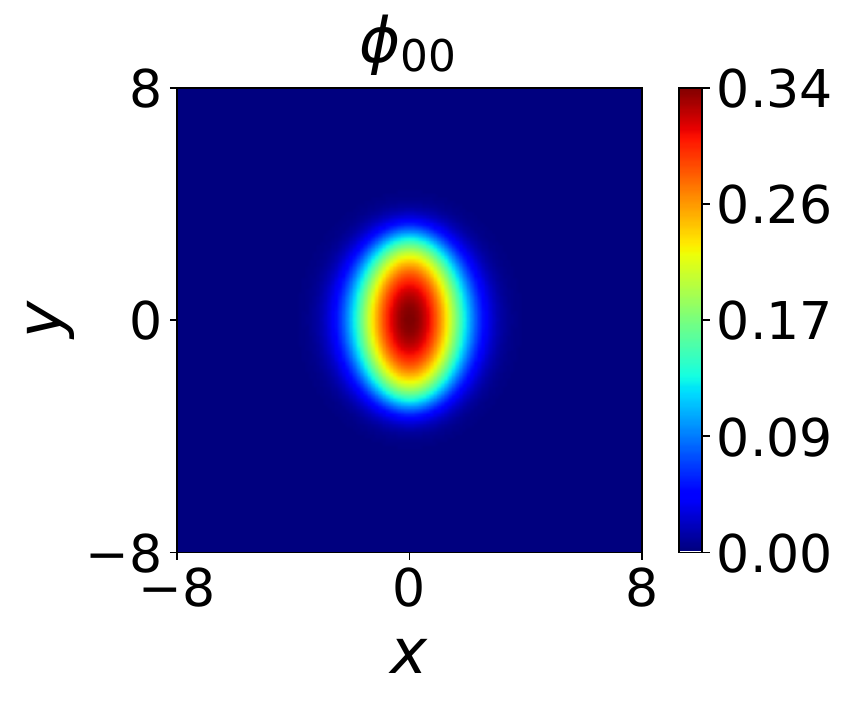}
		\includegraphics[width=0.22\linewidth]{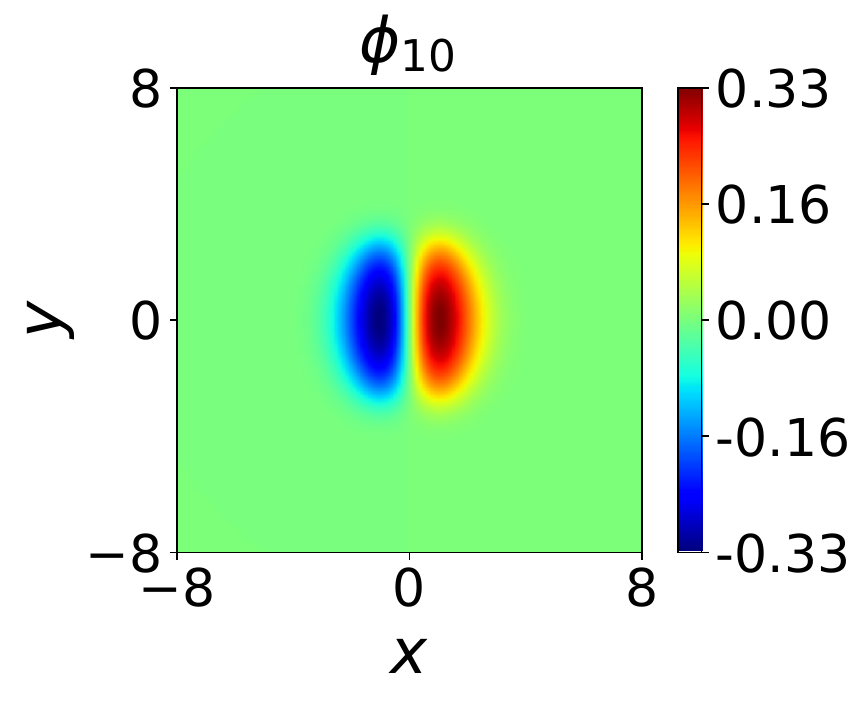}
		\includegraphics[width=0.22\linewidth]{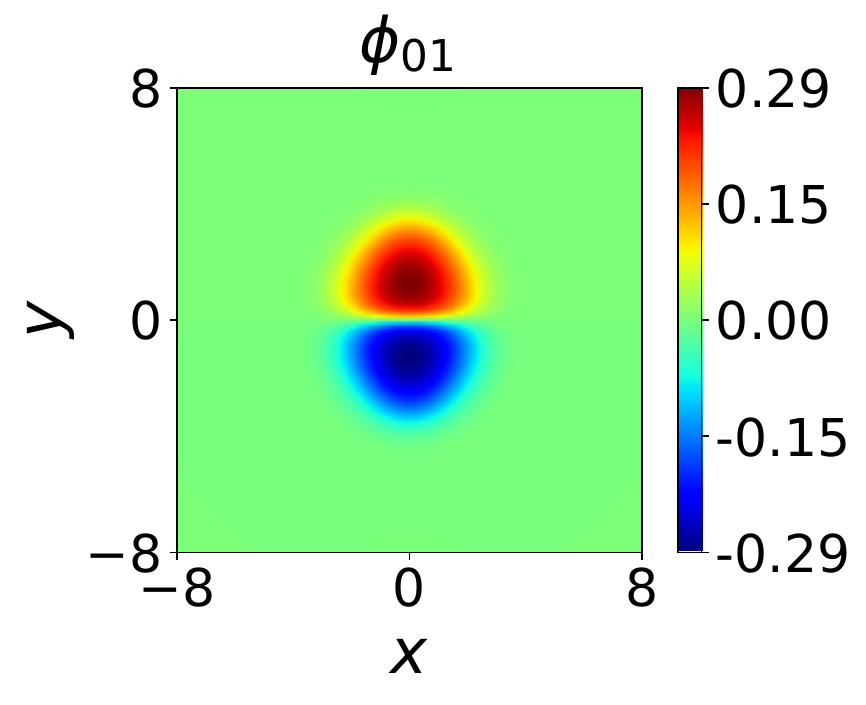}
		\includegraphics[width=0.22\linewidth]{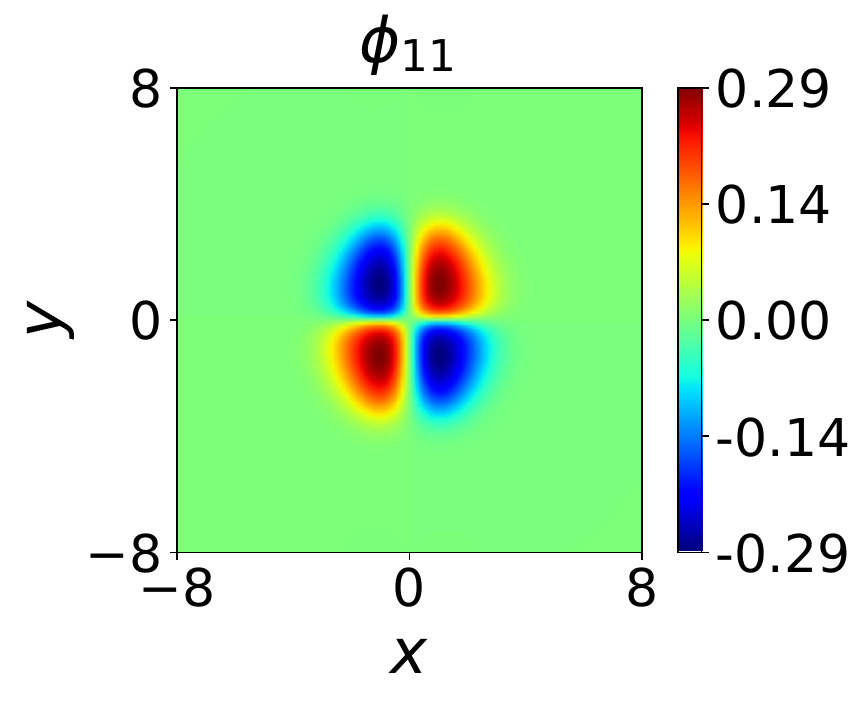}\\
		\caption{$\lambda_2 < 0$ and $\mathbf{n}  = (\frac{1}{\sqrt{2}},0,\frac{1}{\sqrt{2}})$. The pseudo-color of the states $\phi_{00}$, $\phi_{10}$, $\phi_{01}$, and $\phi_{11}$ (from left to right) for different interaction parameter ratios. The first row  corresponds to the cases of $\lambda_2/\beta_2 = 0$, the second row corresponds to the cases of $\lambda_2/\beta_2 = -1$ and the last row  corresponds to the cases of $\lambda_2/\beta_2 = -2$.}
		\label{fig:wavefunctions_comparison}
	\end{figure}

		\begin{figure}[!t]
			\centering
			\includegraphics[width=0.75\linewidth]{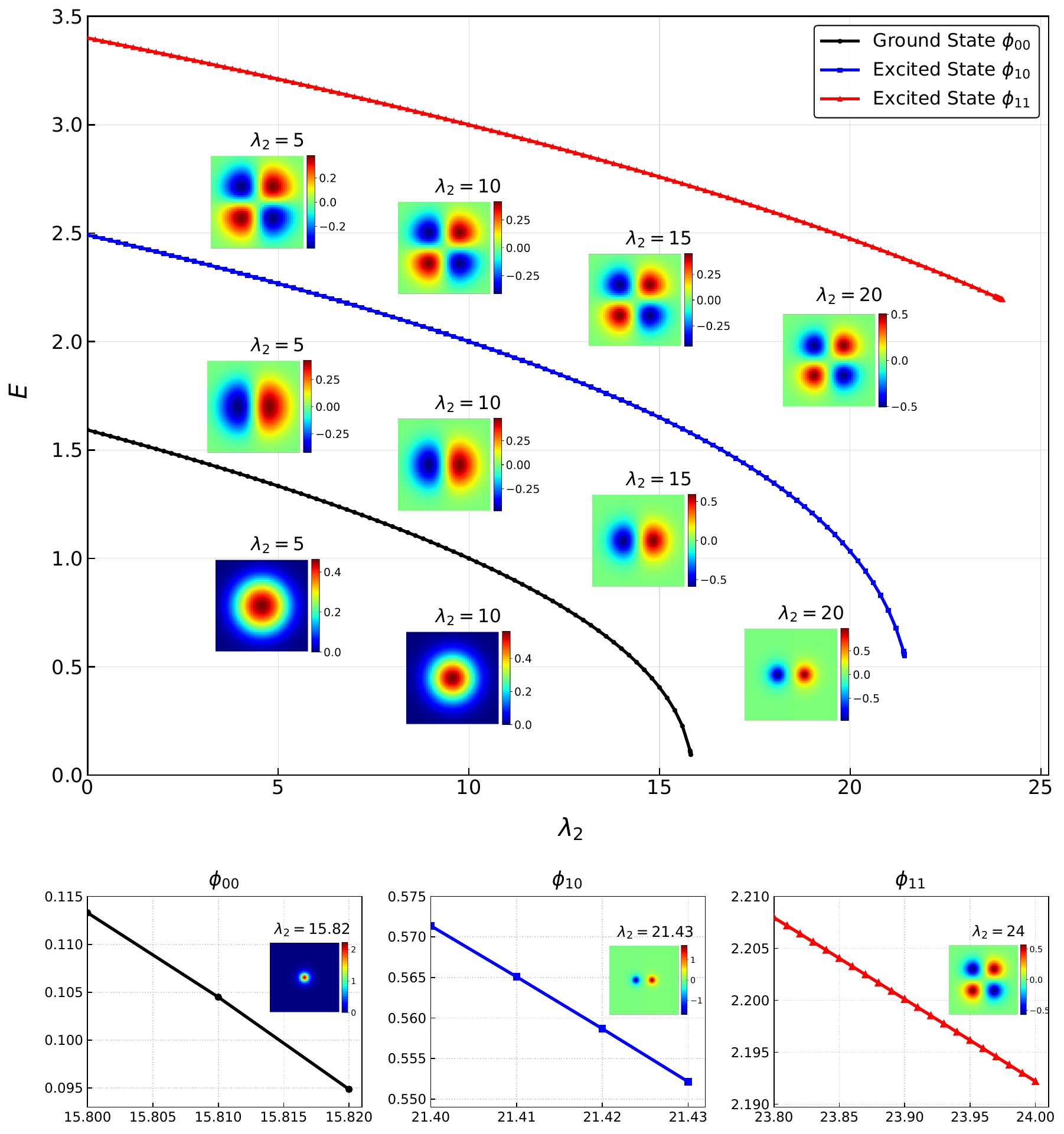}
			\caption{For case I of Example~\ref{ex:existence}, the energies of $\phi_{00}$, $\phi_{10}$, and $\phi_{11}$ as functions of increasing $\lambda_2$. 
				The lower subplots show the critical $\lambda_2$ values at which the ground and excited states are about to collapse.}\label{fig:existence-experiment-cases1}
		\end{figure}

		\begin{figure*}[!t]
			\centering
			\includegraphics[width=0.75\linewidth]{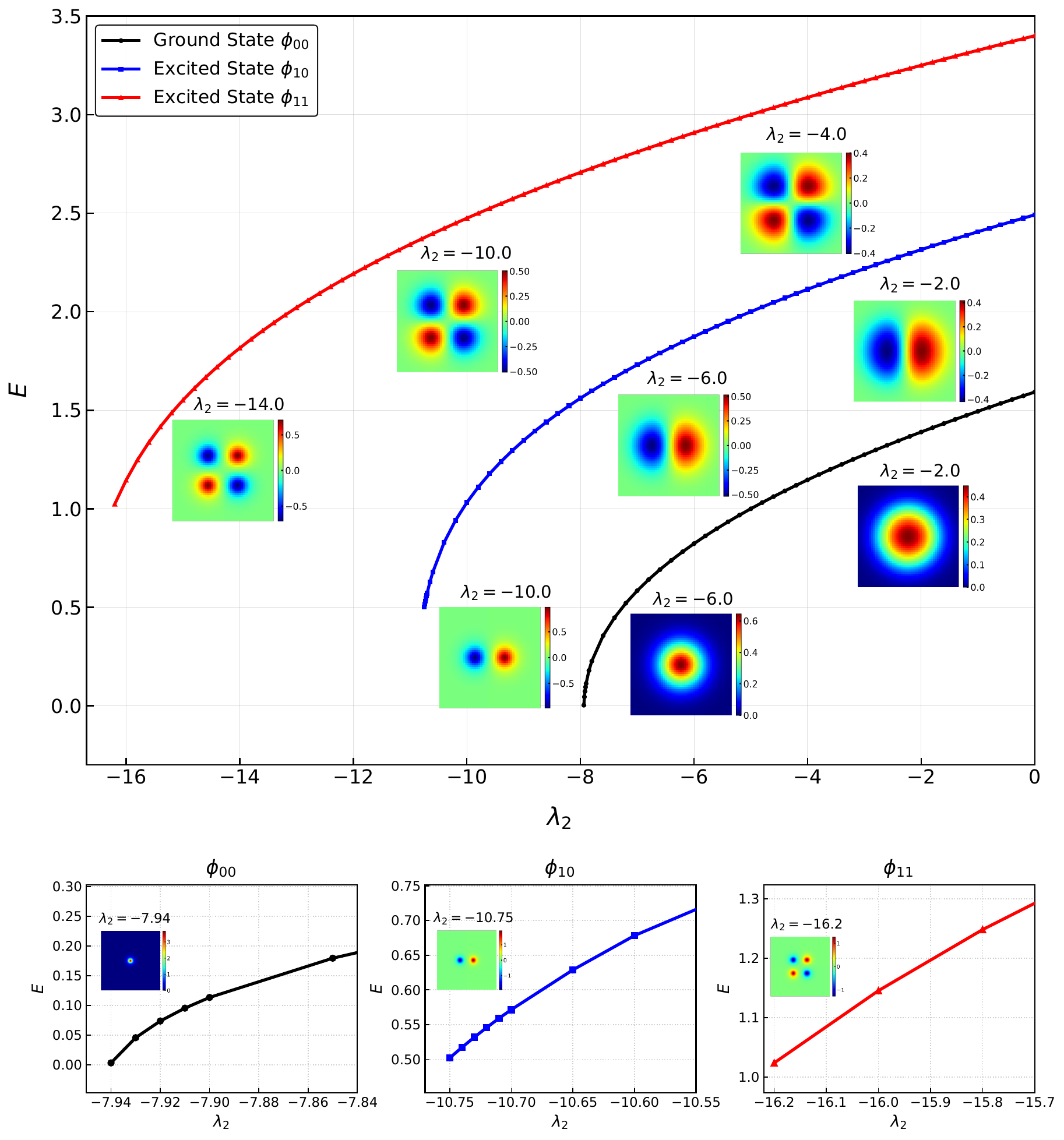}
			\caption{For case II of Example~\ref{ex:existence}, the energies of $\phi_{00}$, $\phi_{10}$, and $\phi_{11}$ as functions of decreasing $\lambda_2$. 
				The lower subplots show the critical $\lambda_2$ values at which the ground and excited states are about to collapse.}\label{fig:existence-experiment-cases2}
		\end{figure*}

\begin{example}
To investigate the impact of a negative dipolar interaction parameter $(\lambda_2 < 0)$ on the properties of both the ground and excited states, we fix the contact interaction strength at a strongly repulsive value of $\beta_2 = 150$, and set the dipole polarization vector strictly in the $x$-$z$ plane, i.e., $\mathbf{n} = \left(\frac{1}{\sqrt{2}},0,\frac{1}{\sqrt{2}}\right)$, $\varepsilon=\sqrt{\frac{1}{80}}$. Table \ref{tab:2d_negative_lambda} tabulates the energies and chemical potentials of the ground state and the excited states for varying relative dipolar interaction strengths $\lambda_2/\beta_2 \in [-2, 0]$. Fig.~\ref{fig:wavefunctions_comparison} plots the pseudo-color images of excited states with $\lambda_2/\beta_2$ from 0 to $-2$.

From the numerical results presented in Fig.~\ref{fig:wavefunctions_comparison} and Table~\ref{tab:2d_negative_lambda}, we draw the following conclusions regarding the impact of negative dipolar interactions $(\lambda_2 < 0)$:
		
	\begin{enumerate}

    \item As shown in Fig.~\ref{fig:wavefunctions_comparison}, a negative $\lambda_2$ together with the dipole orientation $\mathbf{n}$ reverses the DDI; the in-plane $x$-axis is repulsive and the $y$-axis is attractive. The profiles of wave functions are compressed along the $x$-direction and stretched along the $y$-direction.

    \item Table~\ref{tab:2d_negative_lambda} gives the orderings $E(\phi_{00}) < E(\phi_{10}) < E(\phi_{01}) < E(\phi_{11})$ and $\mu(\phi_{00}) < \mu(\phi_{10}) < \mu(\phi_{01}) < \mu(\phi_{11})$ for $\lambda_2/\beta_2 < 0$. This indicates that, for $\lambda_2 < 0$, the energy satisfies $E(\phi_{10}) < E(\phi_{01})$, which is exactly the opposite of the case when $\lambda_2 > 0$ (see Table~\ref{tab:2d_excited_states}).

    \item For the same type of excited states, both the energy and the chemical potential decrease as the ratio $\lambda_2/\beta_2$ decreases from $0$ to $-2.00$. This indicates that the negative nonlocal dipolar interaction brings an overall attractive effect, thus shifting the entire energy spectrum downward.
	\end{enumerate}
\end{example}

\begin{example}\label{ex:existence}
During our numerical simulations, we observed a remarkable phenomenon: in certain parameter regimes where theoretical analysis predicts the absence of a stable ground state, the excited states of the condensate can still persist. To systematically illustrate this, we fix $\beta_2=10$ and $\varepsilon=0.0001$ in the following two cases. Numerical results are shown in Figs.~\ref{fig:existence-experiment-cases1}-\ref{fig:existence-experiment-cases2}.
\begin{itemize}
    \item[] {\bf Case I. $\mathbf{n} = (1,0,0)$,  $\lambda_2>0$.} 
\end{itemize}

In case I, for the ground state, the theoretical existence condition is $\beta_2 - \lambda_2 > -C_b$ \cite{bao2012gross}, where \(C_b\) is the best constant in the Gagliardo--Nirenberg inequality \cite{weinstein1982nonlinear}. As shown in Fig.~\ref{fig:existence-experiment-cases1}, the numerical results suggest that the ground state exists when $\lambda_2 < 15.82=\beta_2 + 5.82\approx \beta_2+C_b$, verifying the existence theory of ground state. In fact, in our numerical experiments, this manifests as the energy converging to a negative value within very few iterations, and this value continues to decrease under mesh refinement. Surprisingly, for excited states, their existence intervals are larger than that of the ground state. Specifically, as shown in Fig.~\ref{fig:existence-experiment-cases1}, when \(\lambda_2 > 15.82\) the ground state can no longer be computed numerically, while the excited states \(\phi_{10}\) and \(\phi_{11}\) can still be obtained. When \(\lambda_2 > 21.43=\beta_2+11.43\) the algorithm for computing \(\phi_{10}\) fails, numerically reflected by the non‑convergence of the residual \(F_{\phi}^n\), but \(\phi_{11}\) can still be computed numerically, even as \(\lambda_2\) further increases to \(24\).

\begin{itemize}
\item[]  
{\bf Case II. $\mathbf{n} = (0,0,1)$, $\lambda_2 < 0$.} 
\end{itemize}

In case II, the theoretical existence condition for the ground state is \(\lambda_2 > \frac{-\beta_2 - C_b}{2}\) \cite{bao2012gross}. As shown in Fig.~\ref{fig:existence-experiment-cases2}, the numerical results suggest that the ground state exists when \(\lambda_2 > -7.94 = \frac{-\beta_2 - 5.88}{2} \approx \frac{-\beta_2 - C_b}{2}\), verifying the existence theory of ground state. The existence intervals for excited states are remarkably larger than that of the ground state. Specifically, as shown in Fig.~\ref{fig:existence-experiment-cases2}, when \(\lambda_2 <-7.94\) the ground state can no longer be computed numerically, while the excited states \(\phi_{10}\) and \(\phi_{11}\) can still be obtained. When \(\lambda_2 < -10.75=\frac{-\beta_2 - 11.5}{2}\) the algorithm for computing \(\phi_{10}\) fails, numerically reflected by the non‑convergence of the residual \(F_{\phi}^n\), but \(\phi_{11}\) can still be computed numerically, even as \(\lambda_2\) further decreases to \(-16.2\).
\end{example}

\subsection{Numerical results in 3D} \label{3D_example}
In this section, we consider the 3D examples \((d=3)\) under an anisotropic trapping in \eqref{V} with $\gamma_x=\gamma_y=1$ and $\gamma_z=0.5$ and the polarization direction \(\mathbf{n}=(0,0,1)\).  The computational domain is taken as \([-8,8]^3\) and the initial data is shown in Table~\ref{tab:initial_3d}. We investigate the  excited states for different relative dipolar interaction strengths \(\lambda_3/\beta_3\) with \(\beta_3=207.16\).

\begin{table}[!t]
\centering
\caption{Initial data for the numerical experiments in Subsection~\ref{3D_example}.}
\label{tab:initial_3d}
\begin{tabular}{c c l l}
\toprule
\!\!Solution\! & Morse index $k$ & Initial guess of $\phi^0$ & Initial guess of $\{v_i^0\}$ \\
\midrule
$\phi_{000}$ & 0 & $\phi_{000}^{\mathrm{ho}}$ & -- \\

$\phi_{100}$ & 1 & $\phi_{100}^{\mathrm{ho}}$
& $\{\phi_{000}^{\mathrm{ho}}\}$ \\

$\phi_{001}$ & 1 & $\phi_{001}^{\mathrm{ho}}$
& $\{\phi_{000}^{\mathrm{ho}}\}$ \\

$\phi_{110}$ & 4 & $\phi_{110}^{\mathrm{ho}}$
& $\{\phi_{000}^{\mathrm{ho}},\phi_{100}^{\mathrm{ho}},
\phi_{010}^{\mathrm{ho}},\phi_{001}^{\mathrm{ho}}\}$ \\

$\phi_{101}$ & 4 & $\phi_{101}^{\mathrm{ho}}$
& $\{\phi_{000}^{\mathrm{ho}},\phi_{100}^{\mathrm{ho}},
\phi_{010}^{\mathrm{ho}},\phi_{001}^{\mathrm{ho}}\}$ \\

$\phi_{111}$ & 10 & $\phi_{111}^{\mathrm{ho}}$
& $\{\phi_{000}^{\mathrm{ho}},\phi_{100}^{\mathrm{ho}},
\phi_{010}^{\mathrm{ho}},\phi_{001}^{\mathrm{ho}},
\phi_{200}^{\mathrm{ho}},\phi_{020}^{\mathrm{ho}},
\phi_{002}^{\mathrm{ho}},\phi_{110}^{\mathrm{ho}},
\phi_{101}^{\mathrm{ho}},\phi_{011}^{\mathrm{ho}}\}$\!\! \\
\bottomrule
\end{tabular}
\end{table}

	\begin{table}[!t]
		\centering\footnotesize
		\caption{Energy $E$ and chemical potential $\mu$ of the ground state and various excited states for different ratios of $\lambda_3/\beta_3$ in 3D.}
		\label{tab:combined_data}
		\resizebox{0.99\textwidth}{!}{%
			\begin{tabular}{c c c c c c c c c c c c c}
				\toprule
				\multirow{2}{*}{$\lambda_3/\beta_3$} & 
				\multicolumn{2}{c}{$\phi_{000}$} & \multicolumn{2}{c}{$\phi_{100}$} & \multicolumn{2}{c}{$\phi_{001}$} & \multicolumn{2}{c}{$\phi_{110}$} & \multicolumn{2}{c}{$\phi_{101}$} & \multicolumn{2}{c}{$\phi_{111}$} \\
				\cmidrule(lr){2-3} \cmidrule(lr){4-5} \cmidrule(lr){6-7} \cmidrule(lr){8-9} \cmidrule(lr){10-11} \cmidrule(lr){12-13}
				& $E$ & $\mu$ & $E$ & $\mu$ & $E$ & $\mu$ & $E$ & $\mu$ & $E$ & $\mu$ & $E$ & $\mu$ \\
				\midrule
				-0.50 & 2.9584 & 3.9301 & 3.8095 & 4.7481 & 3.2667 & 4.1896 & 4.6605 & 5.5635 & 4.1318 & 5.0233 & 4.9947 & 5.8513 \\
				-0.25 & 2.8841 & 3.8187 & 3.7315 & 4.6293 & 3.2386 & 4.1444 & 4.5807 & 5.4390 & 4.0921 & 4.9601 & 4.9471 & 5.7742 \\
				0.00 & 2.7943 & 3.6830 & 3.6406 & 4.4891 & 3.1886 & 4.0615 & 4.4904 & 5.2958 & 4.0361 & 4.8674 & 4.8872 & 5.6736 \\
				0.25 & 2.6875 & 3.5201 & 3.5353 & 4.3242 & 3.1201 & 3.9480 & 4.3883 & 5.1306 & 3.9651 & 4.7482 & 4.8156 & 5.5506 \\
				0.50 & 2.5593 & 3.3213 & 3.4120 & 4.1272 & 3.0323 & 3.8015 & 4.2719 & 4.9377 & 3.8782 & 4.6000 & 4.7314 & 5.4027 \\
				0.75 & 2.3998 & 3.0674 & 3.2635 & 3.8825 & 2.9206 & 3.6118 & 4.1371 & 4.7066 & 3.7719 & 4.4143 & 4.6327 & 5.2242 \\
				1.00 & 2.1838 & 2.7011 & 3.0743 & 3.5517 & 2.7729 & 3.3497 & 3.9764 & 4.4167 & 3.6386 & 4.1708 & 4.5157 & 5.0039 \\
				\bottomrule
			\end{tabular}%
		}
	\end{table}
	\begin{figure}[!t]
		\centering
		\includegraphics[width=0.6\textwidth]{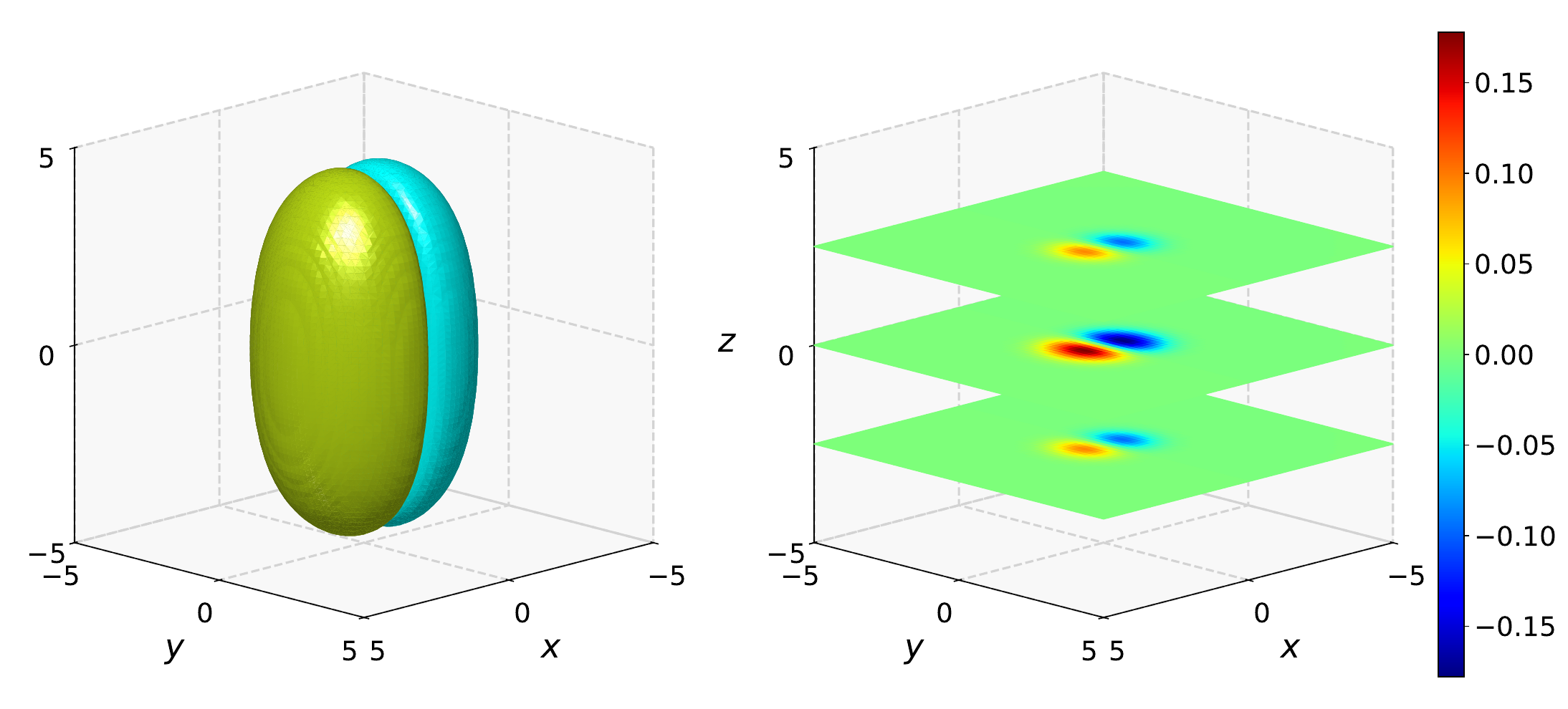}\\
		\includegraphics[width=0.6\textwidth]{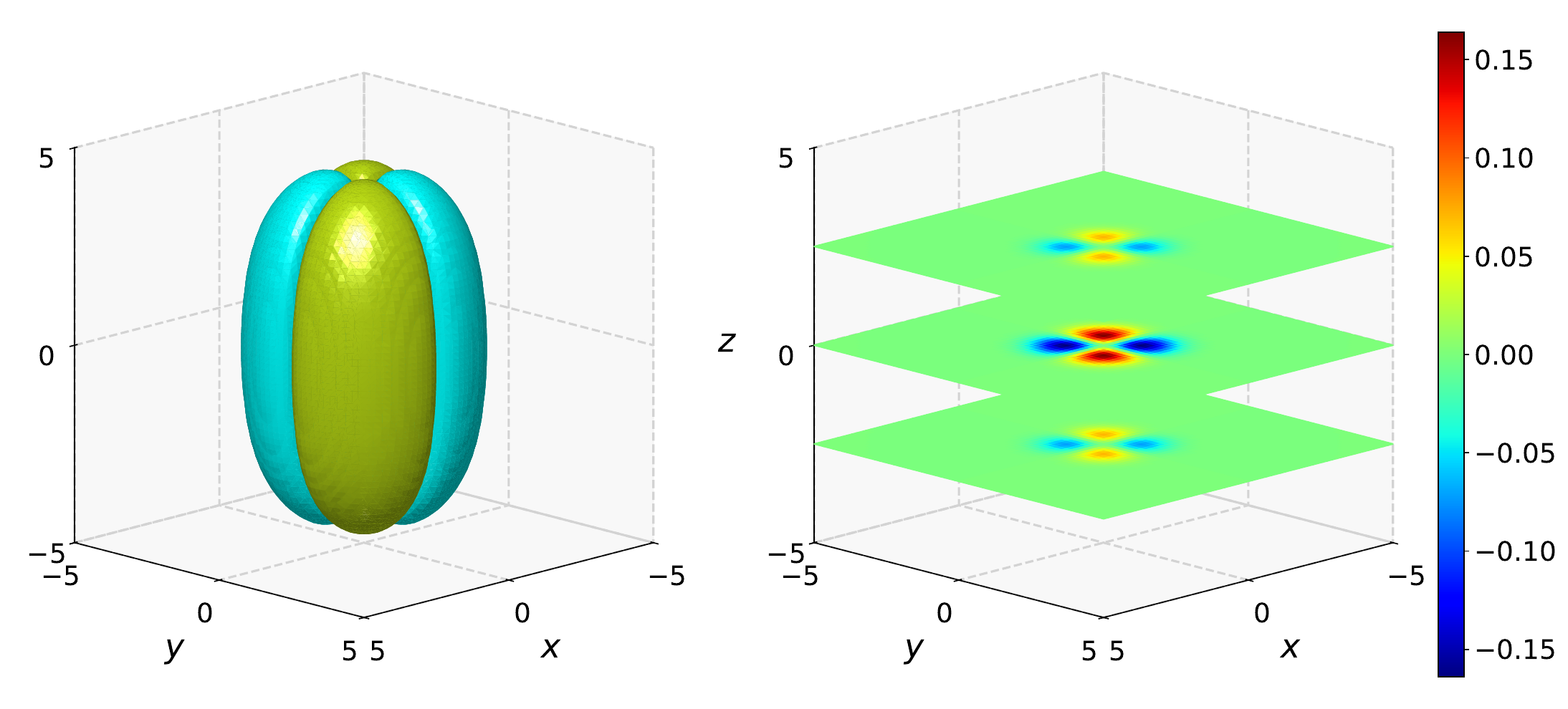}\\
		\includegraphics[width=0.6\textwidth]{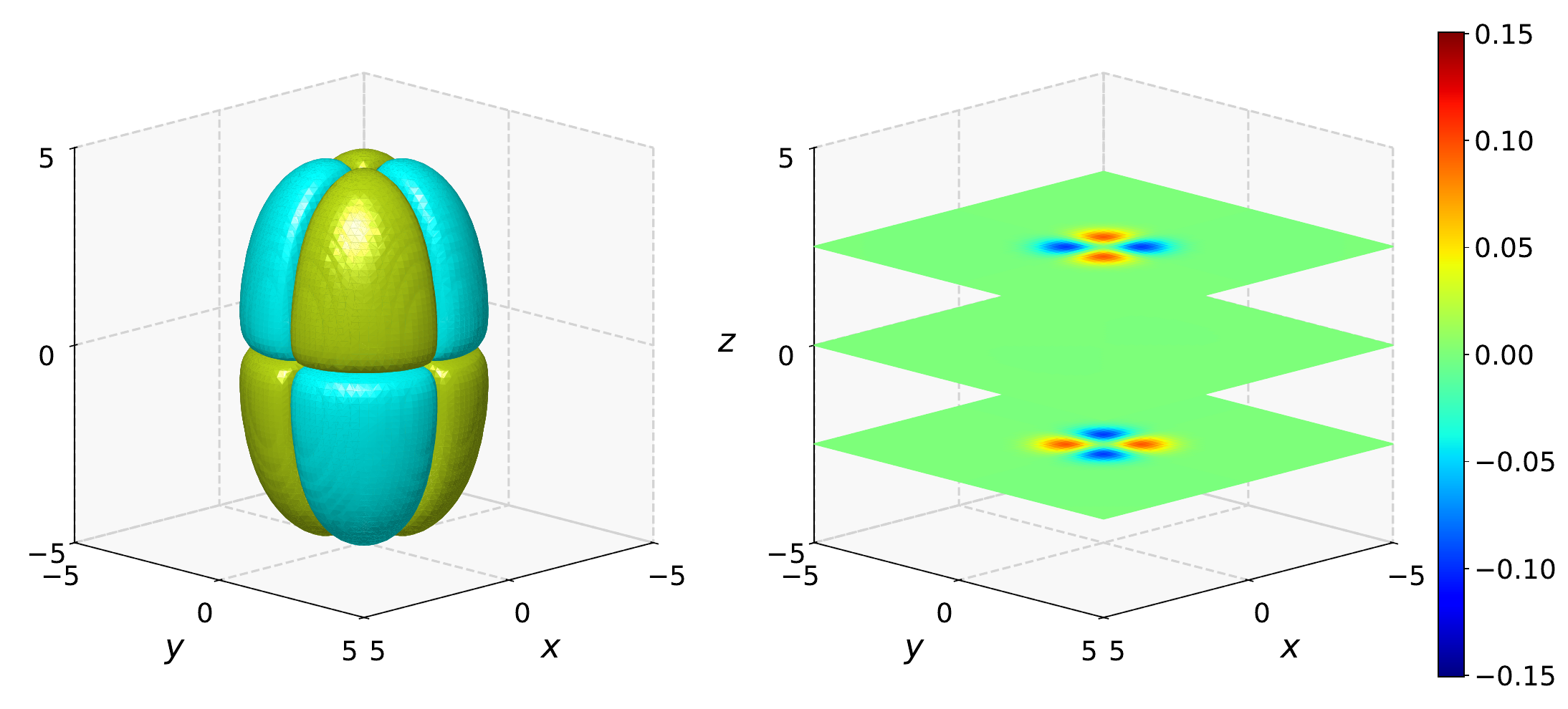}
		\caption{Isosurface plots (left column) and corresponding slice views (right column) for the wave functions of 3D excited states $\phi_{100}$ (top), $\phi_{110}$ (middle), and $\phi_{111}$ (bottom) with \(\beta_3=207.16\) and \(\lambda_3=207.16\). In the left column, the yellow-green and cyan surfaces represent the positive (\(\phi(x,y,z)=0.01\)) and negative (\(\phi(x,y,z)=-0.01\)) isosurfaces, respectively.}
		\label{fig:phi_states_3d}
	\end{figure}
Table~\ref{tab:combined_data} reveals the following phenomena. First, for any fixed ratio $\lambda_3/\beta_3$, the energy increases with the index. Due to the anisotropic trapping potential ($\gamma_z < \gamma_x=\gamma_y$), the confinement along the $z$-direction is weaker; therefore, exciting the condensate along the $z$-direction costs less energy than along the tightly confined $x$- or $y$-directions, i.e., $E(\phi_{001}) < E(\phi_{100})$. Second, as $\lambda_3/\beta_3$ increases from $-0.50$ to $1.00$, both the energy and the chemical potential of all states decrease monotonically, indicating that the enhancement of the attractive component of the DDI lowers the total energy of the excited states.

\section{Conclusion}\label{conclusion}
In this paper, an efficient numerical method is proposed for computing excited states of dipolar BECs. The method incorporates a CGAD to stably locate constrained saddle points with prescribed Morse indices and a predictor-corrector semi-implicit sine-pseudospectral discretization in which the NUFFT is used to evaluate the DDIs. Extensive numerical results reveal interesting physical phenomena. Notably, in parameter regimes where the ground state does not exist and the energy is unbounded below, certain excited states can still be computed stably by the proposed numerical method. This finding provides direct numerical evidence that may guide future theoretical analysis and experimental detection of excited states in dipolar BECs. In addition, acceleration techniques, such as preconditioning and variable-metric gradient methods, are expected to further improve the iterative efficiency of the proposed method. Related investigations are currently in progress.

\section*{Acknowledgements}
This work was supported by  the National Natural Science Foundation of China (Grant Nos. 12571448 and 12371374), the Innovation Research Foundation of NUDT (No. 202402-YJRC-XX-002), and the Science and Technology Innovation Program of Hunan Province (No. 2026RC1041).

\appendix

\section{Proof of Lemma~\ref{lem:Phi-dip-boundedness}}\label{appendix:proof-lem:Phi-dip-boundedness}
\begin{proof}
The proof for \(d=3\) is referred to Lemma~2.1 in \cite{CMS2008Nonlinearity}. It suffices to consider the cases of $d=1$ and $d=2$.

For \(d=1\), the 1D dipole kernel \eqref{eq:U_dip_lower_d} can be rewritten as
\[
U_{\mathrm{dip}}^{(1)}(x)
=
(3n_3^2-1)\delta(x)
-\frac{3(3n_3^2-1)}{2}
\int_0^\infty e^{-s}K_s(x)\,ds,
\]
where \(K_s(x)=\frac{\sqrt{2s}}{2\varepsilon}
e^{-\frac{\sqrt{2s}|x|}{\varepsilon}}\), \(s>0\). Since \(\|K_s\|_{L^1}=1\), by the  Young's convolution inequality, we obtain
\begin{align*}
\|U_{\mathrm{dip}}^{(1)}*f\|_p
\le |3n_3^2-1|\|f\|_p
+\frac{3|3n_3^2-1|}{2}
\int_0^\infty e^{-s}\|K_s\|_{L^1}\|f\|_p\,ds
=\frac52|3n_3^2-1|\|f\|_p.
\end{align*}

For \(d=2\), according to Lemma~3.2 in \cite{bao2012gross}, letting \(\varphi(\mathbf{x}) = \left(\frac{1}{2\pi|\mathbf{x}|}\right) * f, \; \mathbf{x}\in\mathbb{R}^2\), there exists a constant \(C_p > 0\) satisfying the second-order derivative estimate
		\begin{equation} \label{lemma_2D differential operator}
			\|\partial_{\alpha\beta} \varphi\|_p \le C_p \|\partial_\beta f\|_p,\quad \forall p \in (1,\infty),\; \alpha, \beta \in \{x, y\}.
		\end{equation}
		By \eqref{lemma_2D differential operator}, 
        we have
        \begin{equation*}
			\begin{aligned}
        \left\|\left( \partial_{\mathbf{n}_\perp\mathbf{n}_\perp} - n_3^2\Delta \right) \varphi \right\|_p 
                &= \left\|(n_1^2 - n_3^2)\partial_{xx}\varphi + 2n_1 n_2 \partial_{xy}\varphi + (n_2^2 - n_3^2)\partial_{yy}\varphi\right\|_p\\
				&\le |n_1^2 - n_3^2| \|\partial_{xx}\varphi\|_p + 2|n_1 n_2| \|\partial_{xy}\varphi\|_p + |n_2^2 - n_3^2| \|\partial_{yy}\varphi\|_p \\
				&\le C_p \Big( |n_1^2 - n_3^2| \|\partial_x f\|_p + 2|n_1 n_2| \|\partial_y f\|_p + |n_2^2 - n_3^2| \|\partial_y f\|_p \Big) \\
				&\le C_p \Big( |n_1^2 - n_3^2| + 2|n_1 n_2| + |n_2^2 - n_3^2| \Big) \|\nabla f\|_p
                =: \tilde{C}_p\|\nabla f\|_p,
			\end{aligned}
		\end{equation*}
		so that
		\begin{align*} 
				\| U_{\mathrm{dip}}^{(2)} * f\|_p 
                &\le |3n_3^2 - 1| \|f\|_p + \frac{3\sqrt{2\pi}\varepsilon}{2} 
                \left\|\left( \partial_{\mathbf{n}_\perp\mathbf{n}_\perp} - n_3^2\Delta \right) \varphi \right\|_p 
                 \le C_1 \|f\|_p + C_2 \|\nabla f\|_p.
		\end{align*}
		The proof is completed.
\end{proof}

\footnotesize
\bibliographystyle{abbrv}
\bibliography{ckwx1} 

@article{IMF2015,
author = {Gao, Weiguo and Leng, Jing and Zhou, Xiang},
title = {An Iterative Minimization Formulation for Saddle Point Search},
journal = {SIAM J. Numer. Anal.},
volume = {53},
number = {4},
pages = {1786-1805},
year = {2015},
doi = {10.1137/130930339}
}

@book{Chang1993,
  author   = {Chang, Kung-Ching},
  title    = {Infinite Dimensional Morse Theory and Multiple Solution Problems},
  publisher= {Birkh{\"a}user},
  address  = {Boston},
  year     = {1993}
}

@article{CMS2008Nonlinearity,
  author   = {Carles, R{\'e}mi and Markowich, Peter A. and Sparber, Christof},
  title    = {On the {Gross--Pitaevskii} equation for trapped dipolar quantum gases},
  journal  = {Nonlinearity},
  volume   = {21},
  number   = {11},
  pages    = {2569--2590},
  year     = {2008}
}

@article{bao2013mathematical,
  author   = {Bao, Weizhu and Cai, Yongyong},
  title    = {Mathematical theory and numerical methods for {Bose--Einstein} condensation},
  journal  = {Kinet. Relat. Models},
  volume   = {6},
  number   = {1},
  pages    = {1--135},
  year     = {2013}
}

@article{bradley1995evidence,
  author   = {Bradley, C. C. and Sackett, C. A. and Tollett, J. J. and Hulet, R. G.},
  title    = {Evidence of {Bose--Einstein} condensation in an atomic gas with attractive interactions},
  journal  = {Phys. Rev. Lett.},
  volume   = {75},
  number   = {9},
  pages    = {1687--1690},
  year     = {1995}
}

@article{koch2008stabilization,
  author   = {Koch, Tobias and Lahaye, Thierry and Metz, Jonas and Fr{\"o}hlich, Bernd and Griesmaier, Axel and Pfau, Tilman},
  title    = {Stabilization of a purely dipolar quantum gas against collapse},
  journal  = {Nat. Phys.},
  volume   = {4},
  number   = {3},
  pages    = {218--222},
  year     = {2008}
}

@article{bao2012gross,
  author   = {Bao, Weizhu and Ben Abdallah, Naoufel and Cai, Yongyong},
  title    = {{Gross--Pitaevskii--Poisson} equations for dipolar {Bose--Einstein} condensates with anisotropic confinement},
  journal  = {SIAM J. Math. Anal.},
  volume   = {44},
  number   = {3},
  pages    = {1713--1741},
  year     = {2012}
}

@article{bao2016accurate,
  author   = {Bao, Weizhu and Tang, Qinglin and Zhang, Yong},
  title    = {Accurate and efficient numerical methods for computing ground states and dynamics of dipolar {Bose--Einstein} condensates via the nonuniform {FFT}},
  journal  = {Commun. Comput. Phys.},
  volume   = {19},
  number   = {5},
  pages    = {1141--1166},
  year     = {2016}
}

@article{weinstein1982nonlinear,
  author   = {Weinstein, Michael I.},
  title    = {Nonlinear {Schr{\"o}dinger} equations and sharp interpolation estimates},
  journal  = {Commun. Math. Phys.},
  volume   = {87},
  number   = {4},
  pages    = {567--576},
  year     = {1983}
}

@article{jiang2014fast,
  author   = {Jiang, Shidong and Greengard, Leslie and Bao, Weizhu},
  title    = {Fast and accurate evaluation of nonlocal {Coulomb} and dipole-dipole interactions via the nonuniform {FFT}},
  journal  = {SIAM J. Sci. Comput.},
  volume   = {36},
  number   = {5},
  pages    = {B777--B794},
  year     = {2014}
}

@article{liu2023constrained,
  author   = {Liu, Wei and Xie, Ziqing and Yuan, Yongjun},
  title    = {A constrained gentlest ascent dynamics and its applications to finding excited states of {Bose--Einstein} condensates},
  journal  = {J. Comput. Phys.},
  volume   = {473},
  pages    = {111719},
  year     = {2023}
}

@article{bao2010efficient,
  author   = {Bao, Weizhu and Cai, Yongyong and Wang, Hanquan},
  title    = {Efficient numerical methods for computing ground states and dynamics of dipolar {Bose--Einstein} condensates},
  journal  = {J. Comput. Phys.},
  volume   = {229},
  number   = {20},
  pages    = {7874--7892},
  year     = {2010}
}

@article{koberle2009bifurcations,
  author   = {K{\"o}berle, Patrick and Cartarius, Holger and Fab{\v{c}}i{\v{c}}, Toma{\v{z}} and Main, J{\"o}rg and Wunner, G{\"u}nter},
  title    = {Bifurcations, order and chaos in the {Bose--Einstein} condensation of dipolar gases},
  journal  = {New J. Phys.},
  volume   = {11},
  number   = {2},
  pages    = {023017},
  year     = {2009}
}

@article{lahaye2009physics,
  author   = {Lahaye, Thierry and Menotti, Cecilia and Santos, Luis and Lewenstein, Maciej and Pfau, Tilman},
  title    = {The physics of dipolar bosonic quantum gases},
  journal  = {Rep. Prog. Phys.},
  volume   = {72},
  number   = {12},
  pages    = {126401},
  year     = {2009}
}

@article{anderson1995observation,
  author   = {Anderson, Mike H. and Ensher, Jason R. and Matthews, Michael R. and Wieman, Carl E. and Cornell, Eric A.},
  title    = {Observation of {Bose--Einstein} condensation in a dilute atomic vapor},
  journal  = {Science},
  volume   = {269},
  number   = {5221},
  pages    = {198--201},
  year     = {1995}
}

@article{davis1995bose,
  author   = {Davis, Kendall B. and Mewes, M.-O. and Andrews, Michael R. and van Druten, Nicolaas J. and Durfee, Dallin S. and Kurn, D. M. and Ketterle, Wolfgang},
  title    = {{Bose--Einstein} condensation in a gas of sodium atoms},
  journal  = {Phys. Rev. Lett.},
  volume   = {75},
  number   = {22},
  pages    = {3969--3973},
  year     = {1995}
}

@article{griesmaier2007generation,
  author   = {Griesmaier, Axel},
  title    = {Generation of a dipolar {Bose--Einstein} condensate},
  journal  = {J. Phys. B: At. Mol. Opt. Phys.},
  volume   = {40},
  number   = {14},
  pages    = {R91--R134},
  year     = {2007}
}

@article{aikawa2012bose,
  author   = {Aikawa, K. and Frisch, A. and Mark, M. and Baier, S. and Rietzler, A. and Grimm, R. and Ferlaino, F.},
  title    = {{Bose--Einstein} condensation of erbium},
  journal  = {Phys. Rev. Lett.},
  volume   = {108},
  number   = {21},
  pages    = {210401},
  year     = {2012}
}

@article{weinan2011gentlest,
  author   = {E, Weinan and Zhou, Xiang},
  title    = {The gentlest ascent dynamics},
  journal  = {Nonlinearity},
  volume   = {24},
  number   = {6},
  pages    = {1831--1842},
  year     = {2011}
}

@article{griesmaier2005bose,
  author   = {Griesmaier, Axel and Werner, J{\"o}rg and Hensler, Sven and Stuhler, J{\"u}rgen and Pfau, Tilman},
  title    = {{Bose--Einstein} condensation of chromium},
  journal  = {Phys. Rev. Lett.},
  volume   = {94},
  number   = {16},
  pages    = {160401},
  year     = {2005}
}

@article{Lu_2011,
  author   = {Lu, Mingwu and Burdick, Nathaniel Q. and Youn, Seo Ho and Lev, Benjamin L.},
  title    = {Strongly dipolar {Bose--Einstein} condensate of dysprosium},
  journal  = {Phys. Rev. Lett.},
  volume   = {107},
  number   = {19},
  pages    = {190401},
  year     = {2011}
}

@article{yi2001trapped,
  author   = {Yi, S. and You, L.},
  title    = {Trapped condensates of atoms with dipole interactions},
  journal  = {Phys. Rev. A},
  volume   = {63},
  number   = {5},
  pages    = {053607},
  year     = {2001}
}

@article{mottonen2005collapse,
  author   = {M{\"o}tt{\"o}nen, M. and Virtanen, S. M. M. and Salomaa, M. M.},
  title    = {Collapse and revival of excitations in {Bose--Einstein} condensates},
  journal  = {Phys. Rev. A},
  volume   = {71},
  number   = {2},
  pages    = {023604},
  year     = {2005}
}

@article{ferrofluid,
  author   = {Lahaye, Thierry and Koch, Tobias and Fr{\"o}hlich, Bernd and Fattori, Marco and Metz, Jonas and Griesmaier, Axel and Giovanazzi, Stefano and Pfau, Tilman},
  title    = {Strong dipolar effects in a quantum ferrofluid},
  journal  = {Nature},
  volume   = {448},
  number   = {7154},
  pages    = {672--675},
  year     = {2007}
}

@article{Anisotropicexpansion,
  author   = {Stuhler, J. and Griesmaier, A. and Koch, T. and Fattori, M. and Pfau, T. and Giovanazzi, S. and Pedri, P. and Santos, L.},
  title    = {Observation of dipole-dipole interaction in a degenerate quantum gas},
  journal  = {Phys. Rev. Lett.},
  volume   = {95},
  number   = {15},
  pages    = {150406},
  year     = {2005}
}

@article{roton,
  author   = {Petter, D. and Natale, G. and van Bijnen, R. M. W. and Patscheider, A. and Mark, M. J. and Chomaz, L. and Ferlaino, F.},
  title    = {Probing the roton excitation spectrum of a stable dipolar {Bose} gas},
  journal  = {Phys. Rev. Lett.},
  volume   = {122},
  number   = {18},
  pages    = {183401},
  year     = {2019}
}

@article{experiment,
  author   = {Chomaz, Lauriane and Ferrier-Barbut, Igor and Ferlaino, Francesca and Laburthe-Tolra, Bruno and Lev, Benjamin L. and Pfau, Tilman},
  title    = {Dipolar physics: a review of experiments with magnetic quantum gases},
  journal  = {Rep. Prog. Phys.},
  volume   = {86},
  number   = {2},
  pages    = {026401},
  year     = {2023}
}

@article{trapped,
  author   = {Dalfovo, Franco and Giorgini, Stefano and Pitaevskii, Lev P. and Stringari, Sandro},
  title    = {Theory of {Bose--Einstein} condensation in trapped gases},
  journal  = {Rev. Mod. Phys.},
  volume   = {71},
  number   = {3},
  pages    = {463--512},
  year     = {1999}
}

@article{bao2006efficient,
  author   = {Bao, Weizhu and Chern, I-Liang and Lim, Fong Yin},
  title    = {Efficient and spectrally accurate numerical methods for computing ground and first excited states in {Bose--Einstein} condensates},
  journal  = {J. Comput. Phys.},
  volume   = {219},
  number   = {2},
  pages    = {836--854},
  year     = {2006}
}

@article{CNGF,
  author   = {Bao, Weizhu and Du, Qiang},
  title    = {Computing the ground state solution of {Bose--Einstein} condensates by a normalized gradient flow},
  journal  = {SIAM J. Sci. Comput.},
  volume   = {25},
  number   = {5},
  pages    = {1674--1697},
  year     = {2004}
}

@article{3,
  author   = {Chiofalo, M. L. and Succi, S. and Tosi, M. P.},
  title    = {Ground state of trapped interacting {Bose--Einstein} condensates by an explicit imaginary-time algorithm},
  journal  = {Phys. Rev. E},
  volume   = {62},
  number   = {5},
  pages    = {7438--7444},
  year     = {2000}
}

@article{sobolev,
  author   = {Chen, Ziang and Lu, Jianfeng and Lu, Yulong and Zhang, Xiangxiong},
  title    = {On the convergence of {Sobolev} gradient flow for the {Gross--Pitaevskii} eigenvalue problem},
  journal  = {SIAM J. Numer. Anal.},
  volume   = {62},
  number   = {2},
  pages    = {667--691},
  year     = {2024}
}

@article{sobolev2,
  author   = {Danaila, Ionut and Kazemi, Parimah},
  title    = {A new {Sobolev} gradient method for direct minimization of the {Gross--Pitaevskii} energy with rotation},
  journal  = {SIAM J. Sci. Comput.},
  volume   = {32},
  number   = {5},
  pages    = {2447--2467},
  year     = {2010}
}

@article{sobolev3,
  author   = {Henning, Patrick and Peterseim, Daniel},
  title    = {{Sobolev} gradient flow for the {Gross--Pitaevskii} eigenvalue problem: global convergence and computational efficiency},
  journal  = {SIAM J. Numer. Anal.},
  volume   = {58},
  number   = {3},
  pages    = {1744--1772},
  year     = {2020}
}

@article{Riemannian,
  author   = {Henning, Patrick and Yadav, Mahima},
  title    = {Convergence of a {Riemannian} gradient method for the {Gross--Pitaevskii} energy functional in a rotating frame},
  journal  = {ESAIM Math. Model. Numer. Anal.},
  volume   = {59},
  number   = {2},
  pages    = {1145--1175},
  year     = {2025}
}

@article{GFLM,
  author   = {Liu, Wei and Cai, Yongyong},
  title    = {Normalized gradient flow with {Lagrange} multiplier for computing ground states of {Bose--Einstein} condensates},
  journal  = {SIAM J. Sci. Comput.},
  volume   = {43},
  number   = {1},
  pages    = {B219--B242},
  year     = {2021}
}

@article{2017Tang,
  author   = {Antoine, Xavier and Levitt, Antoine and Tang, Qinglin},
  title    = {Efficient spectral computation of the stationary states of rotating {Bose--Einstein} condensates by preconditioned nonlinear conjugate gradient methods},
  journal  = {J. Comput. Phys.},
  volume   = {343},
  pages    = {92--109},
  year     = {2017}
}

@article{min,
  author   = {Bao, Weizhu and Tang, Weijun},
  title    = {Ground-state solution of {Bose--Einstein} condensate by directly minimizing the energy functional},
  journal  = {J. Comput. Phys.},
  volume   = {187},
  number   = {1},
  pages    = {230--254},
  year     = {2003}
}

@article{preconditioned,
  author   = {Shu, Qingzhou and Tang, Qinglin and Zhang, Shaobo and Zhang, Yong},
  title    = {A preconditioned {Riemannian} conjugate gradient method for computing the ground states of arbitrary-angle rotating {Bose--Einstein} condensates},
  journal  = {J. Comput. Phys.},
  volume   = {512},
  pages    = {113130},
  year     = {2024}
}

@article{Regularized,
  author   = {Wu, Xinming and Wen, Zaiwen and Bao, Weizhu},
  title    = {A regularized {Newton} method for computing ground states of {Bose--Einstein} condensates},
  journal  = {J. Sci. Comput.},
  volume   = {73},
  number   = {1},
  pages    = {303--329},
  year     = {2017}
}

@article{weak,
  author   = {Yang, Lin and Li, Xiang-Gui and Yan, Wei and Zhang, Ran},
  title    = {The energy-diminishing weak {Galerkin} finite element method for the computation of ground state and excited states in {Bose--Einstein} condensates},
  journal  = {J. Comput. Phys.},
  volume   = {520},
  pages    = {113497},
  year     = {2025}
}

@article{Revealing,
  author   = {Yin, Jianyuan and Huang, Zhen and Cai, Yongyong and Du, Qiang and Zhang, Lei},
  title    = {Revealing excited states of rotational {Bose--Einstein} condensates},
  journal  = {The Innovation},
  volume   = {5},
  number   = {1},
  pages    = {100546},
  year     = {2024}
}

@article{Boullé2020Deflation,
  title   = {Deflation-based identification of nonlinear excitations of the three-dimensional {Gross-Pitaevskii} equation},
  author  = {Boullé, N. and Charalampidis, E. G. and Farrell, P. E. and Kevrekidis, P. G.},
  journal = {Phys. Rev. A},
  volume  = {101},
  number  = {5},
  pages   = {053307},
  year    = {2020},
}

@article{schneider1999numerical,
  title   = {Numerical approach to the ground and excited states of a {Bose-Einstein} condensed gas confined in a completely anisotropic trap},
  author  = {Schneider, B. I. and Feder, D. L.},
  journal = {Phys. Rev. A},
  volume  = {59},
  number  = {3},
  pages   = {2232--2242},
  year    = {1999},
}

@article{Chien2007Adaptive,
  author   = {Chang, S.-L. and Chien, C.-S.},
  title    = {Adaptive continuation algorithms for computing energy levels of rotating {Bose--Einstein} condensates},
  journal  = {Comput. Phys. Commun.},
  volume   = {177},
  number   = {9},
  pages    = {707--719},
  year     = {2007}
}

@article{WeichungWang2011Exploring,
  author   = {Chen, Jen-Hao and Chern, I-Liang and Wang, Weichung},
  title    = {Exploring ground states and excited states of spin-1 {Bose--Einstein} condensates by continuation methods},
  journal  = {J. Comput. Phys.},
  volume   = {230},
  number   = {6},
  pages    = {2222--2236},
  year     = {2011}
}

@article{eigenfunctions,
  author   = {Marojevi{\'c}, {\v{Z}}elimir and G{\"o}kl{\"u}, Ertan and L{\"a}mmerzahl, Claus},
  title    = {Energy eigenfunctions of the {1D} {Gross--Pitaevskii} equation},
  journal  = {Comput. Phys. Commun.},
  volume   = {184},
  number   = {8},
  pages    = {1920--1930},
  year     = {2013}
}

@article{eignvector,
  author   = {Cerjan, Charles J. and Miller, William H.},
  title    = {On finding transition states},
  journal  = {J. Chem. Phys.},
  volume   = {75},
  number   = {6},
  pages    = {2800--2806},
  year     = {1981}
}

@article{hisd,
  author  = {Su, Hua and Wang, Haoran and Zhang, Lei and Zhao, Jin and Zheng, Xiangcheng},
  title   = {Improved high-index saddle dynamics for finding saddle points and solution landscape},
  journal = {SIAM J. Numer. Anal.},
  volume  = {63},
  number  = {4},
  pages   = {1757--1775},
  year    = {2025}
}

@article{CSDD,
  author  = {Zhang, Jingyan and Du, Qiang},
  title   = {Constrained shrinking dimer dynamics for saddle point search with constraints},
  journal = {J. Comput. Phys.},
  volume  = {231},
  number  = {14},
  pages   = {4745--4758},
  year    = {2012}
}

@article{chisd,
  author   = {Yin, Jianyuan and Huang, Zhen and Zhang, Lei},
  title    = {Constrained high-index saddle dynamics for the solution landscape with equality constraints},
  journal  = {J. Sci. Comput.},
  volume   = {91},
  pages    = {62},
  year     = {2022}
}

@article{string,
  author   = {Du, Qiang and Zhang, Lei},
  title    = {A constrained string method and its numerical analysis},
  journal  = {Commun. Math. Sci.},
  volume   = {7},
  number   = {4},
  pages    = {1039--1051},
  year     = {2009}
}

@article{skym,
  author   = {Bessarab, Pavel F. and Uzdin, Valery M. and J{\'o}nsson, Hannes},
  title    = {Method for finding mechanism and activation energy of magnetic transitions, applied to skyrmion and antivortex annihilation},
  journal  = {Comput. Phys. Commun.},
  volume   = {196},
  pages    = {335--347},
  year     = {2015}
}

@article{SDD,
  author   = {Zhang, Jingyan and Du, Qiang},
  title    = {Shrinking dimer dynamics and its applications to saddle point search},
  journal  = {SIAM J. Numer. Anal.},
  volume   = {50},
  number   = {4},
  pages    = {1899--1921},
  year     = {2012}
}

@article{li2001minimax,
  author   = {Li, Yongxin and Zhou, Jianxin},
  title    = {A minimax method for finding multiple critical points and its applications to semilinear {PDEs}},
  journal  = {SIAM J. Sci. Comput.},
  volume   = {23},
  number   = {3},
  pages    = {840--865},
  year     = {2001}
}

@article{dimer,
  author   = {Henkelman, Graeme and J{\'o}nsson, Hannes},
  title    = {A dimer method for finding saddle points on high-dimensional potential surfaces using only first derivatives},
  journal  = {J. Chem. Phys.},
  volume   = {111},
  number   = {15},
  pages    = {7010--7022},
  year     = {1999}
}
	
\end{document}